\documentclass[12pt, a4page]{article}
\usepackage[utf8]{inputenc}
\usepackage[T1]{fontenc}
\usepackage{natbib}
\setcitestyle{authoryear,open={(},close={)}}

\usepackage{tikz}
\usetikzlibrary{patterns}
\usepackage{graphics}
\usepackage{graphicx}
\usepackage{amsmath}
\usepackage{amsthm}
\usepackage{amsfonts}
\usepackage{geometry}
\usepackage[normalem]{ulem}

\usepackage{hyperref}
\hypersetup{colorlinks=true,linkcolor=blue,citecolor=blue}
\usepackage{dsfont} 
\usepackage{nicefrac}

\date{}

\newtheorem{theorem}{Theorem}
\newtheorem*{theorem*}{Theorem}
\newtheorem{lemma}{Lemma}
\newtheorem{claim}{Claim}
\newtheorem{corollary}{Corollary}

\newtheorem{proposition}{Proposition}

\newtheorem{definition}{Definition}

\theoremstyle{definition}

\theoremstyle{remark}

\newtheorem{example}{Example}

\newcommand{\bP}{\mathbb{P}}

\newcommand{\R}{\mathbb{R}}

\newcommand{\F}{\mathcal{F}}

\newcommand{\var}{\mathrm{Var}}
\newcommand{\one}{\mathds{1}}
\providecommand{\tabularnewline}{\\}

\def\dd{\mathrm{d}}

\def\cA{\mathcal{A}}
\def\cB{\mathcal{B}}

\usepackage{xparse}
\usepackage{mleftright}

\DeclareDocumentCommand\Pr{ m g }{\ensuremath{
    {   \IfNoValueTF {#2}
      {\mathbb{P}\mleft[{#1}\mright]}
      {\mathbb{P}\mleft[{#1}\,\middle\vert\,{#2}\mright]}
    }
}}
\DeclareDocumentCommand\E{ m g }{\ensuremath{
    {   \IfNoValueTF {#2}
      {\mathbb{E}\mleft[{#1}\mright]}
      {\mathbb{E}\mleft[{#1}\,\middle\vert\,{#2}\mright]}
    }
}}

\title{Private Private Information\footnote{
The paper benefited from numerous suggestions and comments from our colleagues. We are grateful to  Marina Agranov, Nageeb Ali, Itai Arieli, Eric Auerbach, Yakov Babichenko, Dirk Bergemann, Alexander Bloedel, Aislinn Bohren, Simina Br\^{a}nzei, Benjamin Brooks, David Dillenberger, Piotr Dworczak, Federico Echenique, Jeffrey Ely, Drew Fudenberg,  Hanming Fang, Wayne Gao, Kai Hao Yang, David Kempe, Toygar Kerman, Andreas Kleiner, Victoria Kostina, Annie Liang, Jonathan Libgober, Elliot Lipnowski, Alessandro Lizzeri, George Mailath, Moritz Meyer-ter-Vehn, Xiaosheng Mu, Pietro Ortoleva,
Luciano Pomatto, Andy Postlewaite, Doron Ravid, Ran Shorrer, Ran Spiegler, Bruno Strulovici, Ina Taneva, Caroline Thomas, Tristan Tomala, Nicolas Vieille, Alexander Wolitzky, Leeat Yariv, and seminar participants
at Caltech, University College London, UC Riverside, USC, Hong Kong Joint Theory Seminar, HSE Moscow and St.\ Petersburg, University of Maryland, Northwestern University, University of Pennsylvania,  Purdue University,  Stanford, 2021 LA Theory Conference, Pennsylvania Economic Theory Conference, Retreat on Information, Networks, and Social Economics,  and Stony Brook Workshop on Strategic Communication and Learning.}}

\author{ \large Kevin He\thanks{University of Pennsylvania. Email: hesichao@gmail.com.} \ \ \ 
Fedor Sandomirskiy\thanks{Princeton University. Email:  fsandomi@princeton.edu. Fedor Sandomirskiy was supported by the Linde Institute at Caltech and the National Science Foundation (grant  CNS 1518941).} \ \ \ 
Omer Tamuz\thanks{Caltech. Email: tamuz@caltech.edu.  Omer Tamuz was supported by a grant from the Simons Foundation (\#419427), a Sloan fellowship, a BSF award (\#2018397), and a National Science Foundation CAREER award (DMS-1944153).}}

\begin{document}
\setcounter{page}{0}
\maketitle
\thispagestyle{empty}
\begin{abstract}

    {Private signals model noisy information about an unknown state. 
    Although these signals are called ``private,'' they may still carry information about each other. Our paper introduces the concept of private private signals, which contain information about the state but not about other signals. To achieve privacy, signal quality may need to be sacrificed. We study the informativeness of private private signals and characterize those that are optimal in the sense that they cannot be made more informative without violating privacy. We discuss implications for privacy in recommendation systems, information design, causal inference, and mechanism design.}
\end{abstract}
    \vspace{2bp}

\newpage

\section{Introduction}
\label{sec:intro}

Privacy has emerged as an important {concern} in contemporary markets, characterized by the pervasive collection and exchange of information. In consequence, privacy has been the focus of a large literature in computer science \citep*{dwork2014algorithmic} and a more recent literature in economics \citep*{eilat2021bayesian, acemoglu2022too, bergemann2022economics,liang2020data}. {In this paper, we introduce the perspective and tools of Blackwell informativeness to the study of privacy.}

As a motivating example, consider an  employer who is interested in an unknown state: whether or not an applicant is qualified for a job. A recommender (e.g., an automated recommendation system or a human letter writer)  knows the applicant's qualification but also knows a sensitive attribute, such as the applicant's age. The recommender wishes to give as much information as possible to the employer about the applicant's qualification, but does not want to reveal any information about the sensitive attribute. 

If the state and the attribute are independent, then the recommender can simply report the state. However, if the state and the attribute are correlated,  reporting the state also inadvertently reveals some  information about the attribute. The recommender faces a privacy-constrained information-design problem: how to optimally give information about the state, while keeping the recommendation independent of the sensitive attribute? Intuitively, as the correlation between the sensitive attribute and the state increases, the privacy constraint becomes more restrictive and further limits how much information the recommender can provide.  

{Generalizing this example, we view the attribute and the recommendation as two signals about the state. We  study \emph{private private information structures}, in  which signals about the state are statistically independent of each other. Requiring independence between these signals imposes a joint restriction on their informativeness. This paper studies  the maximal informativeness of signals independent of each other and applies the results to settings such as the recommender's problem above. }

 {Specifically, we are interested in the maximal amount of information that can be conveyed through private private signals.} We formalize this question using the notion of Pareto optimality with respect to the Blackwell order: a private private information structure is \emph{Blackwell-Pareto optimal} if it is impossible to {make any of the signals more informative} in the Blackwell sense without violating privacy.\footnote{The term Pareto optimality is usually used in the context of allocating goods to agents who have preferences over them. {One can similarly 
 think of an information structure as an allocation of information to agents who prefer signals that are more Blackwell informative.}} The main goal of this paper is to understand the Blackwell-Pareto frontier, allowing a designer to produce optimal private private structures. This includes, for example, producing the most informative recommendation that does not  reveal information about a sensitive attribute.

\paragraph{{Main} Results.} 

In the case of two {signals} and a binary state, we obtain a simple description of the Blackwell-Pareto frontier: given a private private structure, denote by $F_1, F_2 \colon [0,1] \to [0,1]$ the cumulative distribution functions of the {induced beliefs about the state.} 
Then the structure is Blackwell-Pareto optimal if and only if $F_1(x) = 1-F_2^{-1}(1-x)$, {where $F_2^{-1}$ is the inverse of $F_2$.}
This characterization allows for a straightforward test of optimality and a constructive procedure for finding  optimal structures.

{In the general case with any number of signals and states, 
we uncover a surprising connection} 
to the field of mathematical tomography and the study of sets of uniqueness. A subset of $[0,1]^n$ is called a \emph{set of uniqueness} if it is uniquely determined by its projections onto the coordinate axes. Understanding such sets has been an active area of research since the 1940s \citep*{lorentz1949problem}. This problem gained more prominence with the advent of tomography, a technology to reconstruct three-dimensional objects from their projections \citep*[see, e.g.,][]{gardner1995geometric}. 

{In the case of a binary state, we show that private private information structures with $n$ signals can be identified with subsets of $[0,1]^n$, and that the Blackwell-Pareto optimal ones correspond exactly to sets of uniqueness. In the two-dimensional case---which corresponds to the case of two signals---the complete characterization of the sets of uniqueness is known \citep*{lorentz1949problem} and leads to the above-mentioned characterization of the Blackwell-Pareto frontier.
With three or more states and any number of signals, we establish an equivalence 
between Blackwell-Pareto optimality and a generalization of sets of uniqueness that we term \emph{partitions of uniqueness}}.

\paragraph{{Applications.}} 

{We consider a number of applications of our main results. The main application is to} \emph{privacy-preserving recommendations}, a particular case of which is the motivating example presented above.

In the problem of privacy-preserving recommendation, the goal is to  design a maximally informative signal (a recommendation) about the state under the constraint that this recommendation is independent of a given random variable correlated with the state. For example, the state can be a job applicant's productivity type or a loan applicant's creditworthiness, and a recommendation must be as informative about the state as possible but carry no information about sensitive attributes such as the applicant's age or health condition.
{The state, together with the sensitive attribute and the designed recommendation, form a private private structure, thus connecting the problem of optimal privacy-preserving recommendation to Blackwell-Pareto optimality of private private structures.}

When the state is binary (i.e., the applicant's productivity type takes {two possible} values), our results on Blackwell-Pareto optimality provide a complete solution to the problem of  privacy-preserving recommendation.  {As we show, there is a dominant recommendation: a privacy-preserving} recommendation that  Blackwell dominates any other.
{The result} is constructive and gives a simple recipe for generating this recommendation. We also quantify the information loss due to the privacy constraint and discuss its comparative statics. 

Curiously, the {existence of a dominant} recommendation for binary states 
means that knowledge of the decision-maker's  objective is unnecessary. Hence, decisions and recommendations can be decoupled, e.g., by creating an agency assessing job applicants for all prospective employers. By contrast, for three or more states, we give an example where the optimal recommendation {depends on the decision-maker's objective, and so} delegation to an agency could lead to efficiency losses.

\medskip

{Our next application is to information design with privacy across receivers. As a motivating example,  consider a platform market connecting buyers and sellers.  Buyers are uncertain about the quality of a good, which plays the role of the state variable. The platform, informed of the state, wishes to reveal some information about it to the buyers in order to increase their welfare. 
By classical results in mechanism design \citep*{milgrom1982theory,mcafee1992correlated}, even minuscule correlation of information across buyers 
can be leveraged by the seller to extract full surplus.}

{In such settings where correlations in signals lead to undesirable outcomes, the solution to the information-design problem must be a private private information structure. This motivates our study of information design under privacy constraints, in which we abstract away from the origin of the constraint and impose it exogenously. We leverage our insights into private private structures to study the designer's problem. We provide necessary conditions for the optimal structure to be on the Blackwell-Pareto frontier, and describe the optimum explicitly for a class of separable objectives in the binary-state, two-receiver case. In this case, it suffices to use signals with at most three possible realizations, as opposed to infinitely many without the privacy constraint \citep*{arieli2020feasible, cichomski2023existence} and two realizations in the single-receiver case \citep*{kamenica2011bayesian}. In \S\ref{app:applic} we explore another multi-agent information design setting---influencing players in zero-sum games---where private private structures arise endogenously.
}

\medskip

{Private private information structures can also be used to measure the extent to which independent inputs affect a state of interest. We consider two such applications in the context of causal inference and Bayesian mechanism design. One of the basic causal networks is a collider, which represents a collection of independent causes generating an effect. We observe that the collider structure is a private private structure in which the effect takes the role of the state and the causes are the signals. Likewise, in mechanism design, the joint distribution of the individual types and the public outcome defines a private private structure.   

With these applications in mind, we use information-theoretic techniques to quantify and bound signal informativeness in private private structures. Specifically, we show that the sum across signals of the mutual information is bounded by the entropy of the state. 
This bound has a simple interpretation: one can think of private private structures as a way to divide an ``information pie'' so that the sum of pieces is smaller than the whole pie. These results immediately translate to bounds on the causal strength in collider structures  and on the responsiveness of mechanism outcomes to individual types.

\medskip

The connection between private private structures and Bayesian mechanism design outlined above has additional structural implications. We consider mechanisms for public decision making, i.e., social choice rules. Each such rule defines a joint distribution of individual types and the public outcome. The one-agent marginals of such a rule are 
studied 
in the reduced-form approach to mechanism design that replaces the original multi-agent design problem with auxiliary one-agent problems under additional feasibility constraints \citep*{matthews1984implementability,border1991implementation,kleiner2021extreme}. The feasible one-agent rules turn out to be tightly connected to the belief distributions that can be induced by private private information structures. Consequently, our results on private private structures imply new characterizations of feasible one-agent rules. In particular, we extend a version of the characterization of \cite*{hart2015implementation} to asymmetric rules, show that the result of \cite*{che2013generalized} admits a majorization form, and uncover a potential origin of computational intractability of the public-decision setting from \cite*{gopalan2018public}. Our results also imply an alternative simple proof of the main derandomization result of \cite*{chen2019equivalence}.}

\paragraph{More Related Literature.}\label{sec:intro_literature}

The question of which belief distributions can arise in private private information structures was addressed in \cite*{gutmann1991existence} and \cite{arieli2020feasible}. They provide a characterization for two {signals} under additional symmetry assumptions; we discuss the relation to our work below. \cite*{hong2009interpreted} consider the feasibility of private private signals in a stylized setting with binary signals and binary states of nature.  
Concurrently and independently, \cite*{cichomski2020maximal} studies private private structures that maximize the expected divergence between the {induced} posterior beliefs. His approach relies on a connection to the Gale-Ryser Theorem---a classical result on the existence of bipartite graphs with given degree distributions. The analysis suggests that this theorem can be used to derive the same {two-signal,} binary state feasibility result as our Corollary~\ref{cor:feasible} in  Appendix~\ref{sec:feasibility}. \cite*{cichomski2022combinatorial} rely on the connection to bipartite graphs to prove a tight upper bound on the probability that two beliefs induced by private private signals differ by more than some given $\delta$. 
For the question of feasibility without  the privacy constraint, see, e.g., \cite{dawid1995coherent, burdzy2019contradictory, burdzy2020bounds, arieli2020feasible, cichomski2021maximal}.

{Particular examples of private private signals have appeared in the social learning literature \citep*{gale1996have, ccelen2004distinguishing, ccelen2004observational} and in political economy~\citep*{ely2023ruth}. 
These signals arise as} the worst-case information structure for the auctioneer in some problems of robust mechanism design: see \cite*{bergemann2017first} and \cite*{brooks2021optimal}. Private private signals also appear as counterexamples of information aggregation in financial markets: see the discussion in \cite*{ostrovsky2012information} and similar observations  in the computer science literature \citep*{feigenbaum2003computation}.

{In a follow-up paper, \cite*{strack2023privacy} consider a version of our problem of optimal privacy-preserving recommendation and generalize the analysis in a number of directions. Most importantly, they show that the result on the existence of a dominant recommendation---obtained in our paper for the case of a binary state---extends to a real-valued state if the decision-maker's objective is a function of the posterior mean. The realistic scenario where a recommender gets a noisy signal about a binary state reduces to this setting by treating the recommender's belief as a new state variable. We further discuss the relation to their paper after presenting our results; see Footnote~\ref{fn:strack}.  
An alternative to the Blackwell perspective developed in \cite*{strack2023privacy} and in our paper is to fix a particular objective, which is a common approach in the computer science literature. For example, \cite*{li2018strong} consider maximizing mutual information between the recommendation and the state and construct recommendations that are approximately optimal within a logarithmic factor.}

{Our problem of optimal privacy-preserving recommendation admits an alternative fairness interpretation within the 
field of algorithmic fairness, an active area of research at the interface of economics and computer science; see, e.g., \cite{liang2022algorithmic,kleinberg2018algorithmic,rambachan2020economic,mehrabi2021survey,barocas-hardt-narayanan} and references therein.
This literature is concerned with designing AI-powered recommendation systems that avoid discrimination when handling legally protected attributes (e.g., loan applicants' gender or race). 
Our privacy requirement is equivalent to the fairness constraint of \emph{statistical parity}, one of the most popular fairness principles \citep{kleinberg2016inherent, aswani2022optimization, dwork2012fairness}. It requires that the distribution of the recommendations (e.g., the fraction of {the population} with a given level of credit score) has to be the same across sub-populations with different values of protected attributes.\footnote{The requirement of statistical parity is close in spirit to the concept of disparate impact from the US law, which forbids substantial statistical differences in outcomes for groups with different values of protected attributes.} In other words, the recommendation and the collection of protected attributes have to be independent of each other, forming a private private information structure. Thus, our results on optimal privacy-preserving recommendation provide a characterization of optimal recommendation systems satisfying statistical parity. Moreover, our framework allows one to quantify the efficiency loss due to the fairness constraint.}

{A popular notion of differential privacy from computer science formalizes privacy in another context; see~\cite*{dwork2014algorithmic} and \cite*{pai2013privacy} for  surveys. Each member of a population already possesses private information, and a social planner aims to compute an objective depending on this information. Differential privacy ensures that this computation does not expose individual private information too much. So, differential privacy is an approximate ``vertical'' privacy notion. By contrast, our paper introduces an exact ``horizontal'' privacy notion, capturing a social planner who possesses information and aims to distribute it in a private way across the population.
Concepts resembling differential privacy appeared earlier in the economic literature in the context of exchange economies \citep*{gul1992asymptotic} and have been recently brought to mechanism design \citep*{eilat2021bayesian} and information design \citep*{schmutte2022information}.}

As mentioned above, our work is related to the mathematics of sets of uniqueness and mathematical tomography \citep*{lorentz1949problem, fishburn1990sets, kellerer1993uniqueness}. These techniques have been applied in economics, for example, by \cite*{gershkov2013equivalence} to show the equivalence of Bayesian and dominant strategy implementation in an environment with linear utilities and one-dimensional types.

\section{Model}
\label{sec:model}
{A state of nature $\omega$ is a random element of $\Omega = \{0,1,\ldots,m-1\}$
with a full-support prior distribution. Consider a collection of $n\geq 2$ signals $s_1,\ldots, s_n$ about $\omega$.} 
We call the tuple $\mathcal{I} = (\omega,s_1,\ldots,s_n)$ an \emph{information structure}. Formally, fix a standard nonatomic Borel probability space  $(X,\Sigma,\mathbb{P})$, and let $\omega,s_1,\ldots, s_n$ be random variables  defined on this space  that take values in $\Omega \times S_1 \times \cdots \times S_n$, where each $S_i$ is a measurable space.\footnote{An {equivalent and more common} approach would be to define an information structure as a joint distribution over $\Omega \times S_1 \times \cdots \times S_n$. {In our paper, it is notationally convenient to assume that all information structures are defined on the same probability space $(X,\Sigma,\mathbb{P})$ so that the joint distributions of signals from different structures can be considered naturally.}}  The marginal distribution of $\omega$ is the prior over the state. 

Denote by $p(s_i)$ the posterior {belief about $\omega$ induced by observing $s_i$.}  Formally, $p(s_i)$ is the random variable taking value in $\Delta(\Omega)$ given by $p(s_i)(k) = \Pr{\omega=k}{s_i}$. In the case of a binary state (i.e., when $\Omega=\{0,1\}$), we let $p(s_i)$ take value in $[0,1]$ by setting $p(s_i) = \Pr{\omega=1}{s_i}$. 

\begin{definition}
We say that $\mathcal{I}=(\omega,s_1,\ldots,s_n)$ is a \emph{private private information structure} if $(s_1,\ldots,s_n)$ are independent random variables.
\end{definition}
Private private signals should not be confused with \emph{conditionally} independent signals, where $(s_1,\ldots,s_n)$ are independent given $\omega$.\footnote{{For instance, consider a binary state $\omega$, distributed uniformly over $\{0,1\}$, and two agents observing conditionally independent binary signals $s_1,s_2\in \{0,1\}$ that match the state with probability 3/4, i.e., $\bP[\omega=s_i\mid s_i]=3/4$. Then, each  signal contains information about the other signal: when an agent observes a high signal, she becomes more confident that her peer also got a high signal. Indeed, by Bayes formula, $\bP(s_j=1\mid s_i=1)=5/8$ for $i\ne j$.}}
 As a simple example of a private private information structure {with two signals}
 and a binary state, let $s_1,s_2$ be independently and uniformly distributed  on $[0,1]$, and let $\omega$ be the indicator of the event that $s_1+s_2 > 1$, as illustrated in Figure~\ref{fig:uniform}. The  distribution of $(s_1,s_2)$ conditioned on $\omega=1$ is the uniform distribution on the upper right triangle of the unit square. Conditioned on $\omega=0$, $(s_1, s_2)$ have the uniform distribution on the bottom left triangle. Note that  the posterior beliefs are $p(s_i)=s_i$ in this information structure, so both
 {signals are
 strictly informative.}  
 While the two signals are independent, they are not conditionally independent given the state $\omega$.

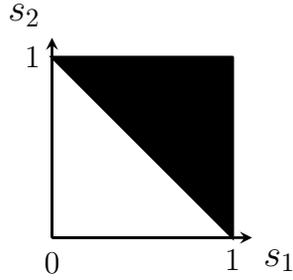
\begin{figure}[ht]
\begin{center}
\begin{tikzpicture}[scale=0.4, line width=1pt]
\draw (0,0) -- (6,0)--(6,6)--(0,6)--(0,0);

\draw [->,>=stealth,gray] (0,0) -- (6.65,0);
\draw [->,>=stealth,gray] (0,0) -- (0,6.65);
\node[below right] at (6.65,0) {\large $s_1$};
\node[above left] at (0,6.65) {\large $s_2$}; 

\node[below] at (6,0) {$1$};
\node[left] at (0,6) {$1$};
\node[below] at (0,-0.1) {$0$};

\filldraw[black] (6,6)--(6,0)--(0,6)--(6,6);

\end{tikzpicture}
\end{center}
\caption{The pair of signals $(s_1,s_2)$ is uniformly distributed on the unit square, with $\omega=1$ in the black area  and $\omega=0$ in the white area. The induced posteriors $p(s_1),p(s_2)$ coincide with the signals.
\label{fig:uniform}
}
\end{figure}

This paper focuses on characterizing the private private signals that are maximally informative,
formalized through the concept of \emph{Blackwell-Pareto optimality} of private private information structures. For the 
{case of one signal}
($n=1$), recall that an information structure $(\omega,s)$ \emph{Blackwell dominates} $(\omega,\hat s)$ if for every convex $\varphi \colon \Delta(\Omega) \to \R$ it holds that
    $\E{\varphi(p(s))} \geq \E{\varphi(p(\hat s))}.$
    
This notion captures a {uniform} sense in which $s$ contains more information about $\omega$ than $\hat s$ does: in any decision problem, an agent maximizing expected utility {performs weakly better when observing $s$ than observing $\hat s$.}

For more than one {signal,}
our next definition introduces a partial order on private private information structures that captures {Blackwell dominance for each signal.}
\begin{definition}\label{def_dominance}
Let $\mathcal{I}=(\omega,s_1,\ldots,s_n)$ and $\hat{ \mathcal{I}}=(\omega,\hat s_1,\ldots, \hat s_n)$ be private private information structures. We say that $\mathcal{I}$ \emph{dominates} $\hat{ \mathcal{I}}$, and write $\mathcal{I} \succeq \hat{ \mathcal{I}}$, if for every $i$ it holds that $(\omega,s_i)$ Blackwell dominates $(\omega,\hat s_i)$. We say that $\mathcal{I}$ and $\hat{ \mathcal{I}}$ are \emph{equivalent} if $\mathcal{I} \succeq\hat{ \mathcal{I}}$ and $\hat{ \mathcal{I}} \succeq \mathcal{I}$.
\end{definition}
It follows from this definition that $\mathcal{I}$ is equivalent to $\hat{\mathcal{I}}$ if and only if, for each $i$, the distributions of $p(s_i)$ and $p(\hat s_i)$ coincide. Thus we can partition the set of private private information structures into equivalence classes, with each class represented by~$n$ distributions $(\mu_1,\ldots,\mu_n)$ on $\Delta(\Omega)$.  

Figure~\ref{fig:34} illustrates another example of a private private information structure,  where the signals are again uniform on $[0,1]$, but each {signal induces
beliefs $\nicefrac{1}{4}$ or $\nicefrac{3}{4}$ equally likely.}
Thus this structure is equivalent to a structure {with binary signals.}
More generally, a structure $(\omega,s_1,\ldots,s_n)$ is always equivalent to the ``direct revelation'' structure $(\omega,p(s_1),\ldots,p(s_n))$ in which {each signal $s_i$ is replaced with the posterior belief it induces.}

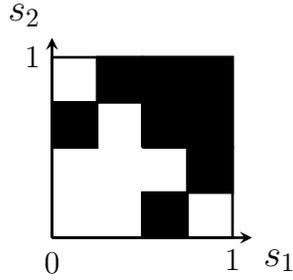
\begin{figure}
\begin{center}
\begin{tikzpicture}[scale=0.4, line width=1pt]
\filldraw[black] (3,3)--(6,3)--(6,6)--(3,6)--(3,3);

\filldraw[black] (0,3)--(1.5,3)--(1.5,4.5)--(0,4.5)--(0,3);
\filldraw[black] (1.5,4.5)--(3,4.5)--(3,6)--(1.5,6)--(1.5,4.5);

\filldraw[black] (3,0)--(4.5,0)--(4.5,1.5)--(3,1.5)--(3,0);
\filldraw[black] (4.5,1.5)--(6,1.5)--(6,3)--(4.5,3)--(4.5,1.5);

\draw (0,0) -- (6,0)--(6,6)--(0,6)--(0,0);

\draw [->,>=stealth,gray] (0,0) -- (6.65,0);
\draw [->,>=stealth,gray] (0,0) -- (0,6.65);
\node[below right] at (6.65,0) {\large $s_1$};
\node[above left] at (0,6.65) {\large $s_2$}; 

\node[below] at (6,0) {$1$};
\node[left] at (0,6) {$1$};
\node[below] at (0,-0.1) {$0$};

\end{tikzpicture}
\end{center}
\caption{The pair of signals $(s_1,s_2)$ is uniformly distributed on the unit square, with $\omega=1$ in the black area  and $\omega=0$ in the white area. The induced posteriors $p(s_1),p(s_2)$ are binary, and equally likely to be either \nicefrac{1}{4} or $\nicefrac{3}{4}$. {The posterior $p(s_2)$ is equal to $3/4$ on the top half on the square and to $1/4$ on the bottom half. Hence, $p(s_2)$ has the same distribution even conditioned on  $s_1$, and so the induced second-order and higher-order beliefs are trivial.}
\label{fig:34}}

\end{figure}

We use the concept of dominance to define Blackwell-Pareto optimality: which private private information structures provide a maximal amount of information, 
so that more information cannot be supplied without violating privacy? 
\begin{definition}
We say that a private private information structure $\mathcal{I}$ is \emph{Blackwell-Pareto optimal} if, for every private private information structure $\hat{ \mathcal{I}}$ such that  $\hat{ \mathcal{I}} \succeq \mathcal{I}$, the structure  $\hat{ \mathcal{I}}$ is equivalent to  $\mathcal{I}$.
\end{definition}
In other words, $\mathcal{I}$ is Blackwell-Pareto optimal if there is no private private information structure $\hat{ \mathcal{I}}$ {where each signal is as informative (in the Blackwell sense) as in $\mathcal{I}$ and at least one is strictly more informative. Consider $n$ decision makers each observing one of the signals. In this interpretation,
Blackwell-Pareto optimality captures a notion of Pareto optimality that is robust across decision problems. 
Indeed, a private private structure is Blackwell-Pareto optimal if, regardless of the decision problems agents face, their utilities cannot be Pareto improved within the class of private private structures.
{A priori, one could worry that there are very few Blackwell-Pareto optimal structures, in the sense that for any structure, one can always find a structure that is dominating, perhaps by very little. Lemma~\ref{lem:pareto-dominated} in the Appendix is a compactness result showing that, for every private private information structure, there exists a weakly dominating Blackwell-Pareto optimal one. In particular, this lemma indicates that 
the set of Blackwell-Pareto optimal structures is rich.}

As we explain in the introduction, there is some  tension between the privacy of an information structure and its informativeness. For example, {consider two agents observing a pair of signals.} The most informative structure from the point of view of agent 1 is the one where $s_1$ completely reveals the state, i.e., $p(s_1) = \delta_\omega$. Likewise, agent 2 would benefit most from a structure where $s_2$ perfectly reveals the state. But then $p(s_1)=p(s_2)$, and so $s_1$ and $s_2$ are not independent.  The question is thus: what are the ways to maximally inform the agents, while still maintaining privacy?

\section{Blackwell-Pareto Optimality and Conjugate Distributions}\label{sec_Pareto_for_two}
The question of Blackwell-Pareto optimality of private private information structures is already non-trivial in the case of two
{signals} 
and a binary state. For example, is the structure given in Figure~\ref{fig:uniform} Blackwell-Pareto optimal? What about the structure in Figure~\ref{fig:34}? In this section, we give a simple description of the Blackwell-Pareto frontier, making it easy to check if a structure is Blackwell-Pareto optimal. In particular, our results imply that the structure in Figure~\ref{fig:uniform} is Blackwell-Pareto optimal while the one in Figure~\ref{fig:34} is not.

To state this result, we   introduce \emph{conjugate distributions} on $[0,1]$. Let $F \colon [0,1] \to [0,1]$ be the cumulative distribution function of a probability measure in $\Delta([0,1])$. The associated \emph{quantile function}, which we denote by $F^{-1}$, is given by
\begin{align}
\label{eq:f-inv}
    F^{-1}(x) = \min\{y \,:\, F(y) \geq x\}.
\end{align}
Since cumulative distribution functions are right-continuous, this minimum indeed exists, and so $F^{-1}$ is well defined. When $F$ is the cumulative distribution function of a  full support and  nonatomic measure, then $F$ is a bijection and $F^{-1}$ is its inverse. More generally, $F^{-1}(x)$ is the smallest number $y$ such that an $x$-fraction of the population has values less than or equal to $y$.

\begin{definition}
The \emph{conjugate} of a cumulative distribution function $F \colon [0,1] \to [0,1]$ is the function $\hat F \colon [0,1] \to [0,1]$, which is given by
\begin{align*}
    \hat F(x) = 1-F^{-1}(1-x).
\end{align*}
\end{definition}
Graphically,  $(x,y)$ is on the graph of $F$ if and only if $(1-y,1-x)$ is on the graph of $\hat F$: in other words, $\hat F$ is the reflection of $F$ around the anti-diagonal of the unit square. An example is depicted in Figure~\ref{fig:conjugate}. 

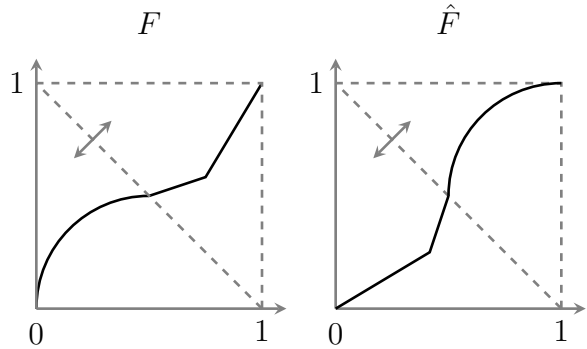
\begin{figure}
\begin{center}
\begin{tikzpicture}[scale=0.5, line width=1pt]

\draw (0,0) arc (180:90:3);
\draw (3,3) -- (4.5,3.5) -- (6,6);

\draw [->,>=stealth,gray] (0,0) -- (6.65,0);
\draw [->,>=stealth,gray] (0,0) -- (0,6.65);
\draw [dashed,gray] (6,0) -- (6,6) -- (0,6);
\draw [dashed,gray] (6,0) -- (0,6);
\draw [<->,>=stealth,gray] (1,4) -- (2,5);

\node[below] at (6,0) {$1$};
\node[left] at (0,6) {$1$};
\node[below] at (0,-0.1) {$0$};

\node[above] at (3,7) {$F$};

\end{tikzpicture}
\begin{tikzpicture}[scale=0.5, line width=1pt]
\draw [->,>=stealth,gray] (0,0) -- (6.65,0);
\draw [->,>=stealth,gray] (0,0) -- (0,6.65);
\draw [dashed,gray] (6,0) -- (6,6) -- (0,6);
\draw [dashed,gray] (6,0) -- (0,6);
\draw [<->,>=stealth,gray] (1,4) -- (2,5);

\draw (0,0) -- (2.5,1.5) -- (3,3);
\draw (3,3) arc (180:90:3);

\node[below] at (6,0) {$1$};
\node[left] at (0,6) {$1$};
\node[below] at (0,-0.1) {$0$};

\node[above] at (3,7) {$\hat F$};

\end{tikzpicture}
\end{center}
\caption{An example of a cumulative distribution function $F$ and its conjugate $\hat F$. The shapes under the curves are congruent: the transformation that maps one to the other is reflection around the anti-diagonal. Qualitatively, $F$ corresponds to the belief distribution of a more informative signal, and $\hat F$ corresponds to that of a less informative signal (because the former assigns less mass to posterior beliefs near 0.5).
\label{fig:conjugate}}

\end{figure}

As we show in the Appendix (Claim~\ref{clm:conjugate}), $\hat F$ is also a cumulative distribution function. Thus, given a measure $\mu \in \Delta([0,1])$, we can define its conjugate measure $\hat \mu \in \Delta([0,1])$ as the unique measure whose cumulative distribution function is the conjugate of the cumulative distribution function of $\mu$. It is easy to verify that the conjugate of $\hat \mu$ is again $\mu$.

The main result of this section is that Blackwell-Pareto optimality can be characterized in terms of conjugates.
\begin{theorem}
\label{thm:binary-pareto}
For a binary state $\omega$ and two {signals,} 
a private private information structure $\mathcal{I} = (\omega,s_1,s_2)$ is Blackwell-Pareto optimal if and only if the distributions of beliefs $p(s_1)$ and $p(s_2)$ are conjugates. 
\end{theorem}
{The essence of the proof is to show that every Blackwell-Pareto optimal structure is equivalent to a structure of the form described in Figures~\ref{fig:uniform} and~\ref{fig:induce_conjugates}: signals are independent and distributed uniformly on $[0,1]^2$, and there is an upward-closed\footnote{{By upward-closed sets, we mean sets that, with each point, also contain all points with higher or equal coordinates.}} subset $A \subseteq [0,1]^2$ such that $\omega=1$ whenever  $(s_1,s_2)\in A$.  Since $A$ is upward-closed, the graphs of the cumulative distribution functions of the beliefs induced by the two signals are both given by the boundary of $A$ (up to reflections) and are easily seen to be conjugates.
The formal} proof of Theorem~\ref{thm:binary-pareto} combines  our more general characterization of Blackwell-Pareto optimality in the {$n$-signal} case (Theorem~\ref{thm:pareto}) together with Theorem~\ref{th_Lorentz}, a classical result of \cite*{lorentz1949problem}  about so-called ``sets of uniqueness,'' which we discuss in detail in \S\ref{sec_Pareto_via_tomography}; these are subsets of $[0,1]^n$ that are uniquely determined by their projections to each of the $n$ axes. 

{A consequence of Theorem~\ref{thm:binary-pareto} is that there are many private private structures on the Blackwell-Pareto frontier; indeed there is such a structure for each pair of conjugate distributions. To see this, we apply an argument that is similar to the one used in the proof sketch above. Given a pair of conjugate cumulative distribution functions $F$ and $\hat{F}$, choose $(s_1,s_2)$ uniformly from the unit square, and let $\omega=1$ be the event that $s_2\geq \hat{F}(1-s_1)$. 
A simple calculation shows that
$\hat{F}(1-s_1)$ is equal to the posterior $p(s_1)$ and has 
 the distribution $F$, and $p(s_2)$ has the distribution $\hat F$. By Theorem~\ref{thm:binary-pareto}, this private private structure is Pareto optimal.  Figure~\ref{fig:induce_conjugates} illustrates this construction.}

Figure~\ref{fig:conjugate} suggests that on the Blackwell-Pareto frontier,  when $s_1$ is very  informative, $s_2$ must be very uninformative. We formalize this in the following claim, which is proved in the Appendix: 
\begin{corollary}\label{prop:comparative}
Suppose that $\omega$ is binary and that  $(\omega,s_1,s_2)$ and $(\omega,t_1,t_2)$ are Blackwell-Pareto optimal private private information structures. If $t_1$ dominates $s_1$, then $t_2$ is dominated by $s_2$.
\end{corollary}

\begin{figure}
\begin{center}
\begin{tikzpicture}[scale=0.4, line width=1pt]
\draw [->,>=stealth,gray] (0,0) -- (6.65,0);
\draw [->,>=stealth,gray] (0,0) -- (0,6.65);
\draw [dashed,gray] (6,0) -- (6,6) -- (0,6);
\filldraw (0,6) arc (90:0:3) -- (3,3) -- (3.5,1.5) -- (6,0) -- (6,6) -- (0,6);

\node[below] at (6,0) {$1$};
\node[left] at (0,6) {$1$};
\node[below] at (0,-0.1) {$0$};
\node[below right] at (6.65,0) {\large $s_1$};
\node[above left] at (0,6.65) {\large $s_2$};

\end{tikzpicture}
\end{center}
\caption{A private private information structure, where the beliefs $p(s_1)$ and $p(s_2)$ are distributed according to the pair of conjugate distributions $F$ and $\hat{F}$ from Figure~\ref{fig:conjugate}: the signals are uniform on $[0,1]^2$, and $\omega=1$ if and only if $s_2\geq \hat{F}(1-s_1)$ (black region).\label{fig:induce_conjugates} 
}

\end{figure}
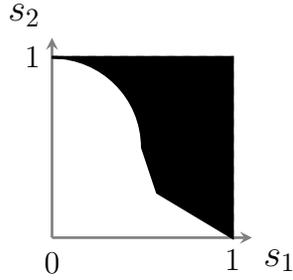

We can use Theorem~\ref{thm:binary-pareto} to understand whether the structures of Figures~\ref{fig:uniform} and~\ref{fig:34} are optimal. The uniform distribution on $[0,1]$ is its own conjugate. Hence, using Theorem~\ref{thm:binary-pareto}'s belief conjugacy test, we can conclude that   Figure~\ref{fig:uniform}'s information structure is Blackwell-Pareto optimal.

To understand the structure of Figure~\ref{fig:34},  consider, more generally, a discrete distribution $\mu$ on $[0,1]$ with $k$ atoms. Its conjugate $\hat{\mu}$ is also atomic: each atom of $\mu$ with weight $w$ corresponds to an interval of zero  mass with length $w$ for $\hat{\mu}$ and, symmetrically, each interval of length $l$ carrying no atoms in $\mu$ becomes an atom of weight $l$ in~$\hat{\mu}$ (see Figure~\ref{fig:discrete_conjugates}).
In particular, $\hat \mu$ has either $k-1$, $k$ or $k+1$ atoms, corresponding to the cases that (1) $\mu$ places positive mass on both $0$ and $1$, (2) $\mu$ places positive mass on exactly one of $\{0,1\}$, or (3) $\mu$ places zero mass on $\{0,1\}$. We obtain the following corollary.

\begin{figure}
\begin{center}
\begin{tikzpicture}[scale=0.5, line width=1pt]

\draw [->,>=stealth,gray] (0,0) -- (6.65,0);
\draw [->,>=stealth,gray] (0,0) -- (0,6.65);
\draw [dashed,gray] (6,0) -- (6,6) -- (0,6);

\draw (0,0) -- (0.6,0);
\draw (0.6,1.2)--(2.4,1.2);
\draw (2.4,3)--(3.6,3);
\draw (3.6,6) -- (6,6);
\draw [dashed,gray] (0.6,0) -- (0.6,1.2);
\draw [dashed,gray] (2.4,1.2) -- (2.4,3);
\draw [dashed,gray] (3.6,3) -- (3.6,6);
\node[below] at (0.6,0) {$0.1$};
\node[below] at (2.4,0) {$0.4$};
\node[below] at (3.6,0) {$0.6$};
\node[left] at (0,1.2) {$0.2$};
\node[left] at (0,3) {$0.5$};

\filldraw[color=black!100, fill=black!100,  thick] (0.6,1.2) circle (0.1);
\filldraw[color=black!100, fill=black!100,  thick] (0.6,1.2) circle (0.1);\filldraw[color=black!100, fill=black!100,  thick] (2.4,3) circle (0.1);\filldraw[color=black!100, fill=black!100,  thick] (3.6,6) circle (0.1);

\filldraw[color=black!100, fill=white!100,  thick] (0.6,0) circle (0.1);
\filldraw[color=black!100, fill=white!100,  thick] (2.4,1.2) circle (0.1);
\filldraw[color=black!100, fill=white!100,  thick] (3.6,3) circle (0.1);

\node[below] at (6,0) {$1$};
\node[left] at (0,6) {$1$};

\node[above] at (3,7) {$F$};

\end{tikzpicture}
\qquad
\begin{tikzpicture}[scale=0.5, line width=1pt]

\draw [->,>=stealth,gray] (0,0) -- (6.65,0);
\draw [->,>=stealth,gray] (0,0) -- (0,6.65);
\draw [dashed,gray] (6,0) -- (6,6) -- (0,6);

\draw (0,2.4) -- (3,2.4);
\draw (3,3.6)--(4.8,3.6);
\draw (4.8,5.4)--(6,5.4);
\draw [dashed,gray] (3,2.4) -- (3,3.6);
\draw [dashed,gray] (4.8,3.6) -- (4.8,5.4);
\node[below] at (3,0) {$0.5$};
\node[below] at (4.8,0) {$0.8$};
\node[left] at (0,2.4) {$0.4$};
\node[left] at (0,3.6) {$0.6$};
\node[left] at (0,5.4) {$0.9$};

\filldraw[color=black!100, fill=black!100,  thick] (3,3.6) circle (0.1);
\filldraw[color=black!100, fill=black!100,  thick] (4.8,5.4) circle (0.1);\filldraw[color=black!100, fill=black!100,  thick] (6,6) circle (0.1);

\filldraw[color=black!100, fill=white!100,  thick] (3,2.4) circle (0.1);
\filldraw[color=black!100, fill=white!100,  thick] (4.8,3.6) circle (0.1);
\filldraw[color=black!100, fill=white!100,  thick] (6,5.4) circle (0.1);

\node[below] at (6,0) {$1$};
\node[left] at (0,6) {$1$};
\node[below] at (0,-0.1) {$0$};

\node[above] at (3,7) {$\hat F$};

\end{tikzpicture}
\end{center}
\caption{The conjugate of a discrete distribution $F$ with three atoms at $0.1$, $0.4$, and $0.6$. Each atom becomes an interval of zero measure with the length equal to the atom's weight, and vice versa. Since $F$ does not have atoms at the endpoints of $[0,1]$, the number of intervals of zero measure exceeds the number of atoms by one, so its conjugate $\hat{F}$ has four atoms at $0$, $0.5$, $0.8$, and $1$. \label{fig:discrete_conjugates} 
}

\end{figure}
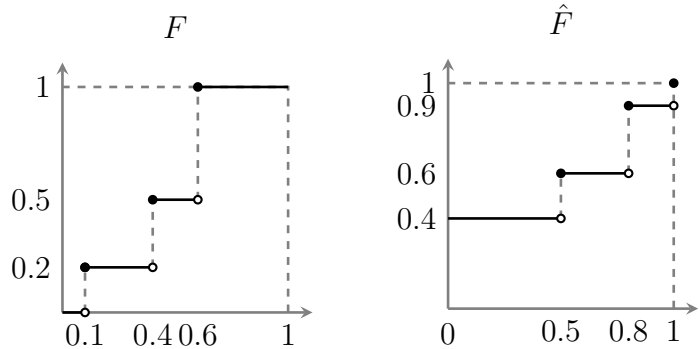
\begin{corollary}
     Two discrete distributions with the same number of atoms can only be conjugates if each of them assigns a non-zero weight to exactly one of $\{0,1\}$. 
\end{corollary}
We conclude that the information structure of Figure~\ref{fig:34},  where both signals induce  beliefs  $\nicefrac{1}{4}$ or $\nicefrac{3}{4}$ is not Blackwell-Pareto optimal.

\section{Blackwell-Pareto Optimality and Sets of Uniqueness}\label{sec_Pareto_via_tomography}

In this section, we study Blackwell-Pareto optimality of private private information structures in the general setting of $n$ {signals} and a state $\omega$ that takes value in $\Omega = \{0,\ldots,m-1\}$. Our main result shows that Blackwell-Pareto optimality can be characterized using ``sets of uniqueness'': subsets of $[0,1]^n$ that are uniquely determined by their projections to the $n$ axes.

As a first step, we show that it is without loss of generality to focus on information structures that are constructed similarly to the examples we have considered above: each $s_i$ is distributed uniformly on $[0,1]$, and each value of $\omega$ corresponds to some subset of $[0,1]^n$. That is, $\omega$ is a deterministic function of the signals (see Figures~\ref{fig:uniform} and \ref{fig:34}, as well as Figure \ref{fig_nonunique} in the Appendix).

More formally, let $\cA = (A_0,\ldots,A_{m-1})$ be a \emph{partition} of $[0,1]^n$ into measurable sets. That is, each $A_k$ is a measurable subset of $[0,1]^n$, the sets in $\cA$ are disjoint, and their union is equal to~$[0,1]^n$. 
\begin{definition}
The private private information structure \emph{associated} with a partition $\cA = (A_0,\ldots,A_{m-1})$ is $\mathcal{I} = (\omega,s_1,\ldots,s_n)$ where $(s_1,\ldots,s_n)$ are distributed uniformly on $[0,1]^n$ and $\{\omega = k\}$ is the event that $\{(s_1,\ldots,s_n) \in A_k\}$.
\end{definition}
Note that if $\cA$ and $\cA'$ are partitions such that each symmetric difference  $A_k \triangle A_k'$ has zero Lebesgue measure, then both are associated with the same information structure, in the strong sense that the joint distributions of $(\omega,s_1,\ldots,s_n)$ coincide. Accordingly, we henceforth consider two subsets of $[0,1]^n$ to be equal if they only differ on a zero-measure set. 

Private private information structures associated with partitions have two important properties that make them useful. First, signals are uniform on $[0,1]$, so that one does not need to consider abstract signal spaces. Second, the collection of signals $(s_1,\ldots,s_n)$ reveals the state. The next result shows that---up to equivalence---these assumptions are without loss of generality.
\begin{proposition}
\label{prop:associated}
For every private private information structure $\mathcal{I}$, there exists a partition $\cA$ whose associated information structure $\mathcal{I}'$ is equivalent to $\mathcal{I}$.
\end{proposition}
While the general proof contained in Appendix~\ref{app_proof_representation} is not constructive, for structures with a finite number of signals and a binary state, we show in Appendix~\ref{app_explicit_representation} how to construct a partition with an equivalent associated structure.

The ideas behind the proof of this proposition are the following. Using standard results,  one  can always  reparameterize the signals so that they are uniformly distributed in $[0,1]$. Thus the main challenge is to ensure that the state is determined by the signals. To this end, given signals that do not determine the state, we add a signal $t$ that resolves the remaining uncertainty, so that $\omega$ is a deterministic function of $(s_1,\ldots,s_n,t)$. Then, we use a ``secret sharing'' technique \citep[similar to][]{shamir1979share} to create a pair of independent random variables $t_1,t_2$ such that $t$ is determined by the pair $(t_1,t_2)$, but each $t_i$ is uninformative about the state and the other signals.\footnote{In the common uses of secret sharing $t,t_1,t_2$ are binary, while our construction requires that they are distributed uniformly on $[0,1]$.} We then 
{augment the signals $s_1$ and $s_2$ with} 
the additional signals $t_1$ and $t_2$, respectively. Thus the information structure $\big(\omega,(s_1,t_1),(s_2,t_2),\ldots,s_n\big)$ is equivalent to $(\omega,s_1,\ldots,s_n)$, but now the signals determine the state.\footnote{{It is instructive to compare Proposition~\ref{prop:associated} with the classical purification result of \cite*{dvoretzky1951elimination}. Given a finite set 
$\Omega = \{0, \ldots, m-1\}$, a measurable map 
$\varphi\colon [0,1]^n \to \Delta(\Omega)$, and a finite collection $F$ of 
integrable functions on $[0,1]^n$, the Dvoretzky-Wald-Wolfowitz theorem 
ensures the existence of a partition 
$\cA = (A_0, \ldots, A_{m-1})$ of $[0,1]^n$ such that
$
\int_{[0,1]^n} \varphi(x,k) f(x)\, \dd x = \int_{A_k} f(x)\, \dd x
$
for every $f \in F$ and each $k \in \Omega$. To derive 
Proposition~\ref{prop:associated} from this theorem, one would need to take 
$F$ to be the infinite collection of all univariate functions $f(x_i)$. However, the 
theorem applies only to finite collections. This obstacle also arises in 
\cite{zeng2023derandomization}, who examines how private types can be 
used to eliminate randomness in two-way communication.}
}

Proposition~\ref{prop:associated} implies that for the purpose of studying the Blackwell-Pareto optimality of private private signals, one can assume without loss of generality that information structures are always associated with partitions. In particular, the question of Blackwell-Pareto optimality can now be phrased as a question about partitions: when does a partition $\cA$ correspond to a Blackwell-Pareto optimal structure? Our next result answers this question. We state this result for the case of a binary state, as it involves significantly simpler notation; the result for a general finite state space  is given in Appendix~\ref{app:Pareto_via_uniqueness}. In the case of a binary state, a partition $\cA=(A_0,A_1)$ is determined by $A_1$, since $A_0$ is its complement. Hence, we will represent $\cA$ by a single set $A = A_1$. The information structure associated with $A$ will refer to the structure associated with the partition $(A^c,A)$.

Given a measurable set $A \subseteq [0,1]^n$, we define $n$ functions $(\alpha_1^A,\ldots,\alpha_n^A)$ that capture the projections of $A$ to the $n$ coordinate axes. Denote by $\lambda$ the Lebesgue measure on $[0,1]^{n-1}$. The projection $\alpha_i^A \colon [0,1] \to [0,1]$ of $A$ to the $i$th axis is
\begin{align*}
    \alpha_i^A(t) = \lambda(\{y_{-i} \,:\, (y_i,y_{-i}) \in A,\, y_i=t\}).
\end{align*}

If $(\omega,s_1,\ldots,s_n)$ is the information structure associated with $A$, then $\alpha_i^A(t)$ is the posterior {induced  by} 
$s_i=t$.
\begin{definition}
A measurable $A \subseteq [0,1]^n$ is a \emph{set of uniqueness} if for every measurable $B\subseteq [0,1]^n$ such that $(\alpha_1^A,\ldots,\alpha_n^A)=(\alpha_1^B,\ldots,\alpha_n^B)$,  it holds that $A=B$.
\end{definition}
Less formally, $A$ is a set of uniqueness if it is determined by the projections $(\alpha_1^A,\ldots,\alpha_n^A)$. {The black area in the left panel of Figure~\ref{fig:uniqueness} is an example of a set of uniqueness. This follows from the fact that upward-closed sets of $[0,1]^2$ are sets of uniqueness (see Theorem~\ref{th_Lorentz} below). Examples of  sets that are not a set of uniqueness are the black areas in the middle and right panels, as they are different but have the same projections on the axes.}

\begin{figure}
\begin{center}
\begin{tikzpicture}[scale=0.4, line width=1pt]
\draw [->,>=stealth,gray] (0,0) -- (6.65,0);
\draw [->,>=stealth,gray] (0,0) -- (0,6.65);
\draw [dashed,gray] (6,0) -- (6,6) -- (0,6);
\filldraw (0,6) arc (90:0:3) -- (3,3) -- (3.5,1.5) -- (6,0) -- (6,6) -- (0,6);

\end{tikzpicture}
\hskip 1 cm
\begin{tikzpicture}[scale=0.4, line width=1pt]
\filldraw[black] (3,3)--(6,3)--(6,6)--(3,6)--(3,3);

\filldraw[black] (0,3)--(1.5,3)--(1.5,4.5)--(0,4.5)--(0,3);
\filldraw[black] (1.5,4.5)--(3,4.5)--(3,6)--(1.5,6)--(1.5,4.5);

\filldraw[black] (3,0)--(4.5,0)--(4.5,1.5)--(3,1.5)--(3,0);
\filldraw[black] (4.5,1.5)--(6,1.5)--(6,3)--(4.5,3)--(4.5,1.5);

\draw (0,0) -- (6,0)--(6,6)--(0,6)--(0,0);

\draw [->,>=stealth,gray] (0,0) -- (6.65,0);
\draw [->,>=stealth,gray] (0,0) -- (0,6.65);

\end{tikzpicture}
\hskip 1 cm
\begin{tikzpicture}[scale=0.4, line width=1pt]
\filldraw[black] (3,3)--(6,3)--(6,6)--(3,6)--(3,3);

\filldraw[black] (0,4.5)--(1.5,4.5)--(1.5,6)--(0,6)--(0,4.5);
\filldraw[black] (1.5,3)--(3,3)--(3,4.5)--(1.5,4.5)--(1.5,3);

\filldraw[black] (3,1.5)--(4.5,1.5)--(4.5,3)--(3,3)--(3,1.5);
\filldraw[black] (4.5,0)--(6,0)--(6,1.5)--(4.5,1.5)--(4.5,0);

\draw (0,0) -- (6,0)--(6,6)--(0,6)--(0,0);

\draw [->,>=stealth,gray] (0,0) -- (6.65,0);
\draw [->,>=stealth,gray] (0,0) -- (0,6.65);

\end{tikzpicture}

\end{center}
\caption{{The black area in the left panel is a set of uniqueness, since it is upward-closed (Theorem~\ref{th_Lorentz}). The middle and right panels are not sets of uniqueness, since they are different but have the same projections on the axes.}\label{fig:uniqueness}}

\end{figure}

The main result of this section characterizes Blackwell-Pareto optimality in terms of sets of uniqueness.
\begin{theorem}
\label{thm:pareto}
A private private information structure is Blackwell-Pareto optimal if and only if it is equivalent to a structure associated with a set of uniqueness $A \subseteq [0,1]^n$.
\end{theorem}
{The connection between Blackwell-Pareto optimality and sets of uniqueness may seem surprising. To get some intuition, 
we first explain how we prove that Blackwell-Pareto optimal structures are associated with sets of uniqueness. Consider a private private structure  associated with a set $A$ that is not a set of uniqueness (for example, the middle panel of Figure~\ref{fig:uniqueness}). We show that this structure can be improved. Since $A$ is not a set of uniqueness, there is a set $B\ne A$ that gives rise to an equivalent structure (e.g., the shape in the right panel).
By randomizing between the two structures using an independent coin flip $t$, we arrive at another equivalent structure $(\omega,s_1,\ldots,s_n)$. In this structure, the signals do not always determine the state and the coin flip $t$ resolving this uncertainty is informative about $\omega$.
Augmenting the first signal with  $t$, we obtain a private private information structure $\big(\omega,(s_1,t),s_2,\ldots,s_n\big)$ that dominates the structure associated with $A$. 
}

Conversely, suppose the information structure associated with $A$ is not Blackwell-Pareto optimal. By considering a Blackwell-Pareto dominating information structure and garbling the signals, we can find a density $f: [0,1]^n \to [0,1]$ that is not an indicator function, but has the same marginals as $A$. We next apply a result of \cite*{gutmann1991existence}: the set of densities valued in $[0,1]$ with given marginals is a convex set whose extreme points are indicator functions. Since $f$ is not an indicator function, the corresponding convex set is not a singleton and must have at least two different extreme points. There exists some other set with the same marginals as $A$, so  $A$ is not a set of uniqueness.  

Theorem~\ref{thm:pareto} shows an equivalence between the two \emph{a priori} unrelated notions of Blackwell-Pareto optimality and sets of uniqueness; a similar result in Appendix~\ref{app:Pareto_via_uniqueness} establishes an analogous equivalence for arbitrary finite state spaces, generalizing sets of uniqueness to \emph{partitions of uniqueness}. This connection lets us use characterization results about sets of uniqueness to study Blackwell-Pareto optimality.   Sets of uniqueness have been studied since \cite*{lorentz1949problem}, who gives a simple characterization in the two-dimensional case. A version of his characterization, as we explain below, leads to Theorem~\ref{thm:binary-pareto}. Beyond the two-dimensional case, sets of uniqueness have also been more recently studied in the \emph{mathematical tomography} literature \citep[e.g.,][]{fishburn1990sets}.  We  discuss below how these newer results help us understand Blackwell-Pareto optimal structures. 

To characterize sets of uniqueness in two dimensions, we will need the following definitions. Say that $A \subseteq [0,1]^2$ is a \emph{rearrangement} of $B \subseteq [0,1]^2$ if for $i=1,2$ and every $q \in [0,1]$, the sets $\{t \in [0,1]\,:\, \alpha_i^A(t) \leq q\}$ and $\{t \in [0,1]\,:\, \alpha_i^B(t) \leq q\}$ have the same Lebesgue measure. That is, $\alpha_i^A$ and $\alpha_i^B$, when viewed as random variables defined on $[0,1]$, have the same distribution. This has a simple interpretation in terms of information structures: $A$ is a rearrangement of $B$ if and only if the two associated information structures are Blackwell equivalent.  This is immediate since in the information structure associated with $A$,  $\alpha_i^A(t)$ is the belief {induced by} 
$s_i=t$. Recall that $B \subseteq [0,1]^n$ is \emph{upward-closed} if $x=(x_1,\ldots,x_n) \in B$ implies that $y=(y_1,\ldots,y_n)\in B$ for all $y \geq x$. 
\begin{theorem}[\cite*{lorentz1949problem}]\label{th_Lorentz}

A measurable subset $A \subseteq [0,1]^2$ is a set of uniqueness if and only if it is a rearrangement of an upward-closed set.
\end{theorem}
This formulation is different than the one that appears in the paper by \cite*{lorentz1949problem}. We, therefore, show in the {Appendix~\ref{sec_proof_th1}} that it is an equivalent characterization. Theorem~\ref{thm:binary-pareto} is a consequence of Theorems~\ref{thm:pareto} and~\ref{th_Lorentz}. 

{When $n \geq 3$, the known characterizations of sets of uniqueness are more involved \citep*{kellerer1993uniqueness}. 
However, there are simple necessary and simple sufficient conditions. Upward closedness remains a necessary condition: any set of uniqueness is a rearrangement of an upward-closed set. Combining this insight with Theorem~\ref{thm:pareto}, we get the following corollary that can be seen as an extension of Theorem~\ref{thm:binary-pareto} for $n>2$ signals.
\begin{corollary}
\label{cor:necessary_pareto}
For a binary state $\omega$ and $n\geq 2$ signals, 
any Blackwell-Pareto optimal private private information structure $\mathcal{I} = (\omega,s_1,\ldots,s_n)$ is equivalent to a structure associated with an upward-closed set $A\subseteq [0,1]^n$.
\end{corollary}
Note that an upward closed set is pinned down by its frontier. By Corollary~\ref{cor:necessary_pareto}, this frontier provides a natural parametrization to the corresponding superset of Blackwell-Pareto optimal structures.

For $n \geq 3$, upward-closedness is not a sufficient condition for a set to be a set of  uniqueness. In particular, not all upward-closed sets give rise to Blackwell-Pareto optimal structures. A sufficient  condition for uniqueness \citep*{fishburn1990sets} is  to be an \emph{additive} set: this holds when there are   bounded  $h_{i}\colon [0,1]\to\mathbb{R}$ such that  $$A=\left\{x\in[0,1]^{n}:\sum_{i=1}^{n}h_{i}(x_{i})\ge0\right\}.$$
Clearly, any additive set is upward-closed up to a rearrangement making each $h_i$ non-decreasing.
In two dimensions, the two concepts are equivalent,}
and so additivity provides another characterization of the sets of uniqueness (and equivalently, of the Blackwell-Pareto optimal structures).
In higher dimensions (i.e., with three or more signals), the sufficiency of additivity implies that every additive set is associated with a Blackwell-Pareto optimal structure. With $n\ge3$, \cite*{kemperman1991sets} demonstrated that there are sets of uniqueness that are not additive. However, additivity is ``almost necessary'':  \cite*{kellerer1993uniqueness} characterizes sets of uniqueness as the closure, in a certain topology, of the class of additive sets.

In light of these results, progress in understanding Blackwell-Pareto optimal private private structures for more than two {signals} is contingent on new breakthroughs regarding these old questions.  
{Another promising} direction for future research is to characterize Blackwell-Pareto optimal private private structures  for a non-binary state and two {signals}. As we discuss in Appendix~\ref{app:Pareto_via_uniqueness}, this will require understanding {\em partitions of uniqueness} with more than two elements. To the best of our knowledge, this has not been investigated in the literature.

\section{{Application to Privacy-Preserving Recommendations}}\label{sec_optimal_recommendation}

{An informed party  wishes to disclose to an uninformed} party as much information as possible about the state $\omega$ using a message $s_2$, but must not reveal any information about a  given random variable $s_1$ that is correlated with $\omega$. In this application, we should interpret {$s_1$ 
as an additional component of the state that must be kept secret for legal or security reasons}, and where the joint distribution of $(\omega,s_1)$ is given exogenously.

As a concrete example,  suppose an uninformed company wants to learn about a decision-relevant type $\omega$ of an applicant (e.g., whether she is a good fit for a job or whether she will pay her rent on time), and an informed party (e.g., a recommender or a credit-rating company) knows both this type and a sensitive or legally protected attribute $s_1$ of the applicant that correlates with the type: this could be the applicant's private medical information, gender, or race. The informed party faces the problem of optimal privacy-preserving recommendation: convey as much information as possible about the applicant without revealing any information about her protected attribute, so that the company's downstream decision based solely on the recommendation will be independent of the protected attribute and therefore not cause disparate impact. Note that even a recommendation that does not explicitly contain the protected attribute may cause disparate impact, if it contains correlates of the attribute.  {Our analysis makes two simplifying assumptions: first, the recommender can infer the state exactly from the data\footnote{{\cite*{strack2023privacy} generalize our results dropping this assumption.}}; second, whereas the algorithmic fairness literature typically features an approximate privacy requirement~\citep*{barocas-hardt-narayanan}, we impose an exact privacy requirement.}

 A less economic (but more colorful) story is that of a government that would like to reveal a piece of intelligence $\omega$, but without revealing any information about the identity of its source $s_1$. These could be naturally correlated: for example, if $\omega$ is the location of an enemy's weapons facility and the source $s_1$ is likely to live close to it. So the government's disclosure $s_2$ should contain as much information as possible about $\omega$, while not revealing any information about $s_1$. 
 
Motivated by these examples, we give the following definition.
\begin{definition}\label{def_disclosure}
{Given an information structure $(\omega,s_1)$, a signal $s_2$ independent of $s_1$ is an \emph{optimal privacy-preserving recommendation}  if there is no $s'_2$ that is independent of $s_1$ and Blackwell dominates  $s_2$.}
\end{definition}
Note that in this definition, we take an approach that is agnostic to any decision problem faced by the receiver or, indeed, any other goal the receiver might have, such as selling this information, using it in a game, etc. Instead, we study signals that are maximal in their information content, and can be used by the receiver towards any of these goals.

The notion of an optimal privacy-preserving recommendation is closely related to Blackwell-Pareto optimality, namely, it is equivalent to requiring that $\mathcal{I} = (\omega,s_1,s_2)$ is a Blackwell-Pareto optimal private private information structure. Indeed, if $(\omega,s_1,s_2)$ is Blackwell-Pareto optimal, none of the signals can be made more informative under the privacy constraint, and so $s_2$ is an optimal privacy-preserving recommendation. The converse is  established in  Appendix~\ref{app_disclosure_proof} (Corollary~\ref{cor_disclosure_def_equivalence}).\footnote{{This corollary is non-trivial as it relies on the following result: given a sub-optimal structure $(\omega,s_1,s_2)$, there is always a dominating structure $(\omega,t_1,t_2)$ where $s_1=t_1$. A priori, one could imagine that, in some cases, only $s_1$ can be improved, or that the improvement of $s_2$ requires replacing $s_1$ with an equivalent (but not equal) signal. }}

{A natural question is whether there are many optimal signals---each appropriate for a different decision problem---or if there is one optimal signal that dominates all others.}
\begin{definition}\label{def_dominant}
{Given an information structure $(\omega,s_1)$, a signal $s_2$ independent of $s_1$ is a \emph{dominant privacy-preserving recommendation} if $s_2$ Blackwell dominates every $s_2'$ that is independent of $s_1$.}
\end{definition}
{When $\omega$ and $s_1$ are uncorrelated, a completely revealing $s_2$ is dominant. When there is correlation, a completely revealing $s_2$ is no longer feasible, and so one could expect that no dominant signal exists, i.e., decision-makers with different objectives prefer different optimal signals.
Our characterization of the Blackwell-Pareto optimal private private signals (Theorem~\ref{thm:binary-pareto}) implies that a dominant signal exists and can be easily constructed, when the state is binary.\footnote{{We call this result a theorem to highlight its importance.  The fact that the signal $s_2^\star$ inducing the conjugate distribution dominates any other $s_2$ independent of $s_1$
is a direct corollary of Theorem~\ref{thm:binary-pareto}. However, the possibility to construct such a signal $s_2^\star$ for given $s_1$ requires a proof that can be found in Appendix~\ref{app_disclosure_proof}.}}
 }
\begin{theorem}
\label{thm:disclosure}
  {For any $(\omega,s_1)$ with a binary state $\omega$, there exists a dominant privacy-preserving recommendation $s_2^\star$. The distribution of $p(s_2^\star)$ is the conjugate of the distribution of $p(s_1)$.}
\end{theorem}
{Theorem~\ref{thm:disclosure} implies that every decision maker would choose the same signal $s_2^\star$, regardless of the decision problem at hand. For example, the posterior $p(s_2^\star)$ simultaneously maximizes the mutual information with $\omega$, as well as minimizes the mean squared distance from it. This uniqueness of the optimal privacy-preserving recommendation is a rather surprising property as the Blackwell order is a partial order, and so one could expect that there are non-equivalent optimal recommendations that are all maximal and incomparable.} Indeed, in Appendix~\ref{app_non_uniqueness}, we demonstrate that uniqueness is a feature of the binary state case by considering an example with three states, binary~$s_1$, and a continuum of {non-equivalent} optimal privacy-preserving recommendations.\footnote{We conjecture that uniqueness of the optimal recommendation for two states and non-uniqueness for three or more states can be related to the fact that the Blackwell order has the lattice property only in the binary setting \citep*{bertschinger2014blackwell,de2021robust}, but we are unaware of a formal connection.}$^{,}$\footnote{\label{fn:strack}{\cite*{strack2023privacy} study a similar but subtly different optimality notion: they consider $s_2$ as a signal not just about $\omega$, but about $(\omega,s_1)$, and so $s_2$ is optimal, in their notion, if there is no $s_2'$ independent of $s_1$ that contains more information about $(\omega,s_1)$. In their setting, a dominant $s_2$ may fail to exist even for binary $\omega$. For example, if $(\omega,s_1)$ are independent and distributed uniformly on $\{-1,+1\}$, then both $s_2=\omega$ and $s_2' = \omega \cdot  s_1
$ are optimal, according to their notion, but such $s_2'$ carries no information about $\omega$ and thus is not optimal in our terms.
Their Theorem 3, which 
considers yet another optimality notion, 
generalizes our Theorem~\ref{thm:disclosure}: it can be used to show that the recommender has a dominant privacy-preserving recommendation in our sense even when only observing a noisy signal about $\omega$. 
}}

Figure~\ref{fig:disclosure} shows the dominant privacy-preserving recommendation when the two states are equally likely and $s_1$ is a symmetric binary signal that matches the state with probability $\nicefrac{3}{4}$. The {dominant} recommendation $s_2^\star$ has the following form: it completely reveals the state with probability $\nicefrac{1}{2}$, and gives no information with the remaining probability---that is, $s_2^\star$ is trinary. More generally, when the states are equally likely and $s_1$ is a symmetric binary signal that matches the state with probability $r\in [\nicefrac{1}{2},1]$, the {dominant} recommendation will be trinary. It completely reveals the state with probability~$2(1-r)$, and gives no information with the complementary probability. Thus, as   the correlation between $s_1$ and $\omega$ increases, the {dominant} recommendation becomes less informative.

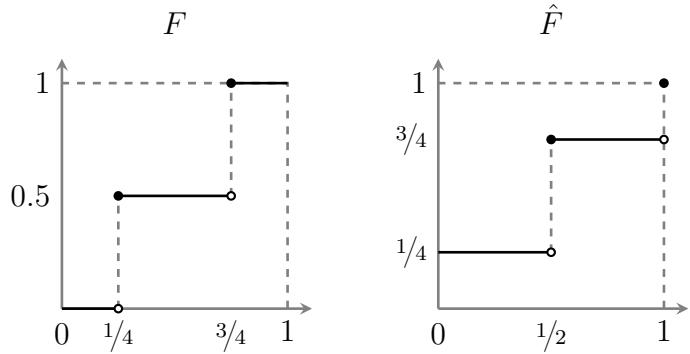
\begin{figure}[ht]
\begin{center}

\begin{tikzpicture}[scale=0.5, line width=1pt]

\draw [->,>=stealth,gray] (0,0) -- (6.65,0);
\draw [->,>=stealth,gray] (0,0) -- (0,6.65);
\draw [dashed,gray] (6,0) -- (6,6) -- (0,6);

\draw (0,0) -- (1.5,0);
\draw (1.5,3)--(4.5,3);
\draw (4.5,6)--(6,6);
\draw [dashed,gray] (1.5,0) -- (1.5,3);
\draw [dashed,gray] (4.5,3) -- (4.5,6);
\node[left] at (0,3) {$0.5$};
\node[below] at (1.5,0) {$\nicefrac{1}{4}$};
\node[below] at (4.5,0) {$\nicefrac{3}{4}$};

\filldraw[color=black!100, fill=black!100,  thick] (1.5,3) circle (0.1);
\filldraw[color=black!100, fill=black!100,  thick] (4.5,6) circle (0.1);

\filldraw[color=black!100, fill=white!100,  thick] (1.5,0) circle (0.1);
\filldraw[color=black!100, fill=white!100,  thick] (4.5,3) circle (0.1);

\node[below] at (6,0) {$1$};
\node[left] at (0,6) {$1$};
\node[below] at (0,-0.1) {$0$};

\node[above] at (3,7) {$F$};
\end{tikzpicture}
\qquad
\begin{tikzpicture}[scale=0.5, line width=1pt]

\draw [->,>=stealth,gray] (0,0) -- (6.65,0);
\draw [->,>=stealth,gray] (0,0) -- (0,6.65);
\draw [dashed,gray] (6,0) -- (6,6) -- (0,6);

\draw (0,1.5) -- (3,1.5);
\draw (3,4.5)--(6,4.5);
\draw [dashed,gray] (3,1.5) -- (3,4.5);
\node[below] at (3,0) {$\nicefrac{1}{2}$};
\node[left] at (0,1.5) {$\nicefrac{1}{4}$};
\node[left] at (0,4.5) {$\nicefrac{3}{4}$};

\filldraw[color=black!100, fill=black!100,  thick] (3,4.5) circle (0.1);
\filldraw[color=black!100, fill=black!100,  thick] (6,6) circle (0.1);

\filldraw[color=black!100, fill=white!100,  thick] (3,1.5) circle (0.1);
\filldraw[color=black!100, fill=white!100,  thick] (6,4.5) circle (0.1);

\node[below] at (6,0) {$1$};
\node[left] at (0,6) {$1$};
\node[below] at (0,-0.1) {$0$};

\node[above] at (3,7) {$\hat{F}$};

\end{tikzpicture}

\end{center}
    \caption{{Dominant} privacy-preserving recommendation when a $\nicefrac{3}{4}$-binary signal $s_1$ must be kept secret. The left panel depicts the cumulative distribution function $F$ of posteriors induced by the symmetric binary signal $s_1$  matching the state with probability $\nicefrac{3}{4}$. The {dominant} recommendation $s_2^\star$ corresponds to the conjugate distribution $\hat{F}$ depicted in the right panel. We see that $s_2^\star$ is trinary: it is completely uninformative  with probability $\nicefrac{1}{2}$ and fully reveals the state with the complementary chance, inducing the posteriors $0$ or $1$ with equal probabilities.
    \label{fig:disclosure}}
\end{figure}

We provide a simple practical procedure for generating a {dominant} privacy-preserving recommendation $s_2^\star$,  given realizations of $(\omega,s_1)$. We know that $s_1$ and $s_2^{\star}$ induce conjugate belief distributions, and so we can use the general procedure outlined in Figure~\ref{fig:induce_conjugates} to construct  $s_2^\star$ as follows:
\begin{quote}
  \begin{itemize}
   \item Calculate $p(s_1)$, the conditional probability of $\omega=1$ given $s_1$.
   
   \item If $\omega=1$, sample $s_2^\star$ uniformly from the interval $[1-p(s_1),\,1]$.
   \item If $\omega=0$, sample $s_2^\star$ uniformly from the interval $[0,\,1-p(s_1)]$. 
   \end{itemize}
\end{quote}
This procedure yields an  $s_2^\star$ that, conditioned on $s_1$, is distributed uniformly on $[0,1]$, and hence is independent of $s_1$. It is simple to verify that $s_2^\star$ is {dominant} (see the proof of Theorem~\ref{thm:disclosure}).

This procedure can be simplified if the posterior $p(s_1)$  only takes finitely many values, in which case there exists a {dominant} privacy-preserving recommendation that is also finitely valued.   Let $[0,1]=\bigsqcup_{k=0}^K I_k$ be a  partition of the unit interval into subintervals using the values of $p(s_1)$. The belief $p(s_2^\star)$ is constant when $s_2^\star$ ranges within $I_k$. Hence, the constructed {dominant} recommendation $s_2^\star$ with values in $[0,1]$ is equivalent to a signal $t_2^\star\in \{0,\ldots,K\}$ such that $t_2^\star=k$ whenever $s_2^\star\in I_k$. The signal $t_2^\star$ is also a {dominant} privacy-preserving recommendation and takes at most one more value than the number of values  of $p(s_1)$.

Consider the symmetric binary $s_1$ matching $\omega$ with probability $\nicefrac{3}{4}$ from Figure~\ref{fig:disclosure}. A {dominant} privacy-preserving recommendation $s_2^\star$ of the state can be generated as follows. It takes three values,  $\{0,1,2\}$.  If $s_1\ne\omega$,  then  $s_2^\star=2\cdot \omega$,  and if $s_1=\omega$ then $s_2^\star=1$ with probability $\nicefrac{2}{3}$ and  $s_2^\star=2\cdot \omega$ with probability $\nicefrac{1}{3}$. As a result, the realization $s_2^\star=1$ is completely uninformative and $s_2^\star\in \{0,2\}$ completely reveals~$\omega$.  

\paragraph{Information Loss {in Privacy-Preserving Recommendations.}} 
We now discuss how the informativeness of the {dominant} privacy-preserving recommendation $s_2^\star$ depends on the correlation between $s_1$ and the state~$\omega$. {As before, we assume
 $\omega$ is binary.} {Since} $s_1$ can represent the collection of all protected attributes, this question is of interest to a policymaker considering the welfare implications of adding an extra protected attribute, thus increasing the correlation.

Corollary~\ref{prop:comparative} implies that informativeness of $s_1$ and $s_2^\star$ moves in opposite directions with respect to the Blackwell order. However, this comparative static does not quantify the amount of information lost due to the privacy constraint. Suppose the informativeness of a signal $s$ about binary $\omega$ is measured by the mutual information $
    I(\omega;s) = H(p)-\E{H(p(s))},$ where $H$ is the Shannon entropy defined by $H(q) =q\log_2(q)+ (1-q)\log (1-q).$ The mutual information $I(\omega;s)=0$ if $s$ is uninformative and $I(\omega;s)=H(p)$ if $s$ pins down the realization of $\omega$. The following inequality provides an upper bound on the informativeness of the {dominant} privacy-preserving recommendation:
\begin{equation}\label{eq_optimal_disclosure_informativeness}
{I(\omega;s_2^\star)\leq H(p)-I(\omega;s_1).}
\end{equation}
The information-theoretic intuition behind~\eqref{eq_optimal_disclosure_informativeness} is that privacy does not allow $s_1$ and $s_2^\star$ to carry the same ``bit of information'' about $\omega$, and the ``total number of bits'' is~$H(p)$. A counterintuitive phenomenon is that~\eqref{eq_optimal_disclosure_informativeness} is always strict whenever $H(p) > I(\omega;s_1)>0$, i.e.,  $s_1$ is informative about $\omega$ but does not reveal it completely. This effect is discussed in~\S\ref{sec:feasibility}, where we prove a tighter bound in a more general setting. {We  also consider quadratic loss as an alternative measure of informativeness in Appendix~\ref{app_quadratic_informativeness}.} 

The bound in inequality~\eqref{eq_optimal_disclosure_informativeness} does not {speak directly to the utility loss suffered by a decision maker due to the privacy constraint.} 
Suppose 
the {dominant} recommendation~$s_2^\star$ is observed by a decision maker,  who selects an action $a\in A$ and receives a payoff of~$u(\omega,a)$. Such a decision maker may be interested in evaluating how good the signal is in her decision problem. Her expected payoff can be expressed as
$$\E{\sup_{\sigma\colon S_2\to A} u(\omega,\sigma(s_2^\star))}.$$
We assume that $u$ is bounded from above and thus 
 the indirect utility  $U(q) =\sup_{a\in A} \big((1-q)\cdot u(0,a)+ q\cdot u(1,a)\big)$ is a continuous convex function. The following proposition 
provides an explicit formula for the decision-maker's expected payoff.
\begin{proposition}\label{prop_optimal_utility}
Consider the {dominant} privacy-preserving recommendation $s_2^\star$ for $(\omega,s_1)$ with binary $\omega$, and let $F$ be the cumulative distribution function of $p(s_1)$. The expected payoff of a decision maker with an indirect utility function $U$ observing $s_2^\star$ is equal to 
\begin{equation}\label{eq_optimal_utility}
\int_{[0,1]} U(1-F(t))\,\dd t.
\end{equation}
\end{proposition}
The mutual information discussed above corresponds to a particular decision problem with $U(q)=H(p)-H(q)$.

The proposition follows from Theorem~\ref{thm:disclosure}. Indeed, the  distribution of beliefs induced by the {dominant}  recommendation is the conjugate $\hat{F}$ of $F$. Thus 
the expected payoff of the decision maker observing $s_2^\star$ is given by 
$$\int_{[0,1]} U(q)\,\dd \hat{F}(q).$$
Changing the variable $t=\hat{F}(q)$ and using the definition of the conjugate, we get formula~\eqref{eq_optimal_utility}. Details can be found in Appendix~\ref{app_optimal_utility}.

We end this section with two notes.

\paragraph{Three and More States.}
Beyond the binary state case, optimal privacy-preserving recommendations are not unique (see Appendix~\ref{app_non_uniqueness}), and so the choice of the recommendation can depend on the receiver's decision problem. In this case, we do not have a simple characterization or construction and leave this question as a direction for future research.

\paragraph{{Approximate Privacy.}}
{Our notion of privacy is very strong, as it requires complete independence between the signals. In some applications, it might be more natural to impose a softer constraint. For example, one could imagine a recommendation system that tries to maximize the difference between some utility, as described above, and a privacy violation cost, e.g., the mutual information between~$s_1$ and~$s_2$. This could refer to a situation where the designer must pay a fine  if they violate the privacy constraint, and the fine is larger for more egregious violations of privacy. 

This setup is closely related to the rational inattention literature. There, a classical model \citep*{sims2003implications} describes a decision maker who chooses $s_2$ to maximize
\begin{align*}
    \E{U(p(s_2))} - \lambda\cdot  I(\omega;s_2),
\end{align*}
where $U$ is an indirect utility,  and $\lambda > 0$ captures the marginal cost of making the signal $s_2$ more informative. The approximate privacy model we refer to would involve maximizing
\begin{align}\label{eq_costly_privacy}
    \E{U(p(s_2))} - \lambda\cdot  I(s_1;s_2).
\end{align}
Exact privacy corresponds to the limit $\lambda\to +
\infty$. However, our results may shed light on the behavior of this problem even for fixed $\lambda$. For example, in the case of a binary state, it is immediate that the solution to~\eqref{eq_costly_privacy} Blackwell dominates the dominant privacy-preserving recommendation $s_2^\star$ as $I(s_1;s_2^\star)=0$.

Since rational inattention models are very well understood, some of the insights from this literature could potentially be useful in studying approximate privacy.  We leave this as a direction
for future research.}

\section{{Application to Information Design with Privacy Across Receivers}}

{In this section, we consider a multi-receiver interpretation of our privacy constraint. We study information design problems in which the designer supplies information to several receivers and is constrained to using private private information structures. As we illustrate in \S\ref{sec_bidders}, such constraints may arise because the designer faces an environment where compromising agents' privacy leads to undesirable outcomes. In \S\ref{sec_multiagent_inf_design}, we define our constrained information design problem and characterize a family of designer objectives where the optimal information structure lies on the Blackwell-Pareto frontier. In a special case, we give a more explicit characterization of the optimal structure and show that it involves very few signal realizations. Finally, in \S\ref{sec:feasibility}, we provide a general recipe for solving information design problems with privacy across multiple receivers.

\subsection{Motivating Example: Informing Bidders}\label{sec_bidders}

As motivation for our privacy constraint, we outline an example scenario where privacy across receivers is crucial for the information designer's objective.

Early auction design literature has pointed out that, in common value auctions,  the privacy of bidders' information matters more than its accuracy for improving bidders' expected profits. Indeed, \cite*{milgrom1982theory} demonstrate that if the information of one bidder contains the information of another, then the latter bidder gets zero profit. 
\cite*{mcafee1992correlated} generalize these extreme examples to a general Bayesian mechanism design context:
\begin{quote}
\emph{Introducing arbitrarily small amounts of correlation into the joint
distribution of private information among the players is enough to render
private information valueless, in the sense that its possessors earn no rents.}\footnote{{The difference between full surplus extraction by \cite*{mcafee1992correlated} and by \cite*{cremer1988full} is that the former considers correlated signals about a common payoff-relevant state, while the latter considers correlated private values.}} 
\end{quote}

To understand the information design implications of this phenomenon, consider the following example: a platform market where a seller sells a single unit of a good to $n\geq 2$ buyers. The buyers have a positive common value $v(\omega)$ for the good, which is determined by a random state  $\omega\in \{0,1\}$. The realized state $\omega$ is observed by the platform only. Without information about $\omega$, a seller can sell the good to one of the buyers, charging her the expected value $\E{v(\omega)}$. This leaves no surplus to the buyers and no reason for them to stay on the platform ({especially if platform monetization relies on subscription fees)}. Thus, the platform is interested in revealing some information about the state to the buyers to ensure that they get positive surplus.
 What is the buyers' welfare for different information structures $\mathcal{I}=(\omega, s_1,\ldots, s_n)$, provided that the profit-maximizing seller can tailor the mechanism to the platform's choice of~$\mathcal{I}$?

First, observe that sending fully informative signals to all buyers leaves them with no surplus. Indeed, the seller can ask one agent to report her value and sell the good to another agent for a price equal to the reported value. The analogous argument shows that identical (but not necessarily fully informative) signals also lead to full surplus extraction.

Full surplus extraction persists even if there is only a tiny correlation between buyers' signals. Indeed, consider $n=2$ buyers and a family of information structures~$I_\varepsilon$ indexed by $\varepsilon\in[0,\,1/4]$,  given by the following tables:

\begin{center}
    \begin{tabular}{|c||c|c|}
\hline 
$\omega=1$ & $s_{2}=1$ & $s_{2}=0$\tabularnewline
\hline 
\hline 
$s_{1}=1$ & $0.5+2\varepsilon$ & $0.25-\varepsilon$\tabularnewline
\hline 
$s_{1}=0$ & $0.25-\varepsilon$ & $0$\tabularnewline
\hline 
\end{tabular} %
\hskip 1cm
\begin{tabular}{|c||c|c|}
\hline 
$\omega=0$ & $s_{2}=1$ & $s_{2}=0$\tabularnewline
\hline 
\hline 
$s_{1}=1$ & $0$ & $0.25-\varepsilon$\tabularnewline
\hline 
$s_{1}=0$ & $0.25-\varepsilon$ & $0.5+2\varepsilon$\tabularnewline
\hline 
\end{tabular}
\end{center}

When $\varepsilon = 0$, the two signals are independent, and we get a private private structure equivalent to the one in Figure~\ref{fig:34}. When $\varepsilon = \nicefrac{1}{4}$, both signals are fully informative. For any $\varepsilon > 0$, the two signals are correlated, and it turns out that the seller has a profit-maximizing mechanism that fully extracts the buyers' surplus. Inspired by \cite*{cremer1988full}, the idea is that the seller asks the buyers to report their signals.  The seller gives both buyers a large bonus if the reports match and charges both a large penalty if the reports mismatch, so that each type's expected transfer is zero when reporting their signal truthfully but sufficiently negative when lying. The seller then calculates the expected value of the good given the two reports and sells it to the first buyer for this amount, extracting full surplus. Thus any $\varepsilon>0$ leaves zero surplus to the buyers. For $\varepsilon=0$, the buyers' surplus is positive since the seller cannot use one buyer to learn about the other buyer's signal, and thus, the buyers retain information  rents.

{A related phenomenon underlies the results of \cite*{bergemann2016informationally}, \cite{du2018robust}, and \cite{brooks2021optimal}, who study worst-case revenue 
maximization in common-value environments, where the worst case is taken over the buyers' possible 
information structures. One of the key insights from these papers is that the worst case arises when buyers receive private private signals about the common value. The intuition is 
that a version of the min-max theorem applies in this environment \citep{brooks2020strong}, allowing the analyst to assume that the adversary moves first, choosing the information structure, and the designer best-responds. If the adversary were to choose a correlated 
information structure, the designer could extract full surplus---just as in our earlier 
example. Hence, without loss of generality, the analysis can focus on private private structures.}

\subsection{{An Information Design Problem with a Privacy Constraint}}\label{sec_multiagent_inf_design}

{More generally, in environments where compromising agents' privacy leads to undesirable outcomes (such as full surplus extraction),
the information designer is effectively restricted to maximizing her objective over private private structures. This consideration motivates a class of information design problems, where the designer aims to maximize some objective under a privacy constraint across agents.

Consider a stylized information-design problem under a privacy constraint across receivers. There are $n$ receivers $i \in \{1,2,\ldots,n\}$ who have common prior $p$ for a state $\omega\in \Omega=\{1,\ldots, m\}$ and observe signals $(s_1,\ldots,s_n)$. The designer's goal is to maximize
\begin{equation}\label{eq_persuasion_n_receivers}
\E{U\big(p(s_1),p(s_2),\ldots, p(s_n)\big)}
\end{equation}
over all private private information structures $(\omega,s_1,s_2,\ldots, s_n)$. Here, $U$ is the designer's utility depending on the profile of realized posterior beliefs. 

For example, utilities of this form arise if each receiver has a decision problem to solve, and the designer aims to maximize the social welfare. Suppose each receiver has to choose an action $a_i \in A_i$ after observing a signal $s_i$ and receives a payoff $u_i(\omega,a_i)$. The social welfare is given by 
\begin{align}\label{eq_social_welfare}
\E{\sum_{i}\sup_{\sigma_{i}\colon S_{i}\to A_{i}}u_{i}(\omega,\sigma_{i}(s_{i}))},
\end{align}
which can be rewritten in the form~\eqref{eq_persuasion_n_receivers} with separable $U$
\begin{equation}\label{eq_U_separable}
U(q_1,\ldots, q_n)=U_1(q_1)+\ldots+U_n(q_n),
\end{equation}
with 
\begin{equation}\label{eq_welfare_U_i}
   U_i(q)=\sup_{a_i\in A_{i}} \  \left(\sum_{k\in \Omega}  q_k \cdot u_{i}(k,a_i)\right).
\end{equation}
Note that the utility $U_i$ defined by~\eqref{eq_welfare_U_i}---as the upper envelope of affine functions of~$q$---is convex and continuous for any $u_i$ bounded from above.
The following lemma demonstrates that the convexity of $U$ plays an important role as it allows the designer to focus on Blackwell-Pareto optimal private private structures. Interestingly, convexity in a single belief $q_i$ is enough, since the designer can always improve her objective by giving the residual informativeness to  agent $i$.\footnote{{If $U$ is convex in each argument, the designer aims to make each signal as informative as possible under the privacy constraint. Contrast this with the secret sharing of} \cite*{shamir1979share}, where the goal is to split information in a way that each agent knows as little as possible both about others and about $\omega$.} The following result is proved in  Appendix~\ref{sec_proof_lm_convex_U}.
\begin{lemma}\label{lm_convex_U}
{Let $U$ be an upper semicontinuous function of $(q_1,\ldots, q_n)$, and suppose there is at least one agent $i$ such that $U$ is convex in $q_i$. Then, the optimal value of~\eqref{eq_persuasion_n_receivers} is attained at a Blackwell-Pareto optimal information structure.}  
\end{lemma}
{The idea behind the proof is the following: if a private private information structure optimizes $U$ but is not Pareto optimal, we can give more information to agent $i$, without hurting the objective.  

We will now discuss the most tractable binary-state two-receiver setting with a separable $U$, before turning to    the general case in \S\ref{sec:feasibility}.} 

\paragraph{Binary State, Two Receivers, and Separable Objective.}
{Consider a binary~$\omega$,  two receivers, and $U(p_1,p_2)=U_1(p_1)+U_2(p_2)$. We assume that at least one of $U_i$ is convex. By Lemma~\ref{lm_convex_U}, one can look for the optimal structures on the Blackwell-Pareto frontier characterized in Theorem~\ref{thm:binary-pareto}.
While the frontier contains a rich set of information structures, including some that induce a continuum of beliefs, the optimal ones have a simpler form.}
\begin{proposition}
\label{prop:ladder}
  {Given a binary $\omega$ and continuous $U(p_1,p_2)=U_1(p_1)+U_2(p_2)$ with convex $U_2$, there exists an optimal private private information structure $(\omega,s_1,s_2)$ such that $s_1$ takes at most two values, $s_2$ takes at most three values, and the distributions of beliefs induced by $s_1$ and $s_2$ are conjugates.
  Furthermore, if both $U_1$ and $U_2$ are convex, there exists such a structure where both $s_1$ and $s_2$ are binary. 
  }
\end{proposition}

{Recall that in the classical single-agent  model of \cite*{kamenica2011bayesian}, signals with two possible values are enough for optimal persuasion. By Proposition~\ref{prop:ladder}, to persuade a pair of receivers under our privacy constraint, we may need three. On the other hand, the fact that only three suffice  is surprising, since in multi-agent persuasion without privacy, there is no such bound on the number of signal values, i.e., there are objectives requiring arbitrarily many signals
\citep{arieli2020feasible, cichomski2023existence}.}
The proposition is proved in Appendix~\ref{app_welfare} using a combination of an extreme-point argument and the characterization of Blackwell-Pareto optimal structures via conjugate distributions (Theorem~\ref{thm:binary-pareto}).

{To illustrate Proposition~\ref{prop:ladder}, we apply it to a simple example of social welfare maximization~\eqref{eq_social_welfare}. There are  two equally likely states, $A_i = \Omega = \{0,1\}$,  and each agent gets utility 1 from matching the state and utility -1 from mismatching it, so that $u_1(\omega,a)=u_2(\omega,a)=1-2|\omega-a|$.} If we reveal the state to agent~1 and give agent~2 no information, then the social welfare is 1. Consider instead a private private information structure where each agent has a posterior belief of $\sqrt{\nicefrac{1}{2}}$ with probability~$\sqrt{\nicefrac{1}{2}}$ and a posterior belief of 0 with the complementary probability. Such a structure exists as this distribution is its own conjugate: see also Figure~\ref{fig:social-welfare}.
Then the social welfare is $4-2\sqrt{2}\approx 1.17$. Let us check that this is the highest possible social welfare  across all private private information structures. 

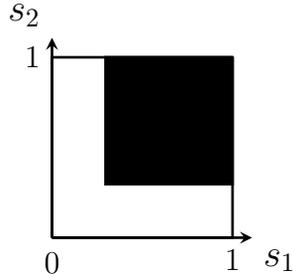
\begin{figure}
\begin{center}
\begin{tikzpicture}[scale=0.4, line width=1pt]
\draw (0,0) -- (6,0)--(6,6)--(0,6)--(0,0);

\draw [->,>=stealth] (0,0) -- (6.65,0);
\draw [->,>=stealth] (0,0) -- (0,6.65);
\node[below right] at (6.65,0) {\large $s_1$};
\node[above left] at (0,6.65) {\large $s_2$}; 

\node[below] at (6,0) {$1$};
\node[left] at (0,6) {$1$};
\node[below] at (0,-0.1) {$0$};

\filldraw[black] (6,6)--(6,1.78)--(1.78,1.78)--(1.78,6)--(6,6);

\end{tikzpicture}
\end{center}
\caption{A social welfare-maximizing private private information structure for the decision problem in which $u_1(\omega,a)=u_2(\omega,a)=1-2|\omega-a|$ .\label{fig:social-welfare} 
}

\end{figure}

By Proposition~\ref{prop:ladder}, we can assume that the distribution of posteriors  $\mu$ induced by $s_1$ is supported on two points. It  has mean $\nicefrac{1}{2}$ since the average posterior equals the prior, and thus can be represented as $\frac{\alpha}{\alpha+\beta}\delta_{\frac{1}{2}-\beta}+\frac{\beta}{\alpha+\beta}\delta_{\frac{1}{2}+\alpha}$ for some constants $\alpha,\beta\in (0,\nicefrac{1}{2}]$, where $\delta_x$ denotes the point mass at $x$. The contribution of the first agent to the welfare is, therefore, $\frac{4\alpha\beta}{\alpha+\beta}$.

The conjugate distribution $\hat{\mu}$ takes the form $\left(\frac{1}{2}-\alpha\right)\delta_0+(\alpha+\beta)\delta_\frac{\beta}{\alpha+\beta}+\left(\frac{1}{2}-\beta\right)\delta_1$. As the problem is state-symmetric, we can assume $\beta\geq\alpha$ without loss of generality and, hence, the middle atom of $\hat{\mu}$ is above $\nicefrac{1}{2}$. Therefore, the second agent contributes $1-2\alpha$ to the welfare, and the total welfare equals 
$\frac{4\alpha\beta}{\alpha+\beta}+1-2\alpha.$ A simple calculation shows that this is maximized when $\beta=\nicefrac{1}{2}$ and $\alpha=\sqrt{\nicefrac{1}{{2}}}-\nicefrac{1}{2}$, which yields the structure described above.

\medskip

\subsection{{A Feasibility Perspective on 
Privacy-Constrained Information Design}}
\label{sec:feasibility}

{Consider the general information-design problem~\eqref{eq_persuasion_n_receivers} where the designer is restricted to private private information structures. So far, we have focused on the case where the designer's objective is maximized at a Blackwell-Pareto optimal structure. However, if the conditions of Lemma~\ref{lm_convex_U} are not satisfied, an optimal structure may not lie on the Blackwell-Pareto frontier. In this case, it is convenient to take the belief-based approach of \cite*{kamenica2011bayesian}, 
identifying information structures with the  belief distributions that they can induce and
rewriting the designer's problem as the maximization over these distributions.} 

\begin{definition}\label{def:feasible_beliefs}
An $n$-tuple $(\mu_1,\ldots,\mu_n)$ of probability measures on $\Delta(\Omega)$ is said to be \emph{feasible under the constraint of privacy} if there exists a private private information structure $\mathcal{I}=(\omega,s_1,\ldots,s_n)$ such that $\mu_i$ is the distribution of $p(s_i)$. 
\end{definition}
The set of all distributions feasible under the constraint of privacy for a given prior $p\in \Delta(\Omega)$ and $n$ agents is denoted by~$\F_n(p)\subset \Delta(\Omega)\times \ldots \times \Delta(\Omega)$. 
For example, Figure~\ref{fig:34} shows that for symmetric binary states and two agents, it is feasible for both agents to have binary signals that induce beliefs of either $\nicefrac{1}{4}$ or $\nicefrac{3}{4}$. That is,  $(\mu_1,\mu_2)\in \F_2(1/2)$ for
\begin{align*}
    \mu_1=\mu_2=\frac{1}{2}\delta_{1/4}+\frac{1}{2}\delta_{3/4}.
\end{align*}
{The designer's objective~\eqref{eq_persuasion_n_receivers} depends on the information structure only through the belief distributions that it induces. Accordingly, the problem can be rewritten as the maximization of 
$$\int_{\Delta(\Omega)\times \ldots\times \Delta(\Omega)} U(q_1,\ldots, q_n)\ \dd\mu_1(q_1)\ldots \dd\mu_n(q_n) $$
over $(\mu_1,\ldots, \mu_n)\in \F_n(p)$. To make this perspective useful, one needs to characterize the set of privacy-constrained feasible distributions $\F_n(p)$, and we focus on this question for the remainder of this section.} 

\medskip

We first give a simple characterization of the set of {privacy-constrained feasible distributions} in the special case of two agents and a binary state, and then discuss what happens in the general case. 

A simple necessary condition for feasibility is given by the  so-called martingale condition, i.e., by the law of iterated expectations. It implies that if the posterior $p(s_i)$ has distribution $\mu_i$ then the expected posterior $\int q\,\dd\mu_i(q)$ must equal to the prior distribution of $\omega$. Thus a necessary condition for $(\mu_1,\ldots, \mu_n)\in \F_n(p)$ is that $\int q\,\dd\mu_i(q) = \int q\,\dd\mu_j(q)$ for all~$i$ and~$j$. 

To characterize feasible distributions, we relate them to 
Blackwell-Pareto optimality. By Lemma~\ref{lem:pareto-dominated} from the Appendix, any private private information structure is  weakly dominated by a Blackwell-Pareto optimal one. Thus $(\mu_1,\ldots,\mu_n)\in \F_n(p)$ for some prior $p$ if and only if there exists a Blackwell-Pareto optimal structure represented by some $(\nu_1,\ldots,\nu_n)$, such that each $\mu_i$ is a mean-preserving contraction of~$\nu_i$. This holds since mean-preserving contractions of the posterior belief distributions correspond to Blackwell dominance. By Blackwell's theorem,  one can take a structure with posteriors $(\nu_1,\ldots,\nu_n)$, and apply an independent garbling to {each signal} to arrive at a structure with posteriors $(\mu_1,\ldots,\mu_n)$.

This observation, together
with Theorem~\ref{thm:binary-pareto}, gives the following characterization for the case of a binary state and two {signals}.
\begin{corollary}
\label{cor:feasible}
For a binary state and two {signals}, a pair $(\mu_1,\mu_2)$ of distributions on $[0,1]$ {belongs to $\F_2(p)$} if and only if  $\mu_2$ is a mean preserving contraction of the conjugate  of~$\mu_1$ {and $p=\int_0^1 q\, \dd\mu_1(q)=\int_0^1 q\, \dd\mu_2(q)$.}\footnote{{For conjugates, $\int_0^1 q \dd\mu_1(q)=\int_0^1 q \dd\mu_2(q)$, so it is enough to impose this last condition on~$\mu_1$.}}
\end{corollary}
Corollary~\ref{cor:feasible} provides a simple tool for checking feasibility {in terms of the cumulative distribution functions of the beliefs}. Indeed, by applying a standard characterization of mean-preserving spreads, the pair of distributions $(\mu_1,\mu_2)$ belongs to $\F_2(p)$ if and only if they have the same expectation~$p$, and the corresponding cumulative distribution functions $(F_1,F_2)$ satisfy 
\begin{align*}
    \int_y^1 F_2(x)\,\dd x \geq \int_y^1 \hat F_1(x)\,\dd x
\end{align*}
for every $y \in [0,1]$. 

{This corollary generalizes partial results on privacy-constrained feasibility obtained by \cite{gutmann1991existence} and  \cite{arieli2020feasible}.
For example, Proposition~2 of \cite{arieli2020feasible}  focuses on $\mu_1=\mu_2=\mu$ being symmetric to reflection around the prior $p=\nicefrac{1}{2}$  and demonstrates that $(\mu,\mu)\in  \F_2(1/2)$ if and only if $\mu$ is a mean-preserving contraction of the uniform distribution on $[0,1]$. 
In Corollary~\ref{cor:feasible}, the two states need not be equally likely, the two signals need not induce the same belief distribution, and the belief distributions need not be symmetric around $\nicefrac{1}{2}$.}

\smallskip

For $n\geq 2$ {signals}, the characterization of Blackwell-Pareto optimal private private structures via sets of uniqueness (Theorem~\ref{thm:binary-pareto}) implies the following characterization of feasible distributions. 

\begin{corollary}\label{cor_feasible_n_geq_3}
For $n\geq 2$ {signals}, a collection  $(\mu_1,\ldots, \mu_n)$ of distributions on $[0,1]$ {belongs to $\F_n(p)$} if and only if there exists a set of uniqueness $A\subset [0,1]^n$ {of Lebesgue measure~$p$} such that each $\mu_i$ is a mean-preserving contraction of 
{the belief distribution 
induced by the signal $s_i$ in the private private information structure $(\omega,s_1,\ldots,s_n)$ associated with $A$.}
\end{corollary}

{Note that one can replace sets of uniqueness $A\subset [0,1]^n$ in Corollary~\ref{cor_feasible_n_geq_3} with 
any other bigger collection of sets and the ``if and only if'' statement will continue to hold. Indeed, for any set $A\subset [0,1]^n$, the information structure associated with $A$ gives rise to a feasible belief distribution. Thus by enlarging the family of $A$, we do not get any new elements beyond $\F_n(p)$.
Accordingly, one can use the collection of all upward-closed sets $A$  instead of sets of uniqueness and obtain a tractable representation of all privacy-constrained feasible distributions. In the context of the privacy-constrained information design problem, this observation reduces the choice of information structure to choosing an upward-closed set and, possibly, individual garbling.}

Corollary~\ref{cor_feasible_n_geq_3} admits a straightforward extension to the general case of $n\geq 2$ {signals} and $m\geq 2$ states  by replacing sets of uniqueness with partitions of uniqueness and relying on Theorem~\ref{thm:pareto_appendix} from Appendix~\ref{app:Pareto_via_uniqueness} instead of Theorem~\ref{thm:binary-pareto}.

\section{{Applications to Causal Strength and Responsiveness of Social Choice Rules}}\label{sec_responsiveness}

{In this section, we consider applications of private private information structures to causal inference and mechanism design. In both applications, we aim to understand how strongly independent inputs can affect a common outcome.

\subsection{Causal Strength}
In causal inference, a \emph{collider} is one of the basic causal structures. It represents a situation where multiple independent causes $s_1,\ldots, s_n$ produce an effect $\omega$ \citep[see, e.g.,][]{pearl2009causality}.\footnote{{In applied work, causes are often correlated. In this case,  colliders are often applied to orthogonalized causes, obtained via a suitable version of the principal component analysis.}}
Consider, for example, two causes $s_1,s_2$---say, nature (DNA) and nurture (socioeconomic status)---that produce an effect $\omega$---say, the winning of an Olympic medal. Assuming that $s_1$ and $s_2$ are independent, $(\omega,s_1,s_2)$ is a private private structure.

More generally, any collider $(\omega, s_1,\ldots, s_n)$ can be thought of as a private private information structure, with the causes  interpreted as signals and the effect interpreted as the state. Similarly to private private signals, which cannot be simultaneously all very informative,
it is impossible for all causes  to simultaneously influence the effect~$\omega$ too strongly. How strong can these causal strengths be?

Following the causal inference literature, we use mutual information to quantify causes' strength  \citep{janzing2013quantifying}. Recall that 
the} \emph{Shannon entropy} of a measure $q \in \Delta(\Omega)$ is 
\begin{align}
\label{eq:H}
    H(q) = -\sum_{k \in \Omega}q(k)\log_2(q(k)).
\end{align}
The \emph{mutual information} between $\omega$ and $s_i$ is given by
\begin{align}\label{eq:I}
    I(\omega;s_i) = H\left(\E{p(s_i)}\right)-\E{H(p(s_i))}.
\end{align}
{In information theory, entropy is often used to quantify the amount of randomness in a distribution. Mutual information is then the expected reduction in this randomness, and is  used as a measure for the amount of information contained in a signal. These notions are also used in economics, e.g., in the rational inattention literature \citep*{sims2010rational, matvejka2015rational}.

We can bound the causal strengths of different causes in a collider by bounding the feasible profiles of mutual information  $I(\omega;s_i)$, $i=1,\ldots, n$, in private private information structures. Before turning to these mutual information bounds, we will discuss another context motivating these results.} 

\subsection{{Responsiveness of Social Choice Rules}}\label{sec_responsive}
{To elicit information from an agent, it is often important for the agent to know that her report matters: that is, the social outcome is sufficiently responsive to her report. For instance, an election in a large society where each ballot has a very low chance of being pivotal provides little incentive for citizens to participate; see, e.g., \cite*{chamberlain1981note}.

We consider a general setting where a social choice rule aggregates agents' reported types into a public outcome. We ask how responsive this rule can be to individual reports.\footnote{
{We focus on rules where each individual agent exerts maximal influence on the outcome. 
In this sense, our objective is in direct contrast to differential privacy 
\citep{dwork2014algorithmic}, which requires that changes in each agent's input 
have only negligible impact on the public outcome, and to contextual privacy 
\citep{haupt2021contextually}, which aims to minimize the information revealed beyond 
what is strictly necessary to compute that outcome.}
}

Consider an environment with $n$ agents in a society, where agent $i$ has a type~$t_i$ drawn from a set $T_i$ according to a known distribution $\nu_i$, independently across agents. A public outcome is chosen from a finite set $\Omega$ via a social choice rule ${C}\colon T_1 \times \ldots\times T_n \to \Delta(\Omega)$ that maps the realized type profile to a possibly random outcome. For example, $t_i\in \R_+^m$ may equal $i$'s value for each of $m$ candidates, and $\omega\in \{1,\ldots, m\}$ determines which one gets elected. Or, $t_i\in \R$ may be the agent's value for a good, and $\omega\in \{1,\ldots, n\}$ specifies who gets it. We ignore incentive issues and 
assume that agents report their true types to ${C}$. So, our results provide an upper bound on feasible levels of responsiveness.

We measure the responsiveness of a social choice rule ${C}$ to agent $i$'s report by the mutual information $I(\omega;t_i)$; see~\eqref{eq:I}. In other words,  responsiveness captures the amount of information about the public outcome that is contained in an individual's report. For example, in an election between two candidates, our measure of responsiveness is related to \emph{pivotality}, defined as the difference in the winning probability of the first candidate when an individual voter switches her vote.
In more general environments, responsiveness measures the average influence of an agent on the outcome across different realizations of her type. By contrast, a natural extension of pivotality to general type and outcome spaces looks at the maximal influence of an agent on the outcome across all types of the agent, so that an agent is deemed influential even if they are only decisive for the outcome for a rare type realization \citep*{al2000pivotal}.

A dictatorship is an example of a social choice rule highly responsive to a single agent's type. Intuitively,  no rule can be highly responsive to all agents simultaneously. We formalize this intuition by uncovering constraints that the profile of responsiveness ~$I(\omega;t_i),$ $i=1,\ldots, n$, must satisfy.

The key observation is that any social choice rule defines a private private information structure. Indeed, consider type spaces $T_i$ for $i=1,\ldots,n$,  type distributions  $\nu_i\in \Delta(T_i) $ , and a social choice rule ${C}\colon T_1\times \ldots \times T_n\to \Delta(\Omega)$. This data defines a joint distribution of the public outcome $\omega$ and types $t_1,\ldots, t_n$. Interpreting types as signals and the outcome as the state, we obtain that $(\omega, t_1,\ldots, t_n)$ is a private private information structure, since the types of different agents are independent. Moreover, responsiveness of ${C}$ to agent $i$ coincides with the informativeness of the signal $t_i$ about $\omega$. Thus understanding feasible profiles of responsiveness---similarly to causal strength in colliders---boils down to the question of signal informativeness in private private structures.} 

\subsection{{Bounds on Signal Informativeness}}\label{sec_bounds_informativeness}
{Consider a private private information structure $(\omega,s_1,\ldots, s_n)$. Depending on the application, $s_1,\ldots, s_n$ and $\omega$ can be interpreted as independent causes and the effect they produce or as individual types and the outcome of a social choice rule. Our goal is to understand how informative $s_1,\ldots, s_n$ can be about $\omega$ as measured by the mutual information $I(\omega;s_i)$, thus obtaining bounds on causal strength and the social choice rule's responsiveness in these settings.}

\begin{proposition}
\label{prop:mutual-info}
{For any finite $\Omega$ 
and a private private structure $(\omega, s_1,\ldots, s_n)$, 
 $$\sum_i I(\omega; s_i) \leq H(p),$$
 where $p\in \Delta(\Omega)$ is the prior distribution of $\omega$.}
\end{proposition}

{In the context of causal inference, the proposition implies that the effect can depend strongly only on a few independent causes. Similarly, a social choice rule can only be strongly responsive to  a few agents' reports. In particular, if the social choice rule treats all agents in the same way---so  $I(\omega;s_i)$ has the same value for each $i$---the responsiveness to each agent vanishes at the rate of $1/n$ as the size $n$ of the population gets large.}

The fact that {$I(\omega;s_i) \leq H(p)$  for each $i$} follows immediately from the definition of mutual information. For general information structures (e.g., conditionally independent signals), there are no further restrictions on the tuple $(I(\omega;s_1),\ldots,I(\omega;s_n))$: each can take any value between $0$ and $H(p)$. Proposition~\ref{prop:mutual-info} shows that the situation is different when it comes to private private information structures. Here, the sum of mutual information is bounded by the entropy of the prior over $\omega$, so that the entropy of $\omega$ behaves like a finite resource that needs to be split between the {signals}.  The proof of this proposition  uses standard information-theoretic tools and is contained in Appendix~\ref{app:mutual-info}. 

{Note that $I(\omega;s_i)$ can be expressed in terms of the prior distribution $p$ and the distribution of posteriors $\mu_i$ induced by $s_i$,
\begin{align}
\label{eq:I:measure}
I(\omega;s_i) = H\left(p\right)-\int H(q)\,\dd\mu_i(q),
\end{align}
 and so it is a Blackwell-equivalence invariant; see Definition~\ref{def_dominance}.
Thus Proposition~\ref{prop:mutual-info} provides an easily-verifiable necessary condition  for the feasibility of the tuple $(\mu_1,\ldots, \mu_n)$ of posterior distributions.\footnote{{Replacing the feasibility constraint with the more tractable constraint from Proposition~\ref{prop:mutual-info} provides an upper bound on a designer's objective in the privacy-constrained multi-agent information design setting discussed in \S\ref{sec:feasibility} and in the mechanism design setting from \S\ref{sec_border}. One can also use the proposition to easily show that some distributions are not feasible under the constraint of privacy,  e.g., the uniform distribution over the $4$-dimensional cube.} 
}
}

{Proposition~\ref{prop:mutual-info} suggests that one can think of private private structures as a way to divide an ``information pie'' among the signals. However, this analogy happens to be incomplete since there are negative externalities, making the sum of the pieces smaller than the whole pie. We show this in the following strengthening of Proposition~\ref{prop:mutual-info} for a binary state.} 
\begin{proposition}
\label{prop:mutual-info-two}
For a binary state and any private private structure $(\omega, s_1,\ldots, s_n)$, 
\begin{align*}
    \sum_i I(\omega; s_i) \leq H(p) -c_p\cdot \sum_{i < j}I(\omega; s_i)\cdot I(\omega; s_j),
\end{align*}
{where $p\in [0,1]$ is the prior of $\omega=1$ and $c_p=2\ln 2\cdot p^2(1-p)^2$.}
\end{proposition}
{This proposition is proved in Appendix~\ref{app:mutual-info-two}.
The proposition shows that, for a binary state, entropy is a finite resource that is costly to divide: the sum of mutual information is  strictly less than the entropy of $\omega$, as long as at least two signals are informative. Thus, one can make the pie-cutting analogy more precise by assuming that the pie is cut using a blunt knife.
Beyond binary $\omega$, the impossibility of perfect division   depends on the structure of the state space---perfect division becomes possible if the state space has a product structure, which requires $|\Omega|$ to be a composite number. For example, if $\omega$ is uniformly distributed over $\Omega=\{0,1\}\times\{0,1,2\}$, then the structure in which $s_1$ is equal to the first coordinate of $\omega$ and $s_2$ is equal to the second satisfies $I(\omega; s_1)+ I(\omega; s_2)=H(p)$, and both signals are informative about~$\omega$.} 

\medskip
{So far, we have relied on entropy as a measure of uncertainty and defined informativeness as uncertainty reduction. By using a different measure of uncertainty, we can obtain a different measure of informativeness, which would correspond to another way of measuring causal strength or the social choice rule's responsiveness. For example, another popular measure of uncertainty is variance, which leads to the expected reduction in variance as the measure of informativeness. In Appendix~\ref{app_quadratic_informativeness}, we show an analog of Proposition~\ref{prop:mutual-info} for this informativeness measure.}

\section{{Application to Feasibility of Reduced-Form Social Choice Rules}}
\label{sec_border}

{In this section, we provide an application of private private signals to feasibility questions in Bayesian mechanism design. 
In this literature, the one-agent marginals of a multi-agent mechanism are known as its reduced forms.} Pioneered by \cite*{matthews1984implementability}, \cite*{maskin1984optimal}, and \cite*{border1991implementation}, the reduced-form approach {reduces the multi-agent problem to a single-agent problem with additional feasibility constraints}. 

{We discuss a version of the reduced-form approach for public decision making  in a general setting where a social choice rule specifies a public outcome as a function of agents' independent types. We describe an equivalence between the feasibility of reduced-form rules and the feasibility of private private belief distributions. Thus, our results on privacy-constrained feasible belief distributions---see \S\ref{sec:feasibility} and \S\ref{sec_bounds_informativeness}---also speak to the feasibility of reduced-form rules.} This connection highlights the link between feasible reduced forms and partitions of uniqueness. It also {indicates the possible relevance} of information-theory concepts such as Shannon's entropy to mechanism design. These insights contribute to a diverse list of known connections of the reduced-form feasibility problem to majorization \citep*{hart2015implementation,kleiner2021extreme, nikzad2022constrained},
support functions of convex sets \citep*{goeree2016reduced,goeree2023geometric}, network flows \citep*{che2013generalized,zheng2021reduced}, and  polymatroids \citep*{vohra2011mechanism,valenzuela2022greedy,lang2022reduced}. 

\medskip

{Formally, consider the Bayesian mechanism design setting as discussed in \S\ref{sec_responsive} in the context of social choice rule responsiveness: There are $n$ agents and a set~$\Omega$ of~$m$ public outcomes.} Each agent $i=1,\ldots, n$ has a type $t_i\in T_i$, where the sets of types $T_i$ are arbitrary measurable spaces, and $t_i$ are random draws from $\nu_i\in \Delta(T_i)$, independent across agents. {To simplify the statements, we assume that each $\nu_i$ is non-atomic.}\footnote{{Our results can be applied to atomic type distributions without loss of generality, by augmenting each agent's type with an auxiliary dimension uniformly distributed on $[0,1]$ independently of the rest of the variables.}}
For example, $t_i\in \R_+^m$ may equal $i$'s value for each of $m$ public projects, and $\omega\in \{1,\ldots, m\}$ determines which project gets implemented.

{A social choice rule ${C}$} selects a randomized outcome for each profile of types,\footnote{One could augment a mechanism by a vector of transfers, but they are irrelevant to the question of feasibility; see, e.g., \cite*{gopalan2018public}.} i.e., ${C}\colon T_1\times\ldots\times T_n\to\Delta(\Omega)$. 
As each agent $i$ only knows her own type, the relevant information for $i$  about ${C}$ is captured by the reduced-form rule ${C}_i\colon T_i\to\Delta(\Omega)$ obtained by averaging of ${C}$ over $t_{-i}$, i.e.,
$${C}_i(t_i)=\int_{T_{-i}} {C}(t_i,t_{-i})\,\dd \nu_{-i}(t_{-i}).$$
The idea dating back to \cite*{maskin1984optimal} and \cite*{matthews1984implementability} is that constraints of Bayesian incentive compatibility, interim individual rationality, and revenue or welfare objectives can all be reformulated in terms of the collection of reduced forms $({C}_1,\ldots, {C}_n)$ rather than ${C}$. {Thus, a multi-agent mechanism design problem boils down to a collection of single-agent ones if we know which reduced forms are feasible.} 

\begin{definition} Single-agent rules $({C}_1,\ldots,{C}_n)$ with ${C}_i\colon T_i\to \Delta(\Omega)$ are \emph{feasible reduced forms} for given type distributions $\nu_i\in \Delta(T_i)$ if there is {an $n$-agent social choice rule ${C}$} such that ${C}_i$ are its reduced forms.
\end{definition}

The feasibility of reduced-form rules turns out to be tightly related to the feasibility of belief distributions induced  by private private information structures, as studied in \S\ref{sec:feasibility}. 
Recall that an $n$-tuple  of distributions $(\mu_1,\ldots, \mu_n)$ is \emph{feasible under the constraint of privacy}  if there is a private private 
information structure $(\omega,s_1,\ldots,s_n)$ such that the  belief $p(s_i)$ is distributed according to $\mu_i$ for each $i=1,\ldots,n$; see Definition~\ref{def:feasible_beliefs}.
\begin{claim}\label{claim_feasible_reduced_forms}
Let ${C}_i\colon T_i\to \Delta(\Omega)$, $i=1,\ldots, n$,
be a collection of single-agent rules, and each agent's type $t_i\in T_i$ be distributed according to non-atomic
$\nu_i\in\Delta(T_i)$. Denote the distribution of ${C}_i(t_i)$ by $\mu_i$. Then $({C}_1,\ldots, {C}_n)$ are feasible reduced forms if and only if {distributions $(\mu_1,\ldots,\mu_n)$ are feasible under the constraint of privacy .}
\end{claim}
{The claim is proved in Appendix~\ref{app:proof_reduced_forms}. The idea is that the agents' types can be interpreted as ``signals'' about the outcome.} 

\smallskip

{The connection to private private signals provides structural insights about mechanisms.
By Claim~\ref{claim_feasible_reduced_forms}, any social choice rule ${C}$ can be interpreted as a private private information structure, where ${C}_i(t_i)$ plays a role of the posterior of agent $i$. Proposition~\ref{prop:associated} implies that for any private private information structure with non-atomic distributions of signals, there is an equivalent one where the signal spaces are the same and the posteriors are the same after each signal realization, but the state is a deterministic function of signals.

\begin{corollary}\label{cor_derandom}
For non-atomic distributions of types $\nu_i\in \Delta(T_i)$ and a social choice rule ${C}\colon T_1\times \ldots \times T_n\to\Delta(\Omega)$ with a finite set of outcomes $\Omega$, there exists another rule ${C}'$ with the same reduced forms ${C}_i'={C}_i$ and such that $\omega$ is a deterministic function of types.
\end{corollary}

This equivalence of deterministic and stochastic mechanisms has recently been established by \cite*{chen2019equivalence}. We show that their main result can alternatively be obtained from our analysis of private private signals. In particular, Corollary~\ref{cor_derandom} implies that a designer maximizing revenue, welfare, or any other objective that can be expressed through reduced forms can focus on deterministic mechanisms without loss of generality. However, it is important to make sure that agents' types contain enough randomness for the derandomization result to apply. For example,   it is known that there is a revenue gap between deterministic and stochastic multi-good auctions with atomic value distributions \citep*{chawla2010power}.

We note that our setting also captures auction design, which has been the main focus of the literature on the reduced-form approach. Indeed, each $t_i\in \R$ may represent the agent's value for a private good while $\omega\in \{1,\ldots, n\}$ specifies who gets it. However, in private good allocation, each agent $i$ typically cares only about her own probability of getting the good, i.e., about the $i$'s component ${C}_i(t_i)(i)\in [0,1]$ rather than the whole vector ${C}_i(t_i)\in \Delta(\Omega)$. Accordingly, it is common to call ${C}_i(t_i)(i)$ the reduced form of ${C}$ for agent~$i$ in the context of private-good allocation. A version of Claim~\ref{claim_feasible_reduced_forms} in that context has appeared in concurrent papers by~\cite*{lang2022feasible,lang2023belief}. 

Claim~\ref{claim_feasible_reduced_forms} allows one to use our feasibility results about belief distributions 
---criteria from Corollaries~\ref{cor:feasible} and~\ref{cor_feasible_n_geq_3} and necessary conditions 
 from Propositions~\ref{prop:mutual-info}, \ref{prop:mutual-info-two}, and~\ref{prop:quadartic-info}---to make statements about the feasibility of reduced-form mechanisms.
For example, taking into account that one can replace the sets of uniqueness in Corollary~\ref{cor_feasible_n_geq_3} with a broader family of upward-closed sets, we get the following result.
\begin{corollary}\label{cor_feasible_reduced_forms_binary}
For a binary outcome space $\Omega=\{0,1\}$,  single-agent rules  $({C}_1,\ldots, {C}_n)$ are feasible reduced forms if and only if there exists an upward-closed set $A\subset [0,1]^n$ such that the distribution of each ${C}_i(t_i)$ is a mean-preserving contraction of 
the distribution 
of beliefs induced by signal $s_i$
in the private private information structure $(\omega, s_1,\ldots,s_n)$ associated with~$A$.
\end{corollary}

Similarly to Corollary~\ref{cor_feasible_n_geq_3}, this result admits a straightforward extension for arbitrary finite outcome spaces $\Omega$, where the role of the set $A$ will be played by a partition of uniqueness. These general feasibility criteria 
resemble the majorization requirements obtained by \cite*{hart2015implementation} and \cite*{kleiner2021extreme} for symmetric auctions and can be thought of as an extension of these results to public decisions and non-symmetric mechanisms.} The only crucial difference is that, instead of a single dominating distribution, we get a family of such distributions indexed by sets of uniqueness~$A$. 
The necessity of checking a continuum of dominating distributions may serve as an explanation for the computational intractability of the feasibility problem for public decisions, established by \cite*{gopalan2018public}.

{For $n=2$ agents, private good allocation is essentially equivalent to public decisions with binary $\Omega$, where the state $\omega$ is the agent that received the good. Thus, for the first agent, the probability of receiving the good conditional on her type is equal to her belief about $\{\omega=1\}$, i.e., $C_1(t_1)(1)=\bP(\omega=1\mid t_1)$. Similarly, for the second agent, $C_2(t_2)(2)=1-\bP(\omega=1\mid t_2)$. We can allow the good to be retained by the designer by considering the set of states $\{1,2,\emptyset\}$, in which case
$C_1(t_1)(1)$ and $C_2(t_2)(2)$ decrease pointwise. Combining these observations with our characterization of feasibility via conjugate distributions, we obtain the following corollary.}
\begin{corollary}\label{cor_asymmetric_hr}
{Suppose a private good is allocated to one of  two agents or retained by the designer using a social choice rule ${C}\colon T_1\times T_2\to \Delta(\{1,2,\emptyset\})$. For a pair of functions $Q_i\colon T_i \to [0,1]$, $i=1,2$, there exists ${C}$ such that $Q_i(t_i)={C}_i(t_i)(i)$ if and only if there is a pair of distributions $G_i$ that are conjugates---i.e., $G_2=\hat{G}_1$---such that the distributions of $Q_1(t_1)$ and $Q_2(t_2)$ are obtained from  $G_1$ and $1-G_2(1-\,\cdot\,)$, respectively, by applying mean-preserving contraction and shifting mass to the left.}\footnote{{Equivalently, the distribution of $1-Q_1$ second-order stochastically dominates $1-G_1(1-\,\cdot\,)$ and the distribution of $1-Q_2$ second-order stochastically dominates $G_2(\,\cdot\,)$.}}  
\end{corollary}
{For the symmetric case of $Q_1=Q_2$ studied by \cite*{hart2015implementation}, one can take $G_1(x)=1-G_2(1-x)$. Together with $G_2=\hat{G}_1$ this yields $G_1(x)=G_2(x)=x$. Thus~$G_i$ are uniform on $[0,1]$, and we deduce the two-agent case of the characterization by \cite*{hart2015implementation} from Corollary~\ref{cor_asymmetric_hr}. The corollary shows that their result naturally extends to asymmetric mechanisms. In that setting, the connection between feasibility and majorization has not been known; see, e.g., \cite*{che2013generalized}.}

\section{Application to Persuasion  in Zero-Sum Games}\label{app:applic}

Here, we outline another applications of private private signals. 
Suppose agents compete in a zero-sum game, and a designer who knows the state wishes to influence how agents' actions correlate with the state. We show that equilibrium signals must be private private. So,  our bounds on the informativeness of private private signals   limit how much the designer can adapt the agents' actions to the state, thus constraining the designer's payoffs. 

Consider a zero-sum game played by two players. The action set of player $i \in \{1,2\}$ is $A_i$, which we take to be finite, and the utilities are given by $u_1=-u_2=u$ for some $u \colon A_1 \times A_2 \to \R$. We assume that this game has a unique mixed Nash equilibrium, which holds for  generic zero-sum games \citep*{viossat2008having}.

There is a random state $\omega$ taking value in $\Omega$. The two players do not know the state and their payoffs do not depend on it. But, there is another agent (the designer) who knows the state and has a utility function $u_d \colon \Omega \times A_1 \times A_2 \to \R$ that depends on the state and the actions of the players. This can model a setting where a designer wants to influence the actions of two competitors, with the  designer's preference over actions given by her private type $\omega$. The designer commits to a (not necessarily private private) information structure $(\omega,s_1,s_2)$. When the state $\omega$ is realized, the designer observes it and sends the signal $s_1$ to player 1 and $s_2$ to player 2. The players  choose their actions after observing the signals. 

As a simple example, suppose the game is rock-paper-scissors, so that $A_1 = A_2 = \{R,P,S\}$ and $u(a_1,a_2)$ equals $1$ on $\{(P,R),(R,S),(S,P)\}$, zero on the diagonal, and $-1$ on the remaining action pairs. The state $\omega$ takes values in $\{0,1\}$ and is equal to $1$ with probability $\nicefrac{1}{2}$. The designer gets a payoff of 1 for each player who chooses scissors in the high state or chooses rock in the low state.

A pure strategy of player $i$ is a map $f_i \colon S_i \to A_i$, and a mixed strategy $\sigma_i$ is a random pure strategy. An equilibrium consists of an information structure together with a strategy profile $(\sigma_1,\sigma_2)$ such that each agent maximizes her expected utility given her signal. That is, for every $s_i \in S_i$ and~$a_i \in A_i$
\begin{align*}
    \E{u_i(\sigma_i(s_i),\sigma_{-i}(s_{-i}))|s_i} \geq \E{u_i(a_i,\sigma_{-i}(s_{-i}))|s_i}.
\end{align*}
This is just the incentive compatibility condition of a correlated equilibrium, and so, by a direct revelation argument, we can assume that $S_i = A_i$ and that $\sigma_i$ is always the identity: in equilibrium, the designer recommends an action to each agent, and the agents follow the recommendations. We refer to such equilibria as direct-revelation equilibria.

The next claim shows that private private information structures arise endogenously in this setting.
\begin{claim}
\label{clm:zero-sum}
In every direct-revelation equilibrium, the information structure $(\omega,s_1,s_2)$ is a private private information structure.
\end{claim}
Indeed, a zero-sum game with a unique Nash equilibrium has a unique correlated equilibrium which is equal to that Nash equilibrium \citep*{forges1990correlated}. Thus $(s_1,s_2)$ form a Nash equilibrium, and in particular $s_1$ must be independent of~$s_2$.

The intuition behind this result is simple: revealing to player $i$ any information about the recommendation given to player $-i$ gives $i$ an advantage that she can exploit to increase her expected utility beyond the value of the game. But player $-i$ can guarantee that $i$ does not get more than the value, and hence $s_i$ cannot contain any information about $s_{-i}$. Note that Claim~\ref{clm:zero-sum} applies beyond generic zero-sum games to any game with any number of players, provided that it has a unique correlated equilibrium.\footnote{The set of games with a unique correlated equilibrium is open \citep*{viossat2008having},  so a small enough perturbation of (for example) the rock-paper-scissors game will still have a unique correlated equilibrium, although it will not be zero-sum. As a side note, we are unaware of interesting examples of three-player games with a unique mixed correlated equilibrium. In particular, the following question is open, to the best of our knowledge: does there exist a three-player game with a unique correlated equilibrium in which no player plays a pure strategy?}

In the rock-paper-scissors example above, the joint distribution of $(s_1,s_2)$ must be uniform over $\{R,P,S\}\times\{R,P,S\}$, by Claim~\ref{clm:zero-sum}. However, the designer is free to choose the joint distribution between $(s_1,s_2)$ and  $\omega$. Thus her problem is to maximize $\E{u_d(\omega,s_1,s_2)}$ over all  structures in which $(s_1,s_2)$ is uniform over $\{R,P,S\}\times\{R,P,S\}$. Choosing $(s_1,s_2)$ independently of $\omega$ yields a payoff of $\nicefrac{6}{9}$. A straightforward calculation shows that an optimal structure yields her a payoff of  $\nicefrac{10}{9}$. By comparison, in a relaxed problem where the designer is allowed to dictate the players' actions without worrying about the privacy constraint,   she can achieve utility 2 by revealing the state to both players, telling them to both choose scissors when the state is high and rock when the state is low.

Beyond the specifics of the rock-paper-scissors example, the fact that equilibrium signals are private private means that any bound on the informativeness of private private signals yields a bound on the designer's equilibrium utility: if the designer's recommendations only contain a limited amount of information about the state, then she cannot hope that the players' actions efficiently adapt to the state and yield her high utility. Thus our results, including Theorem~\ref{thm:binary-pareto} and Propositions~\ref{prop:mutual-info},~\ref{prop:mutual-info-two}, and~\ref{prop:quadartic-info}, constrain what can be achieved by the designer in any such setting.

\bibliography{divide_info}

\newpage
\appendix
\begin{center}
\textbf{\LARGE{}Appendix}{\LARGE\par}
\par\end{center}

\section{Omitted Proofs}
Note that we sometimes prove results in a different order than the order that they appear in the main text, since some of the results we state earlier are implied by some of the later results. 

\subsection{Preliminary Lemmas}\label{app_prelim_lemmas}
Let $\mathcal{I}=(\omega,s_1, \ldots,s_n)$ be a private private information structure. The signals $s_1,\ldots,s_n$ can be combined into a new signal $s=(s_1,\ldots,s_n)$. The following lemma gives a lower bound on the informativeness of the combined signal  $s$ in terms of the informativeness of the individual signals. It can be seen as superadditivity of mutual information for independent signals. Recall that the mutual information $I(\omega\,;\,s)$ is defined in \eqref{eq:I}.

\begin{lemma}\label{lm_supperadditive}
For a private private information structure $(\omega,s_1,\ldots, s_n)$ the following inequality holds
\begin{equation}\label{eq_superadditive}
\sum_{i=1}^n I\big(\omega\,;\,s_i\big)\leq I\big(\omega\,;\,(s_1,\ldots,s_n)\big).
\end{equation}
\end{lemma}
\begin{proof}
The result for $n\geq 3$ follows from the result for $n=2$ by applying it sequentially to $(s_1,\ldots,s_k)$ for $k \leq n$. Consequently, in the rest of the proof we assume $n=2$.

Our goal is to show that $\Delta = I\big(\omega\,;\,(s_1,s_2)\big)-I\big(\omega\,;\,s_1\big)-I\big(\omega\,;\,s_2\big)\geq 0.$ Let $p_1(k) = p(s_1)(k) = \Pr{\omega=k}{s_1}$, define $p_2$ likewise, and let $p_{12}(k) = p(s_1,s_2)(k) = \Pr{\omega=k}{s_1,s_2}$. Let $p$ denote the prior distribution of $\omega$. By the martingale property $\E{p_{12}}{p_i} = p_i$ and $\E{p_i}=p$. (Here we view $p_1$, $p_2$, and $p_{12}$ as functions of the random variables $(s_1,s_2)$ and always take expectation with respect to the joint distribution of $(\omega, s_1, s_2)$.) 

Using this notation and the definition of mutual information, we can write for $i \in \{1,2\}$
\begin{align*}
 I\big(\omega\,;\,s_i\big) = \E{\sum_k p_i(k)\log\frac{p_i(k)}{p(k)}}~~~\text{and}~~~I\big(\omega\,;\,s_1,s_2\big) = \E{\sum_k p_{12}(k)\log\frac{p_{12}(k)}{p(k)}}.
\end{align*}
By the martingale property we can replace the first $p_i$ by $p_{12}$:
\begin{align*}
 I\big(\omega\,;\,s_i\big) = \E{\sum_k p_{12}(k)\log\frac{p_i(k)}{p(k)}}.
\end{align*}
Thus
\begin{align*}
  \Delta
  &= -\E{\sum_k p_{12}(k)\log\frac{p_1(k)p_2(k)}{p_{12}(k)p(k)}}.
\end{align*}
Applying Jensen's inequality to the logarithm, we get that
\begin{align*}
  \Delta
  \geq -\log\E{\sum_k p_{12}(k)\frac{p_1(k)p_2(k)}{p_{12}(k)p(k)}}.
\end{align*}
By cancelling and rearranging, we get
\begin{align*}
  \Delta
  \geq -\log\sum_k\frac{1}{p(k)}\E{ p_1(k)p_2(k)}.
\end{align*}
Since $p_1(k)$ and $p_2(k)$ are independent,
\begin{align*}
  \Delta
  \geq -\log\sum_k\frac{1}{p(k)}\E{ p_1(k)}\E{p_2(k)}.
\end{align*}
By the martingale property $\E{p_i(k)}=p(k)$, and so $\Delta
  \geq -\log\sum_kp(k) = 0$.
\end{proof}
Note that this proof only used the independence of $(s_1,s_2)$ to the extent that it implies that $p(s_1)$ is uncorrelated with $p(s_2)$.

\bigskip

To show that a given private private information structure $\mathcal{I}=(\omega,s_1,\ldots,s_n)$ is Blackwell-Pareto dominated, we will often use the following technique:  construct an additional informative signal $t$ independent of $s_1,\ldots,s_n$, and {augment one of $s_i$ with~$t$,}
say, the first one. The new information structure $\mathcal{I}' = (\omega,(s_1,t),s_2,\ldots,s_n)$ strictly dominates $\mathcal{I}$ thanks to the following direct corollary of Lemma~\ref{lm_supperadditive}.
\begin{corollary}
\label{cor_extra_independent_signal}
Fix $\omega$, and consider a pair of signals $s$ and $t$ such that
\begin{itemize}
    \item $s$ and $t$ are independent, and
    \item $t$ is not independent of $\omega$.
\end{itemize}
Then the information structure $(\omega, (s,t))$ strictly dominates
$(\omega, s)$ with respect to the Blackwell order.
\end{corollary}
\begin{proof}
Clearly $(\omega,s)$ is weakly dominated by $(\omega,(s,t))$. We show that this domination is strict.

Since $t$ is informative, $I(\omega;t)>0$. Hence, by Lemma~\ref{lm_supperadditive}, $I(\omega;(s,t)) \geq I(\omega;s) + I(\omega;t) > I(\omega;s).$
Since $I(\omega;s)$ is the value of a particular decision problem (where the indirect utility is given by minus the entropy), it follows that $(\omega,(s,t))$ strictly dominates $(\omega,s)$.
\end{proof}
The next lemma shows that, without loss of generality, induced posteriors are equal to signals, which can be seen as a version of the revelation principle for private private information structures. 
\begin{lemma}\label{lm_posteriors_as_signals}
Any private private information structure $\mathcal{I}=(\omega,s_1,\ldots,s_n)$ is equivalent to 
$\mathcal{J}=(\omega,t_1,\ldots,t_n)$ where each {signal $t_i$ is the posterior $p(s_i)$ in the structure~$\mathcal{I}$.}
\end{lemma}
\begin{proof}
By the law of total expectation, $p(t_i)=t_i$. It follows that $p(s_i)$ and $p(t_i)$ have the same distribution, and so are Blackwell equivalent.
\end{proof}
For a private private information structure $\mathcal{I}=(\omega,s_1,\ldots, s_n)$, recall that we denote by $\mu_i\in \Delta(\Delta(\Omega))$  the distribution of the belief $p(s_i)$. 
Let $\F_n\subset \Delta(\Delta(\Omega))^n$ be the set of distributions $\mu_1,\ldots,\mu_n$ {feasible under the constraint of privacy}, i.e., those that correspond to some private private information structure $\mathcal{I}$ for some prior.  
\begin{lemma}\label{lm_feasible_is_closed}
{For any number of signals $n$,} the set of distributions $\F_n$ {feasible under the constraint of privacy } is compact in the topology of weak convergence.
\end{lemma}
\begin{proof}
Since the set of probability measures $\Delta(\Delta(\Omega))$ is compact, to prove the compactness of $\F_n$, it is enough to check that it is closed. In other words, we need to check that if a sequence of feasible distributions $(\mu_1^{l},\ldots,\mu_n^{l})$ weakly converges to $(\mu_1^{\infty},\ldots,\mu_n^{\infty})$ as $l\to\infty$, then the limit is also feasible.

Let $\mathcal{I}^{l}=(\omega,s_1^{l},\ldots, s_n^{l})$ be an information structure inducing $(\mu_1^{l},\ldots,\mu_n^{l})$. By Lemma~\ref{lm_posteriors_as_signals}, we can assume without loss of generality that the signals $s_i^{l}$ are in $\Delta(\Omega)$ and they coincide with the induced beliefs, i.e., $p\big(s_i^{l}\big)=s_i$. Let $\psi^{l}\in \Delta(\Omega\times \Delta(\Omega)^n)$ be the joint distribution of $\omega$ and the beliefs $s_1^{l},\ldots,s_n^{l}$. By compactness of the set of probability measures, we can extract a subsequence of $\psi^{l}$ converging to some $\psi^{\infty}$. By definition, the marginal of $\psi^{\infty}$ on the belief coordinates equals $\mu_1^{\infty}\times\ldots\times\mu_n^{\infty} $. 

Consider a private private information structure $\mathcal{I}^{\infty}=(\omega,s_1^{\infty},\ldots, s_n^{\infty})$, where signals $s_i^{\infty}$ belong to $\Delta(\Omega)$ and the joint distribution of the state and signals is given by $\psi^{\infty}$. Each signal $s_i^{\infty}$ has distribution $\mu_i^{\infty}$. Let us check that the induced beliefs coincide with the signals, i.e., $p\big(s_i^{\infty}\big)(k)=s_i^{\infty}(k)$ almost surely for each $k\in \Omega$. We verify an equivalent integrated version of this identity:
\begin{equation}
\int \left(\sum_k h(k,s_i^{\infty})p(s_i^\infty)(k)\right)\dd \psi^{\infty}=
\int \left(\sum_k h(k,s_i^{\infty})s_i^{\infty}(k)\right)\dd \psi^{\infty}
\end{equation}
for any continuous function $h$ on $\Omega\times \Delta(\Omega)$.  
Since the left-hand side is just the integral of $h$, this is equivalent to 
\begin{equation}\label{eq_posteriors_are_signals_in_the_limit}
\int h(\omega,s_i^{\infty})\dd \psi^{\infty}=
\int \left(\sum_k h(k,s_i^{\infty})s_i^{\infty}(k)\right)\dd \psi^{\infty}.
\end{equation}
For each $l<\infty$, the beliefs in $\mathcal{I}^l$ coincide with the signals, i.e.,  
$$\int h(\omega,s_i^{l})\dd \psi^{l}=\int \left(\sum_k h(k,s_i^{l})s_i^{l}(k)\right)\dd \psi^{l}.$$
As integration of a continuous function commutes with taking weak limits, letting $l$ go to infinity, we obtain~\eqref{eq_posteriors_are_signals_in_the_limit}.

We conclude that each belief $p\big(s_i^\infty\big)$ in $\mathcal{I}^\infty$ coincides with the signal $s_i^\infty$ and the latter is distributed according to $\mu_i^\infty$. Therefore, $(\mu_1^\infty \ldots,\mu_n^\infty )$ are {feasible under the constraint of privacy}, and so the set of feasible distributions is closed and thus compact.
\end{proof}
The next lemma shows that our order on private private information structures is well-behaved, in the sense that each structure is dominated by a Blackwell-Pareto optimal one: each element of the partially ordered set of private private information structures is upper bounded by a maximal element.
\begin{lemma}
\label{lem:pareto-dominated}
For any private private information structure $\mathcal{I}=(\omega,s_1,\ldots,s_n)$, there exists a Blackwell-Pareto optimal structure $\mathcal{I}'=(\omega,s_1',\ldots,s_n')$ that weakly dominates~$\mathcal{I}$.
\end{lemma}
\begin{proof}
Recall that $\mathcal{I}\preceq \mathcal{J}$ if for any continuous convex $\varphi \colon \Delta(\Omega) \to \R$ and any~$i=1,\ldots, n$,
\begin{equation}\label{eq_dominance_mu}
\int \varphi(q)\dd \mu_i(q)\leq \int \varphi(q)\dd \nu_i(q)
\end{equation}
where $\mu_i$ and $\nu_i$ are the distributions {of beliefs induced by  signal $i$  in $\mathcal{I}$ and $\mathcal{J}$, respectively. As the map $\mu \mapsto \int \varphi\dd\mu$ of integration against a continuous function $\varphi$ is, by definition, a continuous function in the weak topology on $\Delta(\Delta(\Omega))$, it follows that this partial order is continuous.}

   Let $(\mu_1,\ldots,\mu_n)$ be the distributions of posteriors induced by $\mathcal{I}$ and let $\F$ be the set of distributions {feasible under the constraint of privacy }, endowed with the weak topology. As $\F$ is compact by Lemma~\ref{lm_feasible_is_closed} and the dominance order is continuous, there is  a maximal element $(\nu_1,\ldots,\nu_n)\in \F_n$ dominating $(\mu_1,\ldots,\mu_n)$. Since $(\nu_1,\ldots,\nu_n)$ is feasible, it is induced by some private private information structure $\mathcal{I}'$. By the construction, $\mathcal{I}'$ dominates $\mathcal{I}$ and is Blackwell-Pareto optimal.
\end{proof}

\subsection{Proof of Proposition~\ref{prop:associated}}\label{app_proof_representation}
We need to show that, given a private private information structure $\mathcal{I}=(\omega,s_1,\ldots,s_n)$ with $\Omega=\{0,\ldots m-1\}$, there is an equivalent structure associated with a partition $\cA=(A_0,\ldots, A_{m-1})$ of $[0,1]^n$.
The construction relies on two lemmas. Lemma~\ref{lm_uniform_signals_wlog} shows that assuming signals $s_i$ are uniform on $[0,1]$ is without loss of generality. Hence it remains to show that there is an equivalent information structure where  signal realizations determine the state. This is done using a secret-sharing scheme from Lemma~\ref{lm_secret}.

\begin{lemma}\label{lm_uniform_signals_wlog}
For any private private information structure $\mathcal{I}=(\omega,s_1,\ldots,s_n)$, there is an equivalent  private private information structure $\mathcal{I}'=(\omega,s_1',\ldots,s_n')$ such that each $s_i'$ is uniformly distributed on $[0,1]$. 
\end{lemma}
\begin{proof}
Consider the information structure $\mathcal{J}=(\omega,t_1,\ldots,t_n)$ where $t_i=(s_i,r_i)$, and each $r_i$ is independent and uniformly distributed on $[0,1]$. Clearly, $\mathcal{I}$ and $\mathcal{J}$ are equivalent.  As $t_i$ is nonatomic, and since all standard nonatomic probability spaces are isomorphic, $t_i$ can be reparametrized to be uniform on $[0,1]$.
\end{proof}

We say that a signal $r$ is \emph{split} into $r_1$ and $r_2$ if $r$ is a function of $r_1$ and $r_2$, i.e., $r=f(r_1,r_2)$.
\begin{lemma}\label{lm_secret}
A signal $r$ distributed uniformly  on $[0,1]$ can be split into $r_1$ and $r_2$ such that each $r_i$ is uniformly distributed on $[0,1]$ and the three random variables $r$, $r_1$, and $r_2$ are pairwise independent. Furthermore, if $t$ is an additional signal that is independent of $r$, then we can take the pair $(r_1,r_2)$ to be independent of $t$.
\end{lemma}
This lemma extends the classic secret sharing idea from cryptography which applies to discrete random variables. The proof, by construction, is immediate.\footnote{We are thankful to Tristan Tomala for suggesting this construction.} 
\begin{proof}
{Denote by $\lfloor x \rfloor$ the number $x$ rounded down to the nearest integer, and denote by $\mathrm{frac}(x)=x-\lfloor x \rfloor$ the fractional part of $x \in \R$. Take $r_1$ independent of both $r$ and $t$ and distributed uniformly on $[0,1]$, and let $r_2=\mathrm{frac}(r_1 + r)$. Then $r = \mathrm{frac}(r_2-r_1)$ and $r_1$, $r_2$, and $r$ are easily seen to be pairwise independent and also independent of $t$ altogether. }
\end{proof}

With the help of Lemmas~\ref{lm_uniform_signals_wlog} and~\ref{lm_secret}, we are ready to prove Proposition~\ref{prop:associated}.  
\begin{proof}[Proof of Proposition~\ref{prop:associated}]
We are given a private private information structure $\mathcal{I}=(\omega,s_1,\ldots,s_n)$ with sets of signal realizations $S_i$, $i=1,\ldots, n$. We aim to construct an equivalent one, $\mathcal{I}'$, where each signal $s_i'$ is uniformly distributed on $[0,1]$ and the realization of signals $(s_1',\ldots, s_n')$ determines the state or, equivalently, $\mathcal{I}'$ is associated with some partition $\cA=(A_1,\ldots, A_{m-1})$ of $[0,1]^n$.

By Lemma~\ref{lm_uniform_signals_wlog}, we can find a private private information structure $(\omega,t_1,\ldots,t_n)$ equivalent to $\mathcal{I}$ where each $t_i$ is uniformly distributed in $[0,1]$. If the signals $(t_1,\ldots, t_n)$ determine the state, then the proof is completed.

Consider the case where $(t_1,\ldots, t_n)$ do not determine the state $\omega$. To capture the  uncertainty in $\omega$ remaining after the  signals have been realized, we construct a new signal $t$ as follows.

Let $q:\, [0,1]^n\to \Delta(\Omega)$ be a conditional distribution of $\omega$ given all the signals, i.e., $q(x_1,\ldots, x_n)(k)=\Pr{\omega =k}{t_1=x_1,\ldots, t_n=x_n}$ for any $k \in \Omega$. 
With each distribution $q\in \Delta(\Omega)$ we associate a partition of $[0,1)$ into $m$ intervals 
$$B_k(q)=\left[\sum_{l=0}^{k-1}q(l),\ \sum_{l=0}^{k}q(l)\right),\qquad k=0,\ldots, m-1.$$
The length of $B_k(q)$ equals the mass assigned by $q$ to $\omega=k$.
Let $r$ be a random variable uniformly distributed on $[0,1]$ and independent of $(t_1,\ldots, t_n)$. Consider a new state variable $\omega'\in \Omega$ such that $\omega'=k$ whenever $r\in B_k\big(q(t_1,\ldots, t_n)\big)$. By definition, the joint distributions of $(\omega,t_1,\ldots, t_n)$ and $(\omega',t_1,\ldots, t_n)$ coincide and, therefore, the two  structures are equivalent.

The new state $\omega'$ is determined by the realizations of $t_1,\ldots, t_n$  and the new signal $r$. Using Lemma~\ref{lm_secret}, we split the signal $r$ into $r_1$ and $r_2$ that are independent of each other, and where each $r_i$ is independent of $r$. Note that, by this lemma, we can take $r_1$ and $r_2$ to be independent of $(t_1,\ldots,t_n)$. Since each $r_i$ is uninformative of $r$, the  structure $(\omega',(t_1,r_1),(t_2,r_2),t_3,\ldots, t_n)$,  where $r_1$ {augments the first signal} 
and  $r_2$ {augments} the second one, is a private private structure equivalent to $\mathcal{I}$. Since $r$ is a function of $r_1$ and $r_2$, the signals $(t_1,r_1),(t_2,r_2),t_3,\ldots, t_n$ determine the state. 

It remains to reparameterize the {first two signals} 
so that, instead of being uniform on $[0,1]^2$, they become uniform on $[0,1]$. Consider any  bijection $h:\, [0,1]^2\to[0,1]$ preserving the Lebesgue measure; such a bijection exists since both are standard nonatomic spaces. 
Define $s_1'=h(t_1,r_1)$, $s_2'=h(t_2,r_2)$, and $s_i'=t_i$ for $i=3,\ldots, n$. The private private information structure $\mathcal{I}'=(\omega',s_1',s_2',\ldots, s_n')$ is equivalent to $\mathcal{I}$, all the signals are uniform on $[0,1]$, and the realization of signals determines $\omega'$.
\end{proof}

\subsection{Proof of Theorem~\ref{thm:pareto}}\label{app:Pareto_via_uniqueness}

We formulate and prove an extension of Theorem~\ref{thm:pareto} applicable to non-binary sets of states $\Omega=\{0,1,\ldots, m-1\}$.

Consider a partition of $[0,1]^n$ into $m$ measurable sets $\cA=(A_0,\ldots A_{m-1})$.
Recall that the structure $\mathcal{I}=(\omega,s_1,\ldots,s_m)$ is said to be associated with a partition $\cA$ if all the signals are uniform on $[0,1]$ and $\omega=k$ whenever $(s_1,\ldots s_n)\in A_k$.

We say that two partitions $\cA=(A_0,\ldots A_{m-1})$ and $\cB=(B_0,\ldots B_{m-1}) $ are \emph{equal} if $A_k$ and $B_k$ differ by a set of zero Lebesgue measure for each $k$. Recall that the projection of a measurable set $A \subseteq [0,1]^n$ on the $i$th coordinate is denoted by $\alpha^A_i$ (see~\S\ref{sec_Pareto_via_tomography}).
The notion of sets of uniqueness from~\S\ref{sec_Pareto_via_tomography} extends to partitions as follows.
\begin{definition}
A partition $\cA=(A_0,\ldots, A_{m-1})$ is a \emph{partition of uniqueness} if for any partition $\cB=(B_0,\ldots, B_{m-1})$ such that $\alpha_i^{A_k}=\alpha_i^{B_k}$ for all $i$ and $k$, it holds that $\cA=\cB$.
\end{definition}
\begin{theorem}[Extension of Theorem~\ref{thm:pareto} to $m$ states]\label{thm:pareto_appendix}
A private private information structure $\mathcal{I}$ is Blackwell-Pareto optimal if and only if it is equivalent to a structure associated with a partition of uniqueness $\cA$.
\end{theorem}
Note that in the case of $m=2$ states, a set $A_1$ in a partition $\cA=(A_0,A_1)$ determines $A_0=[0,1]^n\setminus A_1$. Hence, $\cA=(A_0,A_1)$ is a partition  of uniqueness if and only if $A_1$ is a set of uniqueness. Hence, Theorem~\ref{thm:pareto} is an immediate corollary of its extended version. 
For an application of the theorem for $m>2$, see the example contained in Appendix~\ref{app_non_uniqueness}. This example also demonstrates that   partitions of uniqueness are not necessarily composed of sets of uniqueness  for $m> 2$ and, hence, the requirement of a partition to be a partition of uniqueness does not boil down to restrictions on its elements unless $m=2$.  
\medskip

The proof of the theorem is split into a sequence of lemmas. We say that a private private information structure is \emph{perfect} if the {realizations of all the signals together determine}
the realization of $\omega$, i.e., there exists a function $f\colon S_1\times\ldots\times S_n\to \Omega$ such that $\omega=f(s_1,\ldots,s_n)$. In particular, a structure with signals uniform in $[0,1]$ is associated with some partition if and only if it is perfect.

The next lemma shows that perfection is necessary for Blackwell-Pareto optimality. 
\begin{lemma}\label{lm_PO_implies_perfection}
If a private private information structure $\mathcal{I}=(\omega,s_1,\ldots,s_n)$ is equivalent to a structure that is not perfect, then $\mathcal{I}$ is not Blackwell-Pareto optimal.
\end{lemma}
The construction of the Blackwell-Pareto improvement resembles the proof of Proposition~\ref{prop:associated} except for the fact that the newly constructed signal {is revealed entirely on top of one of the existing signals thus strictly improving its informativeness in the Blackwell order,} by Corollary~\ref{cor_extra_independent_signal}.
\begin{proof}
Without loss of generality, $\mathcal{I}$ itself is imperfect.
Let 
$q\colon S_1\times\ldots\times S_n\to\Delta(\Omega)$
be the  distribution of $\omega$ conditional on $s_1=x_1,\ldots, s_n=x_n$, i.e., $q(x_1,\ldots,x_n)(k)=\Pr{\omega=k}{ s_1=x_1,\ldots, s_n=x_n}$, $k=0,\ldots,m-1$. Since $\mathcal{I}$ is not perfect, we can find a state $k_0 \in \Omega$ such that the event  $\{\omega=k_0\}$ is not always determined by the signals. That is, the random variable $q(s_1,\ldots,s_n)(k_0)$ does not always take values in $\{0,1\}$.
Without loss of generality, we assume that $k_0=0$.
With each $q\in \Delta(\Omega)$ we associate a partition of $[0,1)=\bigsqcup_{k\in \Omega} B_k(q)$, where 
$$B_k(q)=\big[q(\{0,\ldots,k-1\}),\ q(\{0,\ldots,k\})\big),\qquad k=0,\ldots, m-1$$
so that the length of $B_k(q)$  equals $q(k)$. 

We construct a new equivalent structure with an extra signal $t$ as in the proof of Proposition~\ref{prop:associated}. Let $t$ be a random variable uniformly distributed on $[0,1]$ and independent of  $s_1,\ldots, s_n$. Define a new state $\omega'$ as a function of these variables in the following way: $\omega'=k$  whenever $t\in B_k\big(q(s_1,\ldots, s_n)\big).$
The joint distribution of $(\omega',s_1,\ldots,s_n)$ coincides with that of $(\omega,s_1,\ldots,s_n)$ and, hence, the two structures are equivalent.

To get a Blackwell-Pareto improvement, we {augment the first signal with $t$}
and obtain a private private information structure $\mathcal{I}'=(\omega', (s_1,t),s_2,\ldots, s_n)$. To argue that $\mathcal{I}'$ is indeed a Blackwell-Pareto improvement we need to show that $t$ itself is an informative signal about $\omega'$, i.e., the posterior $p(t)\in\Delta(\Omega)$ is not equal to the prior $p$ with a positive probability. It is enough to show $p(t)(k_0)$ takes different values for $t$ in $[0,\varepsilon]$ and in $[1-\varepsilon,1]$. As we assume without loss of generality that $k_0=0$, the interval $B_{k_0}$ is the leftmost one in the partition, and so $p(t)(k_0) = p(t)(0)=\Pr{t< q(s_1,\ldots,s_n)(0)}{t}.$
That is, if we denote by $Q$ the cumulative distribution function of $q(s_1,\ldots,s_n)(0)$, then $p(t)(k_0) = 1-G(t)$. Since  $G$ is a non-constant function on $(0,1)$ by our assumption on $k_0$, the induced belief $p(t)(k_0)$ is not a constant. Thus $t$ is informative. By Corollary~\ref{cor_extra_independent_signal}, this implies that the signal $(s_1,t)$ which 
in $\mathcal{I}'$ strictly dominates the signal $s_1$ 
in $\mathcal{I}$. As the {all other signals remain}
the same in the two structures,  $\mathcal{I}'$ strictly Blackwell-Pareto dominates~$\mathcal{I}$.
\end{proof}
The next step is to show that only structures corresponding to partitions of uniqueness can be Blackwell-Pareto optimal.
\begin{lemma}
\label{lm_PO_implies_uniqueness}
If a private private information structure $\mathcal{I}=(\omega,s_1,\ldots,s_n)$ is Blackwell-Pareto optimal, then $\mathcal{I}$ is equivalent to  a structure associated with a partition of uniqueness. 
\end{lemma}
\begin{proof}
By Proposition~\ref{prop:associated}, we can find a private private information structure $\mathcal{J}=(\omega,t_1,\ldots, t_n)$ equivalent to $\mathcal{I}$ and associated with some partition $\cA=(A_0,\ldots, A_{m-1})$ of $[0,1]^n$. Let us demonstrate that $\cA$ is a partition of uniqueness. Towards a contradiction, assume that there is another partition $\cA'=(A_0',\ldots,A'_{m-1})$ not equal to $\cA$ but such that the projections $\alpha_i^{A_k}=\alpha_i^{A'_k}$ for all $i$ and $k$. So $\mathcal{I}$ is also equivalent to the structure $\mathcal{J}'=(\omega,t'_1,\ldots, t'_n)$ associated with $\cA'$.

By Lemma~\ref{lm_PO_implies_perfection}, to get a contradiction, it is enough to construct an information structure $\mathcal{I}'$ that is equivalent to $\mathcal{I}$ but not perfect, as this would imply the existence of a strict Blackwell-Pareto improvement. We define $\mathcal{I}'$ as a structure where the joint distribution of the state and signals is a convex combination of the corresponding distributions in $\mathcal{J}$ and $\mathcal{J}'$. Formally, let $s_1',\ldots, s_n'$ be independent random variables each uniformly distributed on $[0,1]$ and let $\theta\in \{0,1\}$ be a symmetric Bernoulli random variable independent of $(s_1',\ldots, s_n')$. Define the state $\omega'$ as follows:
$$\omega'=k \quad\mbox{if}\quad\left[\begin{array}{c}
(s_1',\ldots, s_n')\in A_k \  \mbox{and} \ \theta=0\\
(s_1',\ldots, s_n')\in A_k' \  \mbox{and} \ \theta=1 \end{array}\right..$$
Since elements of the partitions $\cA$ and $\cA'$ have the same projections, the posterior induced by observing $t_i=x$ in $\mathcal{J}$ is identical to the one induced by observing $t'_i=x$ in $\mathcal{J}'$. Hence it is again identical to the posterior induced by observing $s_i'=x$ in $\mathcal{I}'$. 
As the partitions $\cA$ and $\cA'$ are not equal, there are $k\ne k'$ such that the intersection $A_k\cap A_{k'}'$ has a non-zero Lebesgue measure. Hence, if $(s_1',\ldots, s_n')\in A_k\cap A_{k'}'$, whether $\omega=k$ or $\omega=k'$ is determined by $\theta$. We conclude that, with positive probability, the signals $(s_1',\ldots, s_n')$ do not determine the state, so $\mathcal{I}'$ is not perfect and thus both $\mathcal{I}'$ and $\mathcal{I}$ can be Blackwell-Pareto improved by Lemma~\ref{lm_PO_implies_perfection}. 
\end{proof}

We see that Blackwell-Pareto optimal structures are contained  in those associated with  partitions of  uniqueness (up to equivalence of information structures). This shows one direction of  Theorem~\ref{thm:pareto_appendix}. It remains to demonstrate that any partition of uniqueness leads to a Blackwell-Pareto optimal structure, i.e., the structures associated with different partitions of uniqueness cannot dominate each other. 

For this purpose, we need two intermediate steps contained in the next two lemmas. Lemma~\ref{lm_garbling_imperfection} shows that a garbling of an information structure is never perfect and Lemma~\ref{lm_uniqueness_as_extreme_point} implies that imperfect structures cannot be equivalent to those associated with partitions of uniqueness.  Recall that for a pair of information structures $(\omega,t)$ and $(\omega,s)$, the signal $t$ is  a garbling of $s$ if, conditional on $s$, $t$ and $\omega$ are independent.
A structure $\mathcal{I}=(\omega, s_1, \ldots, s_n)$ is a garbling of $\mathcal{I}'=(\omega, s_1',\ldots,s_n')$ if each $s_i$ is a garbling of $s'_i$ and each $s_i$ is independent of $(s_j')_{j\ne i}$. The last requirement means that each 
signal is garbled independently. 
 Note that, by Blackwell's Theorem \citep[Theorem 12]{blackwell1951comparison}, $\mathcal{I}'$ (weakly) dominates $\mathcal{I}$ if and only if $\mathcal{I}$ is equivalent to a garbling of~$\mathcal{I}'$.

\begin{lemma}\label{lm_garbling_imperfection}
If $\mathcal{I}$ is a garbling of a private private information structure $\mathcal{I}'$, then $\mathcal{I}$ is not perfect unless $\mathcal{I}$ and $\mathcal{I}'$ are equivalent.
\end{lemma}
\begin{proof}
Suppose that $\mathcal{I}$ is perfect, and so $\omega=f(s_1,\ldots,s_n)$  for some $f\colon\ S_1\times\ldots\times S_n\to \Omega$. Our goal is to show that $\mathcal{I}$ is equivalent to $\mathcal{I}'$. For a given realization of $s_i$, the state $\omega$ is a function of the remaining signals $s_j$, $j\ne i$. Since $s_i'$ and the collection $(s_j)_{i\ne j}$ are independent, we see that $\omega$ is independent of $s_i'$ conditional on $s_i$. In other words, $s_i'$ is also a garbling of $s_i$. We conclude that both $\mathcal{I}$ is a garbling of $\mathcal{I}'$ and $\mathcal{I}'$ is a garbling of $\mathcal{I}$, so they are equivalent.
\end{proof}
The next lemma is used to show that imperfect private private information structures  cannot correspond to partitions of uniqueness. Before stating it, we will need to introduce the following concept. A \emph{fuzzy partition} is a tuple $(g_0,\ldots,g_{m-1})$ of measurable functions $g_k \colon [0,1]^n \to [0,1]$ such that $\sum_k g_k=1$. We can identify this tuple with a single function $g \colon [0,1]^n \to \Delta(\Omega)$. The case of a partition is one in which each $g_k$ is the indicator of a set $A_k$ in a partition of $[0,1]^n$. As with partitions, we identify two fuzzy partitions if they agree almost everywhere. We denote the collection of fuzzy partitions by $G$. 

We define the projection of $g_k$ to its $i$th coordinate by
$$\alpha_i^{g_k}(x_i)=\int_{[0,1]^{n-1}} g_k(x_i,x_{-i})\,\dd x_{-i}.$$
When $g_k$ is the indicator of a set $A_k$, the projection $\alpha_i^{A_k}$ as defined in the main text is equal to the projection of $g_k$.
With each partition $\cA=(A_0,\ldots, A_{m-1})$ of $[0,1]^n$ we associate the set 
$G_{\cA}=\Big\{g \in G \quad\mbox{such that}\quad  \forall k,i \ \ \alpha_i^{g_k}=\alpha_i^{A_k}\Big\}$
of fuzzy partitions that have the same projections as $\cA$.
\begin{lemma}\label{lm_uniqueness_as_extreme_point}
A partition $\cA=(A_0,\ldots, A_{m-1})$ of $[0,1]^n$ is a partition of uniqueness if and only if $G_{\cA}$
is a singleton.
\end{lemma}
Note that $G_\cA$ always contains at least one element, namely, the indicators of the partition $\cA$, i.e.,  $(\one_{A_0},\ldots, \one_{A_{m-1}})\in G_\cA$. The idea behind the lemma is that all extreme points of $G_{\cA}$ are indicators of partitions with the same projections as $\cA$. Hence, if $G_{\cA}$ is not a singleton it has at least two distinct extreme points, i.e., there is at least one more partition with the same projections as $\cA$, which is incompatible with the fact that $\cA$ is a partition of uniqueness. This identification of extreme points and indicators has appeared before in the context of sets of uniqueness \citep[see][]{gutmann1991existence}.
\begin{proof}
First we show that $G_\cA$ can be treated as a non-empty compact convex subset of a locally convex Hausdorff vector space. Non-emptiness and convexity is straightforward and compactness is to be checked once  an appropriate topology is defined. 

Let $M([0,1]^n)$ be the set of all finite signed measures on $[0,1]^n$ endowed with the topology of weak convergence, making it a locally convex {completely metrizable (and hence Hausdorff)} topological vector space. We identify a bounded function $g_k\colon\, [0,1]^n\to \mathbb{R}$ with a measure $\mu_k$ on $[0,1]^n$ having the density $g_k$ with respect to the Lebesgue measure, i.e., $\dd\mu_k(x_1,\ldots, x_n)=g_k(x)\dd x_1\ldots\dd x_n$. Hence, $G_\cA$ can be identified with a subset of $\Big(M([0,1]^n)\Big)^\Omega$. 

Let $\Delta_\leq([0,1]^n)$ be the set of sub-probability measures, i.e., non-negative measures $\mu$ with $\mu([0,1]^n)\leq 1$. The set $\Delta_\leq([0,1]^n)$ is a compact subset of $M([0,1]^n)$, {as it is the closed convex hull of the compact set given by the union of the probability measures and the zero measure \citep[Theorem 5.35]{aliprantis2006infinite}.} As $G_\cA$ is a subset of the compact set  $\big(\Delta_\leq([0,1]^n)\big)^\Omega$, compactness of $G_\cA$ follows from its closedness. To check closedness, we rewrite the conditions defining $G_\cA$ in an integrated form using as test functions the continuous functions $h$ on $[0,1]^n$. The tuple of measures   $(\mu_0,\ldots,\mu_{m-1})\in \Big(M([0,1]^n)\Big)^\Omega$ belongs to $G_\cA$ if and only if
\begin{align}
  \int_{[0,1]^n} \big|h(x_1,\ldots,x_n)\big|\dd\mu_k&\geq 0 \label{eq_positivity}\\
  \sum_k\int_{[0,1]^n} h(x_1,\ldots,x_n)\dd\mu_k&=\int_{[0,1]^n} h(x_1,\ldots,x_n)\dd x_1\ldots\dd x_n \label{eq_on_the_sum}\\
   \int_{[0,1]^n} h(x_i)\dd\mu_k&=\int_{[0,1]} h(x_i)\alpha_i^{A_k}(x_i)\dd x_i  \label{eq_on_projections}
\end{align}
for all $k=0,\ldots, m-1$, $i=1,\ldots, n$, and continuous functions $h$ on $[0,1]^n$ (in the last condition, $h$ depends on one of the coordinates only). Condition \eqref{eq_positivity} is non-negativity, condition \eqref{eq_on_the_sum} is equivalent to $\sum_k g_k =1$, and condition \eqref{eq_on_projections} corresponds to the equal projections condition $\alpha_i^{g_k}=\alpha_i^{A_k}$. 
By the definition of the weak topology, integration of a continuous function commutes with taking weak limits. We conclude that $G_\cA$ contains all its limit points and thus is closed.

By the Krein-Milman theorem, any compact convex subset of a locally convex Hausdorff vector  space is the closed convex hull of its extreme points \citep[see][Theorem 7.68]{aliprantis2006infinite}. Thus $G_\cA$ is the closed convex hull of its extreme points. Consequently, if $G_\cA$ is not a singleton, it has at least two distinct extreme points. 
To prove the lemma, it remains to demonstrate that all extreme points of $G_\cA$ correspond to partitions. Towards a contradiction, assume that $g=(g_0,\ldots,g_{m-1})$ is an extreme point of $G_\cA$ but it is not a partition, i.e., there is a state $k_0$ such that $g_{k_0}(x)\notin \{0,1\}$ for  $x=(x_1,\ldots,x_n)$ in a set of positive Lebesgue measure. Since $\sum_k g_k=1$, there is $k'\ne k$ such that the set of $x$ where both $g_{k}(x)>0$ and $g_{k'}(x)>0$ has positive measure. Hence, for some $\varepsilon>0$, the set $D\subseteq[0,1]^n$ of $x$ such that both $g_{k}(x)>\varepsilon$ and $g_{k'}(x)>\varepsilon$ also has positive measure. Without loss of generality, we assume that $k=0$ and $k'=1$.

By Corollary~2  of~\cite{gutmann1991existence}, for any $D$ of positive measure, there are two disjoint sets $D_1,D_2\subseteq D$ also of positive measure having the same projections, i.e., $\alpha_i^{D_1}=\alpha_i^{D_2}$ for any $i=1,\ldots,n$. Hence, the function $a(x)=\varepsilon\big(\one_{D_1}(x)-\one_{D_2}(x)\big)$ has zero projections, is bounded by $\varepsilon$ in absolute value, and is equal to zero outside of the set $D$. For $\sigma\in \{-1,+1\}$, define 
$$g_0^\sigma(x)=g_0(x)+\sigma\cdot a(x),\qquad g_1^\sigma(x)=g_1(x)-\sigma\cdot a(x).$$
By definition, $g_0^\sigma$ and $g_1^\sigma$ have the same projections as $g_0$ and $g_1$, they are non-negative, and $g_0^\sigma+g_1^\sigma=g_0+g_1$ (hence, $g_0^\sigma+g_1^\sigma+\sum_{k\geq 2} g_k=1$).

We conclude that the two tuples $(g_0^\sigma,g_1^\sigma,g_2,g_3,\ldots,g_{m-1})$, $\sigma\in \{-1,+1\}$, belong to $G_\cA$. They are not equal to each other as the sets $D_1$ and $D_2$ are disjoint. Since the original collection $(g_0,\ldots, g_{m-1})$ is the average of the two constructed ones, it cannot be an extreme point. This contradiction implies that all the extreme points of $G_\cA$ correspond to partitions and completes the proof.
\end{proof}
Relying on the last two lemmas, we can demonstrate that any structure associated with a partition of uniqueness is Blackwell-Pareto optimal.
\begin{lemma}
\label{lm_uniqueness_implies_PO}
Let $\mathcal{I}$ be a private private information structure equivalent to a structure associated with a partition of uniqueness, then $\mathcal{I}$ is Blackwell-Pareto optimal.
\end{lemma}
\begin{proof}
Without loss of generality, $\mathcal{I}=(\omega,s_1\ldots,s_n)$ is itself
 a structure associated with a partition of uniqueness $\cA=(A_,\ldots,A_{m-1})$ of $[0,1]^n$.

Towards a contradiction, assume that there is a private private information structure $\mathcal{J}$ strictly dominating~$\mathcal{I}$.
By Blackwell's theorem, $\mathcal{I}$ is equivalent to some garbling of $\mathcal{J}$ denoted by $\mathcal{I}'=(\omega,s_1',\ldots,s_n')$. By Lemma~\ref{lm_garbling_imperfection}, $\mathcal{I}'$ is not perfect. Let $t_i=p(s_i')\in \Delta(\Omega)$ be the posterior belief induced by $s_i'$ and $\mu_i\in \Delta(\Delta(\Omega))$ be its distribution. Consider the structure $\mathcal{I}''=(\omega,t_i,\ldots, t_n)$. It  is equivalent to $\mathcal{I}'$ (and hence to $\mathcal{I}$) by Lemma~\ref{lm_posteriors_as_signals}.
As $t_i$ is a function of $s_i'$ and $\mathcal{I}'$ is not perfect, $\mathcal{I}''$ cannot be perfect either (this is also a consequence of the fact that $\mathcal{I}''$ is a garbling of $\mathcal{J}$).

Let $f\colon \, \Delta(\Omega)\times\ldots\times \Delta(\Omega)\to \Delta(\Omega)$ be the conditional distribution of $\omega$ given the realized signals $t_1,\ldots, t_n$. This function is defined $\mu$-everywhere with  $\mu=\mu_1\times\ldots\times\mu_n$. As $\mathcal{I}''$ is not perfect, there is a state $k_0 \in \Omega$ such that $f_{k_0}\notin\{0,1\}$ on a set of positive $\mu$-measure.

Choose a fuzzy partition $g\colon\ [0,1]^n\to \Delta(\Omega)$ so that the following identity holds:\footnote{To construct such a $g$, define $h_i \colon [0,1] \to \Delta(\Omega)$ by
$h_i(x_i)(k) = \Pr{\omega=k}{s_i=x_i}.$
That is, $h_i$ is the map that assigns to each signal realization the induced posterior, so that $p(s_i) = h_i(s_i)$ holds as an equality of random variables.
Then let $g\colon\ [0,1]^n\to \Delta(\Omega)$ be given by $
    g(x_1,\ldots,x_n) = f(h_1(x_1),\ldots,h_n(x_n))$.}
$g(s_1,\ldots, s_n)=f\big(p(s_1),\ldots, p(s_n)\big).$ The distributions of posteriors $(p(s_1),\ldots,p(s_n))$ and $(p(t_1),\ldots,p(t_n))$ both coincide with $\mu$ as the structures $\mathcal{I}$ and $\mathcal{I}''$ are equivalent. Hence, $g_{k_0}\ne\{0,1\}$ on a set of positive Lebesgue measure, i.e., $g$ does not correspond to a partition. On the other hand, $g$ has the same projections as the partition $\cA$. Indeed, let us compute $\alpha_i^{g_k}(x)$:
\begin{align*}
    \alpha_i^{g_k}(x)
    &=\E{g_k(s_1,\ldots,s_n)}{ s_i=x}\\
    &=\E{g_k(s_1,\ldots,s_{i-1},x,s_{i+1},\ldots,s_n)}\\
    &=\E{f_k(p(s_1),\ldots,p(s_{i-1}),q,p(s_{i+1}),\ldots,p(s_n))},
\end{align*}
where $q$ is the posterior induced by $s_i=x$. Since the distribution of $p(s_j)$ is identical to that of $t_j$,
\begin{align*}
    \alpha_i^{g_k}(x)
    &=\E{f_k(t_1,\ldots,t_{i-1},q,t_{i+1},\ldots,t_n)}\\
    &=\E{f_k(t_1,\ldots,t_n)\mid t_i=q}\\
    &=q(k),
\end{align*}
where in the last equality we rely on the fact that the belief induced by $t_i$ coincides with $t_i$. Since $q$ is the posterior induced by $s_i=x$, the posterior $q(k)$ is equal to~$\alpha_i^{A_k}(x)$.

 We thus constructed $g$ not equal to $(\one_{A_0},\ldots, \one_{A_{m-1}})$ but having the same projections. By Lemma~\ref{lm_uniqueness_as_extreme_point}, the partition $\cA=(A_0,\ldots, A_{m-1})$ cannot be a partition of uniqueness. This contradiction shows that no structure can  dominate the one associated with a partition of uniqueness, i.e., such structures are Blackwell-Pareto optimal.
 \end{proof}
The proof of Theorem~\ref{thm:pareto_appendix} is now immediate. 
\begin{proof}[Proof of Theorem~\ref{thm:pareto_appendix}]
 By Lemma~\ref{lm_PO_implies_uniqueness}, for each Blackwell-Pareto optimal $\mathcal{I}$, we can find an equivalent structure associated with  a partition of uniqueness $\cA=(A_0,\ldots, A_{m-1})$. By Lemma~\ref{lm_uniqueness_implies_PO}, any  structure admitting such an equivalent representation is Blackwell-Pareto optimal.
\end{proof}

\subsection{Proof of Theorem~\ref{th_Lorentz}}

\begin{proof}
\cite*{lorentz1949problem}'s characterization of two-dimensional sets of uniqueness uses the idea of a \emph{non-increasing rearrangement} $\grave \varphi$ of a function $\varphi:\, [0,1]\to [0,1]$. The function  $\grave \varphi$ is defined almost everywhere by the following two properties: it is non-increasing  on $[0,1]$ and,
for any $q\in [0,1]$, the lower-contour sets $\{t\in[0,1]\,:\, \varphi(t)\leq q\}$ and $\{t\in[0,1]\,:\, \grave \varphi(t)\leq q\}$ have the same Lebesgue measure. A non-increasing rearrangement exists and moreover is unique (as an element of $L^\infty([0,1])$).

\cite*{lorentz1949problem} proved that $A\subseteq [0,1]^2$ is a set of uniqueness if and only if  the non-increasing rearrangements of its two projections are inverses of each other, i.e.,
\begin{equation}\label{eq_Lorentzs_condition}
\grave \alpha_1^A=\big(\grave\alpha_2^A\big)^{-1}.
\end{equation}
Formally, if the inverse $\big(\grave\alpha_2^A\big)^{-1}(t)$ is not unique for some $t$, the equality~\eqref{eq_Lorentzs_condition} is to be understood as the inclusion: $\grave\alpha_2^A(t)\in \big(\grave\alpha_1^A\big)^{-1}(t)$.

Let us demonstrate that the characterization from Theorem~\ref{th_Lorentz} is equivalent to the original characterization of \cite*{lorentz1949problem}. That is, we need to check that a set $A$ is a rearrangement of an upward-closed set if and only if the condition~\eqref{eq_Lorentzs_condition} holds.
Note that for any downward-closed\footnote{A set $B\subseteq[0,1]^2$ is downward-closed if, with each point $(x_1,x_2)$, it contains all the points $(x_1',x_2')\in [0,1]^2$ such that $x_1'\leq x_1$ and $x_2'\leq x_2$.} set $B$, its image under  the map $x_1\mapsto 1-x_1$ and $x_2\mapsto 1-x_2$ is upward-closed. Hence, it is enough to check the equivalence between~\eqref{eq_Lorentzs_condition} and the existence of a downward-closed rearrangement of $A$.

Suppose that $A$ is a rearrangement of a downward-closed set $B$. Towards showing that~\eqref{eq_Lorentzs_condition} holds, note that any downward-closed set $B$ can be represented through its projections in two symmetric ways: $B=\{x_2\leq\alpha_1^B(x_1)\}$ and 
$B=\{x_1\leq\alpha_2^B(x_2)\}$ up to a zero-measure set. Hence,
\begin{equation}\label{eq_downward_identity}
\alpha_2^B=(\alpha_1^B)^{-1}.  
\end{equation}
Since  $B$ is downward-closed, its projections are non-increasing. Moreover, the sets $\{t\in[0,1]\,:\, \alpha_i^B(t)\leq q\}$ and $\{t\in[0,1]\,:\, \alpha_i^A(t)\leq q\}$ have the same measure for any $i$ and $q$ as $B$ is a rearrangement of $A$. Thus $\alpha_i^B=\grave\alpha_i^A$ and we obtain~\eqref{eq_Lorentzs_condition} from~\eqref{eq_downward_identity}.  

Now assume that the condition~\eqref{eq_Lorentzs_condition} is satisfied and construct the downward-closed set $B$ as follows:
$B=\{(x_1,x_2)\in [0,1]^2\,:\, x_2\leq \grave\alpha_1^A(x_1)\}.$
By the definition, the projection $\alpha_1^B$ equals $\grave\alpha_1^A$. For any downward closed set, the projections satisfy the identity~\eqref{eq_downward_identity} and thus
$\alpha_2^B=\big(\alpha_1^B\big)^{-1}=\big(\grave\alpha_1^A\big)^{-1}=\grave\alpha_2^A,$
where the last equality follows from~\eqref{eq_Lorentzs_condition}. Hence, for any $i$ and $q$, the measure of $\{\alpha_i^B(t)\leq q\}$ coincides with that of $\{\grave\alpha_i^A(t)\leq q\}$ and thus with the measure of $\{\alpha_i^A(t)\leq q\}$. We conclude that $B$ is a downward-closed rearrangement of $A$.
\end{proof}

\subsection{Proof of Theorem~\ref{thm:binary-pareto}}\label{sec_proof_th1}
First, we show that the conjugate of a cumulative distribution function on $[0,1]$ is also a cumulative distribution function.

\begin{claim}
\label{clm:conjugate}
The conjugate $\hat F$ is a cumulative distribution function. Furthermore, it has the same mean: $\int x\,\dd\hat F(x) = \int x\,\dd F(x)$.
\end{claim}
\begin{proof}
To show that $\hat F$ is a cumulative distribution function it suffices to show that it is weakly increasing, right-continuous, that $\hat F(0) \geq 0$, and that $\hat F(1)=1$.

We first note that $F^{-1}$ is weakly increasing, by its definition at $x$ as the minimum of the preimage of $[x,\infty)$ under $F$. Hence $\hat F$ is also weakly increasing.

To see that $\hat F$ is right continuous, let $\lim_k x_k = x \in [0,1]$, with $x_k \leq x$. Then
\begin{align*}
    \lim_k F^{-1}(x_k)
    &=  \lim_k \min\{y \,:\, F(y) \geq x_k\}\\
    &= \min\{y \,:\, F(y) \geq x\}\\
    &= F^{-1}(x)
\end{align*}
where the penultimate equality follows from the fact that $F$ is right-continuous. Hence $F^{-1}$ is left-continuous, and so $\hat F$ is right-continuous. 

It is immediate from the definitions that $\hat F(0) \geq 0$ and $\hat F(1)=1$, and thus $F$ is a cumulative distribution function. Finally, the expectations of $F$ and $\hat F$ are identical since the shape under $F$ (whose measure is equal to its expectation), given by $\{(x,y) \in [0,1]^2\,:\, y \leq F(x)\}$ is congruent to the shape under $\hat F$, since one is mapped to the other by the measure-preserving transformation $(x,y) \mapsto (1-y,1-x)$.
\end{proof}

Now, we prove Theorem~\ref{thm:binary-pareto}.

\begin{proof}
First, suppose $\mathcal{I}$ is Blackwell-Pareto optimal. 
By Theorem \ref{thm:pareto}, $\mathcal{I}$ is equivalent to some structure $\mathcal{I}'$ associated with a set of uniqueness
$A$. By Theorem \ref{th_Lorentz}, $A$ is a rearrangement of
 an upward-closed set $A'$, whose associated structure $\mathcal{I}''$ must
also be equivalent to $\mathcal{I}$. We show that the two 
belief distributions induced by $\mathcal{I}''$ are conjugates of each other.

Define $\tilde{h}:[0,1]\to[0,1]$ by $\tilde{h}(x_{1})=\inf\{x_{2}:(x_{1},x_{2})\in A'\}$.
We have that $\tilde{h}$ is a decreasing function since $A'$ is upward-closed. Define a left-continuous version of $\tilde{h}$ as $h(x)=\lim_{z\to x^{-}}\tilde{h}(z)$.
For any $q\in[0,1]$, in the structure associated with $A'$, {the belief induced by the realization of the first signal equal to $x_1$ is} 
lower or equal to~$q$ if and only if $h(x_{1})\ge1-q$, so the cumulative distribution function
of {beliefs} is $F_{1}(q)=\max\{x_{1}:h(x_{1})\ge1-q\}.$
{For the second signal, the belief for any signal realization} lower
than $x_{2}=h(1-q)$ is lower or equal to~$q$, while beliefs at
higher signals are strictly {above} $q.$ So, the cumulative distribution
function of {beliefs induced by the second signal}
is $F_{2}(q)=h(1-q).$ Note
that $F_{2}^{-1}(1-q)=\min\{y:h(1-y)\ge1-q\}=1-\max\{x_{1}:h(x_{1})\ge1-q\}=1-F_{1}(q),$
so $F_{1}$ and $F_{2}$ are conjugates.

Conversely, suppose the distributions of $p(s_{1})$ and $p(s_{2})$
in a private private information structure $\mathcal{I}$ are conjugates. Write
$\tilde{F}_{1}$ and $\tilde{F}_{2}$ for the cumulative distribution functions
of $p(s_{1})$ and $p(s_{2})$, and consider the set $A\subseteq[0,1]^{2}$
where $(x_{1},x_{2})\in A$ if and only if $x_{2}\ge\tilde{F}_{2}(1-x_{1}).$
We show that the structure associated with $A$ is equivalent to~$\mathcal{I}$; Figure~\ref{fig:induce_conjugates} illustrates the construction. 
Let $\tilde{h}(x)=\tilde{F}_{2}(1-x)$, and define a left-continuous
version of $\tilde{h}$ as $h(x)=\lim_{z\to x^{-}}\tilde{h}(z)$.
For the structure associated with $A$, by the same argument as above,
the distribution function of {beliefs induced by the second signal}
is $F_{2}(q)=h(1-q)=\tilde{F}_{2}(q).$
The distribution function of {beliefs induced by the first signal}
is $F_{1}(q)=\max\{x_{1}:h(x_{1})\ge1-q\}=1-\min\{x_2:\tilde{F}_{2}(x_2)\ge1-q\}=1-\tilde{F}_{2}^{-1}(1-q).$
Using the hypothesis that $\tilde{F}_{1}$ and $\tilde{F}_{2}$ are
conjugates, $1-\tilde{F}_{2}^{-1}(1-q)=\tilde{F}_{1}(q).$ So, the structure associated with $A$ is equivalent to $\mathcal{I}$. Because
$x_{1}\mapsto \tilde{F}_{2}(1-x_{1})$ is a decreasing function, the set $A$
is upward-closed. Using Theorem~\ref{thm:pareto} and Theorem~\ref{th_Lorentz},
$\mathcal{I}$ is Blackwell-Pareto optimal. 
\end{proof}

\subsection{Proof of Theorem~\ref{thm:disclosure}}\label{app_disclosure_proof}
We first demonstrate that the definition of an {optimal privacy-preserving recommendation} (Definition~\ref{def_disclosure}) is equivalent to the requirement of Blackwell-Pareto optimality of the corresponding private private information structure.  For this purpose, we need the following result refining Lemma~\ref{lem:pareto-dominated} and applicable for any number of states.
\begin{lemma}\label{lm_single_improvement}
If $\mathcal{I}=(\omega,s_1,s_2)$ is a private private information structure that is not Blackwell-Pareto optimal, then there exists $s_2'$ independent of $s_1$ such that the private private structure $(\omega, s_1,s_2')$ dominates $\mathcal{I}$ and is Blackwell-Pareto optimal.
\end{lemma}
\begin{proof}
Denote by $\mu_1$ and $\mu_2$ the distributions of beliefs induced by $s_1$ and $s_2$, respectively. Let us verify the existence of a Blackwell-Pareto optimal structure  $\mathcal{J}=(\omega,t_1,t_2)$ such that $p(t_1)$ is distributed according to $\mu_1$ and $t_2$ dominates $s_2$.
Consider the set {${\F_{2|1}}(\mu_1)$} of distributions $\mu_2'$ of beliefs such that the pair $(\mu_1,\mu_2')$ is {feasible under the constraint of privacy}, i.e., there is a private private information structure inducing these distributions. In particular, $\mu_2$ belongs to {${\F_{2|1}}(\mu_1)$}.
By Lemma~\ref{lm_feasible_is_closed}, {${\F_{2|1}}(\mu_1)$} is compact in the weak topology  as a closed subset of the set of feasible pairs $\F_2$. Since the Blackwell order is continuous in the weak topology (see the proof of Lemma~\ref{lem:pareto-dominated}), there is a maximal element $\mu_2'\in{\F_{2|1}(\mu_1)}$ dominating~$\mu_2$. Let $\mathcal{J}=(\omega,t_1,t_2)$ be the private private information structure inducing the pair of distributions $(\mu_1, \mu_2')$. The structure $\mathcal{J}$ must be Blackwell-Pareto optimal. Else, by Lemma~\ref{lm_garbling_imperfection}, there is an equivalent structure $\mathcal{J}'=(\omega,t_1',t_2')$ where the signals do not determine the state. Then, by the construction from Lemma~\ref{lm_PO_implies_perfection}, there exists an informative signal $t$ independent of $(t_1',t_2')$. By 
{augmenting the second signal with $t$,}
we obtain a strict Blackwell-Pareto improvement of $\mathcal{J}'$ where the distribution of beliefs induced by the first signal remains fixed, but the distribution of beliefs induced by the second signal is improved to $\mu_2''$. So we have $\mu_2''\in{\F_{2|1}}(\mu_1)$ and $\mu_2''$ strictly dominates $\mu_2'$, which contradicts the maximality of $\mu_2'$ in ${\F_{2|1}}(\mu_1)$.
This contradiction implies that  $\mathcal{J}$ is Blackwell-Pareto optimal.

Without loss of generality, we can assume that $s_1=p(s_1)$ and $t_1=p(t_1)$, i.e., signals coincide with the induced posteriors. Thus, by the equivalence of $s_1$ and $t_1$, the joint distributions of $(\omega,s_1)$ and $(\omega,t_1)$ are the same. We complete the proof by defining $s_2'$ so that the joint distributions of $(\omega,s_1,s_2')$ and $(\omega, t_1,t_2)$ coincide.
\end{proof}
\begin{corollary}\label{cor_disclosure_def_equivalence}
For any number of states, given a one-signal information structure $(\omega,s_1)$, a signal $s_2$ is an optimal privacy-preserving recommendation if and only if $\mathcal{I}=(\omega,s_1,s_2)$ is a Blackwell-Pareto optimal private private information structure.
\end{corollary}
\begin{proof}
If $\mathcal{I}$ is Blackwell-Pareto optimal, then it cannot be dominated by any private private structure, in particular, by a structure of the form $(\omega, s_1,s_2')$. Therefore, $s_2$ is an optimal privacy-preserving recommendation.
Conversely, if $\mathcal{I}$ is not Blackwell-Pareto optimal, then, by Lemma~\ref{lm_single_improvement}, there is a dominating private private structure of the form $(\omega, s_1,s_2')$. Thus $s_2$ is dominated by $s_2'$, and so it is not an optimal privacy-preserving recommendation.
\end{proof}
We  are now ready to prove Theorem~\ref{thm:disclosure}.
\begin{proof}[Proof of Theorem~\ref{thm:disclosure}]
We are given $(\omega,s_1)$ and aim to construct a {dominating recommendation $s_2^\star$, i.e., a new signal independent of $s_1$  such that  any other signal $s_2$ independent of $s_1$ is weakly dominated by~$s_2^\star$; see Definition~\ref{def_dominant}.}

As usual, $p(s_1)$ is the belief induced by $s_1$. We sample $s_2^\star$ uniformly from the interval $[1-p(s_1),\,1]$ if the state is $\omega=1$ and from  $[0,\,1-p(s_1)]$ if $\omega=0$. Hence, conditioned on $s_1$, the constructed signal is distributed uniformly on $[0,1]$ and so   $s_2^\star$ is independent of $s_1$. Denote by  $F$ the cumulative distribution function of $p(s_1)$ and compute the belief induced by $s_2^\star$. The conditional probability of  $\omega=1$ given $s_2^\star=t$ is equal to $\Pr{1-p(s_1)\leq t}$. Hence, $p(s_2^\star)=1-F(1-s_2^\star)$. Thus the distribution function $F^\star$ of $p(s_2^\star)$ is given by
$$F^\star(x)=\Pr{p(s_2^\star)\leq x}=\Pr{1-F(1-s_2^\star)\leq x}=\Pr{s_2^\star\leq 1- F^{-1}(1-x)},$$
where $F^{-1}$ is defined as in \eqref{eq:f-inv}.  
Since $s_2^\star$ is uniformly distributed on $[0,1]$, we get
$F^\star(x)=1- F^{-1}(1-x)=\hat{F}(x),$
where $\hat{F}$ is the conjugate of $F$.

We conclude that $(\omega,s_1,s_2^\star)$ is a private private information structure and the distributions of posteriors induced by $s_1$ and $s_2^\star$ are conjugates. Therefore, $(\omega,s_1,s_2^\star)$ is Blackwell-Pareto optimal by Theorem~\ref{thm:binary-pareto}. Corollary~\ref{cor_disclosure_def_equivalence} implies that $s_2^\star$ is an optimal privacy-preserving recommendation.

Now, we show that any $s_2$ independent of $s_1$ is (weakly) dominated by $s_2^\star$. If $(\omega,s_1,s_2)$ is itself Blackwell-Pareto optimal, then,  by Theorem~\ref{thm:binary-pareto}, the cumulative distribution function of beliefs induced by $s_2$ is $\hat{F}$ and thus $s_2$ is equivalent to $s_2^\star$. 
Hence, it suffices to consider the case where $\mathcal{I}=(\omega,s_1,s_2)$ is not Blackwell-Pareto optimal. By Lemma~\ref{lm_single_improvement}, there exists an $s_2'$ dominating $s_2$ such that  $(\omega,s_1,s_2')$ is a Blackwell-Pareto optimal private private structure. 
Hence, by Theorem~\ref{thm:binary-pareto}, the distribution of beliefs induced by $s_2'$ is the conjugate of~$F$, {and so $s_2'$ is equivalent to~$s_2^\star$. We conclude that $s_2^\star$ dominates any $s_2$ independent of $s_1$, and thus $s_2^\star$ is a dominant privacy-preserving recommendation.}
\end{proof}

\subsection{Proof of Proposition~\ref{prop_optimal_utility}}\label{app_optimal_utility}

The indirect utility of the decision maker $U(q) =\sup_{a\in A} \big((1-q)\cdot u(a,0)+ q\cdot u(a,1)\big)$ is a continuous convex function as $u$ is assumed to be bounded from above. By a standard argument, the expected payoff of the decision maker observing a signal $s$ is equal to $\E{U(p(s))}$.
 By Theorem~\ref{thm:disclosure}, the distribution of $p(s)$ for the {dominant} privacy-preserving recommendation $s=s_2^\star$ is the conjugate to that of~$p(s_1)$. Thus  Proposition~\ref{prop_optimal_utility} is implied by the following lemma about integrating with respect to conjugate distributions.
\begin{lemma}\label{lm_change_of_variables}
Let $F$ and $\hat{F}$ be a pair of conjugate distributions, and  $U$ be a  continuous function. The following identity holds
$$\int_{[0,1]}U(q)\,\dd\hat{F}(q)=\int_{[0,1]}U\big(1-F(t)\big)\,\dd t.$$
\end{lemma}
\begin{proof}
Assume first that  $F$ is a bijection $[0,1]\to[0,1]$. By the definition of a conjugate, we obtain \begin{equation}\label{eq_change_of_variable}
\int_{[0,1]}U(q)\,\dd\hat{F}(q)=\int_{[0,1]}U(q)\,\dd\left(1-F^{-1}(1-q)\right)=\int_{[0,1]}U\big(1-F(t)\big)\,\dd t,
\end{equation}
where we changed the variable $q=1-F(t)$ in the second equality. 

Now we show that the  identities~\eqref{eq_change_of_variable} hold even without the assumption that $F$ is a bijection. Since any continuous function on $[0,1]$ can be approximated by a linear combination of indicators $\one_{[0,a]}$ in the $\sup$-norm, it is enough to prove that
$$\int_{[0,1]}\one_{[0,a]}(q) \dd\hat{F}(q)=\int_{[0,1]}\one_{[0,a]}\big(1-F(t)\big)\dd t$$ or, equivalently, that
\begin{equation}\label{eq_check_on_indicators}
\hat{F}(a)=\lambda\big(\{t\in [0,1]:\, 1-F(t)\leq a \}\big),
\end{equation}
where $\lambda$ stands for the Lebesgue measure. By the monotonicity of $F$, the set from the right-hand side of~\eqref{eq_check_on_indicators} is an interval $[t_a,1]$ where $t_a=\min\{t: F(t)\geq 1-a\}$, i.e., $t_a=F^{-1}(1-a)$ as defined in~\S\ref{sec_Pareto_for_two}. We conclude that~\eqref{eq_check_on_indicators} holds as it is equivalent to the equality $\hat{F}(a)=1-F^{-1}(1-a)$ defining the conjugate distribution and thus~\eqref{eq_change_of_variable} holds as well.
\end{proof}

{
\subsection{Proof of Lemma~\ref{lm_convex_U}}\label{sec_proof_lm_convex_U}
\begin{proof}
The upper semicontinuity of $U$ and the compactness of the set of belief distributions that can be induced by private private structures (Lemma~\ref{lm_feasible_is_closed}) ensure that the optimum is attained at some private private structure $\mathcal{I}=(\omega, s_1,\ldots, s_n)$; see Lemma~4.3 in \cite{villani2008optimal} for a proof that integration against an upper semicontinuous function defines an upper semicontinuous functional over distributions, endowed with the weak topology.
By the same compactness argument, there exists a Blackwell-Pareto optimal private private structure $\mathcal{I}^*=(\omega, s_1^*,\ldots, s_n^*)$ such that each signal $s_j^*$ is equivalent to $s_j$ for $j\ne i$, and  $s_i^*$ weakly dominates $s_i$. 
By the convexity of~$U$ in $q_{i}$,  the expected value $\E{U}$ cannot decrease when $i$'s signal becomes Blackwell-more informative. Since $\mathcal{I}$ is optimal for the designer and the value of $\mathcal{I}^*$ is at least as high, $\mathcal{I}^*$ is also optimal.
\end{proof}
}

\subsection{Proof of Proposition~\ref{prop:ladder}}\label{app_welfare}

\begin{proof}

{Consider the problem of maximizing $W(\mathcal{I})=\E{U_1\big(p(s_1)\big)+U_2\big(p(s_2)\big)}$ with continuous $U_1,U_2$ and convex $U_2$
over private private information structures $\mathcal{I}=(\omega,s_1,s_2)$. 
By Lemma~\ref{lm_convex_U}, the convexity of the objective in one of the arguments implies that the optimum is attained at a Blackwell-Pareto optimal $\mathcal{I}$.} By Theorem~\ref{thm:binary-pareto}, the distributions of posteriors $p(s_1)$ and $p(s_2)$ induced by $\mathcal{I}$ are conjugates. Denote the distribution of $p(s_1)$ by $\mu$, which can be an arbitrary measure on $[0,1]$ with mean equal to the prior $p$. Denote the set of all such measures by $\Delta_p([0,1])$. The choice of $\mu\in \Delta_p([0,1])$ determines the distribution $\hat{\mu}$ of $p(s_2)$. Thus, to maximize $W(\mathcal{I})$ over $\mathcal{I}$ it is enough to find $\mu\in \Delta_p([0,1])$ maximizing the functional
\begin{equation}\label{eq_w_mu}
w(\mu)=\int_{[0,1]}U_1(q)\dd\mu(q)+\int_{[0,1]}U_2(q)\dd\hat{\mu}(q).
\end{equation}
Below, we check that $w(\mu)$ is convex and continuous in the weak topology. Hence, by Bauer's principle, the optimum is attained at an extreme point of $\Delta_p([0,1])$. It is well-known that the extreme points of this set are measures with the support of size at most two: see, e.g., \cite*{winkler1988extreme}. Since the optimal $\mu$ is supported on at most two points, its conjugate $\hat{\mu}$ is supported on at most three points (see the discussion after Theorem~\ref{thm:binary-pareto}) and we conclude that there is an optimal structure $\mathcal{I}$ where $s_1$ takes at most two values and $s_2$ takes at most three values.

{We now check that $w$ is convex and continuous in the weak topology. The first integral in~\eqref{eq_w_mu} is linear in $\mu$ (hence, convex) and continuous thanks to the continuity of the integrated.  We, therefore, focus on the second integral. By Lemma~\ref{lm_change_of_variables}, it can be expressed through the cumulative distribution function $F$ of~$\mu$ as follows 
\begin{equation}\label{eq_2nd_integral}
\int_{[0,1]}U_2(q)\dd\hat{\mu}(q)=\int_{[0,1]} U_2\big(1-F(q)\big) \dd q.
\end{equation}
It is a convex function of $F$ by convexity of $U_2$. 
To show its continuity,} note that the weak convergence $\mu_k\to \mu$ implies the convergence of $F_k(q)\to F(q)$ for all points $q$ of continuity of $F$ \cite[see][Theorem 15.3]{aliprantis2006infinite}.
Since any monotone function is continuous almost everywhere with respect to the Lebesgue measure, the sequence of functions $U_2\big(1-F_k\big)$ converges almost everywhere in $[0,1]$ and is bounded thanks to the boundedness of $U_2$. The Lebesgue dominated convergence theorem implies that $\int_{[0,1]}U_2\big(1-F_k(q)\big)\dd q$ converges to $\int_{[0,1]}U_2\big(1-F(q)\big)\dd q$. {We conclude that the second integral in~\eqref{eq_w_mu} is a convex continuous function of $\mu$, and thus so is $w(\mu)$ itself.}

{It remains to show that there is an optimal structure with binary signals if both $U_1$ and $U_2$ are convex. As a first step, we demonstrate uniqueness of optimal belief distributions for strictly convex  $U_2$ without any assumptions on $U_1$.
It is enough to show that $w(\mu)$ is a strictly convex functional of $\mu$, i.e.,
$$w(\mu)> \alpha\cdot w(\nu)+(1-\alpha) w(\nu') \quad\text{with}\quad \mu=\alpha\cdot\nu+(1-\alpha)\nu'$$
for all distinct $\nu,\nu'\in\Delta_p([0,1])$ and $\alpha\in(0,1).$
The first integral  in~\eqref{eq_w_mu} is affine in~$\mu$, and thus we need to show strict convexity of the second integral. Denote the cumulative distribution functions of $\nu$, $\nu'$, and $\mu$ by $G$, $G'$, and $F=\alpha\cdot G+(1-\alpha)G'$. By~\eqref{eq_2nd_integral}, we get
\begin{align*}
\int_{[0,1]}U_2(q)\dd\hat{\mu}(q)=\int_{[0,1]} U_2\Big(1-\alpha\cdot G(q)+(1-\alpha)G'(q)\Big) \dd q.
\end{align*}
By the assumption that $\nu\ne \nu'$, the functions $G$ and $G'$ differ on a set of positive Lebesgue measure. On this set, by strict convexity of $f(t)=U_2(1-t)$, we get 
$$U_2\big(1-\alpha\cdot G(q)+(1-\alpha)G'(q)\big)>\alpha\cdot U_2\big(1-G(q)\big)+(1-\alpha) U_2\big(1-G'(q)\big).$$
We conclude that 
$$
\int_{[0,1]}U_2(q)\dd\hat{\mu}(q)>\alpha \int_{[0,1]}U_2(q)\dd\hat{\nu}(q)+(1-\alpha) \int_{[0,1]}U_2(q)\dd\widehat{\nu'}(q).$$
Thus the functional $w(\mu)$ is strictly convex in $\mu$, and so the optimal $\mu$ is unique.

Now suppose that $U_1$ and  $U_2$ are strictly convex. By convexity of~$U_2$, there is an optimal private private structure $(\omega,s_1,s_2)$ inducing distributions of posteriors $\mu$ and $\hat{\mu}$ such that $\mu$ is supported on at most two points and $\hat \mu$ on at most three. Exchanging the roles of signals $s_1$ and $s_2$ and using convexity of $U_1$, we conclude that there is another optimal structure inducing distributions $\mu'$ and $\widehat{\mu'}$ with~$\mu'$ supported on at most three points and $\widehat{\mu'}$, on at most two. By the strict convexity, the optimal distributions are unique and so $\mu=\mu'$ and $\hat{\mu}=\widehat{\mu'}$. We obtain that, in $(\omega,s_1,s_2)$, each signal induces at most two different posteriors and thus this structure is equivalent to a structure with each signal taking at most two different values.

Finally, we relax the assumption of strict convexity. Consider convex but not necessarily strictly convex~$U_1$ and $U_2$. For $\varepsilon>0$, define $U_i^\varepsilon(q)=U_i(q)+\varepsilon\cdot q^2$. By strict convexity of  $U_i^\varepsilon$, the corresponding problem admits optimal belief distributions $\mu^\varepsilon$ and $\widehat{\mu^\varepsilon}$ with at most two atoms. Extracting a weakly-convergent subsequence as~$\varepsilon\to 0$, we deduce the existence of a solution to the original problem with binary signals.}
\end{proof}

\subsection{Proof of Corollary~\ref{prop:comparative}}
\begin{proof}
Since $(\omega,s_1,s_2)$ is Blackwell-Pareto optimal, the signal $s_2$ can be seen as an optimal privacy-preserving recommendation corresponding to $(\omega,s_1)$. By Theorem~\ref{thm:disclosure}, such $s_2$ dominates any other signal $s_2'$ independent of $s_1$. Hence, to conclude that $s_2$ dominates $t_2$, it is enough to demonstrate that there is a private private information structure $(\omega,s_1',s_2')$ such that $s_1'$ is equivalent to $s_1$ and $s_2'$ is equivalent to $t_2$. By the assumption, $t_1$ dominates $s_1$ and, therefore, the signal $s_1$ is equivalent to some garbling $s_1'$ of $t_1$. Putting $s_2'=t_2$, we get the desired private private information structure $(\omega,s_1',s_2')$ and deduce that $t_2$ is dominated by~$s_2$.
\end{proof}

\subsection{Proof of Proposition~\ref{prop:mutual-info}}\label{app:mutual-info}

\begin{proof}[Proof of Proposition~\ref{prop:mutual-info}]
We have $I(\omega; (s_1, \ldots, s_n)) \le H(p)$. By Lemma~\ref{lm_supperadditive} proved in Appendix~\ref{app_prelim_lemmas}, $\sum_{i=1}^n I(\omega;s_i)\leq I(\omega; (s_1, \ldots, s_n))$ provided that the signals are independent.
\end{proof}

\subsection{Proof of Proposition~\ref{prop:mutual-info-two}}\label{app:mutual-info-two}

\begin{proof}
We have $I(\omega; (s_1, \ldots, s_n)) \le H(p)$, so it suffices to show that 
\begin{equation}\label{eq_superadditive_strict}
    \sum_i I(\omega;s_i) \leq I(\omega; (s_1, \ldots, s_n)) - c_p\sum_{i < j}I(\omega;s_i)I(\omega;s_j). 
\end{equation}
Similarly to the proof of Lemma~\ref{lm_supperadditive}, the result for general $n$ follows from the result for $n=2$ via an inductive argument. Indeed, assume that the statement holds for $n\leq n_0$ with $n_0\geq 2$ and show that it holds for $n=n_0+1$ as well:
\begin{align*}
\lefteqn{I\Big(\omega\,;\,\big(s_1,\ldots,s_{n_0}, s_{n_0+1}\big)\Big)}\\
&=I\Big(\omega\,;\,\big((s_1,\ldots,s_{n_0}), s_{n_0+1}\big)\Big)\\ 
&\geq I\Big(\omega\,;\,\big(s_1,\ldots,s_{n_0}\big)\Big)+I(\omega\,;\,s_{n_0+1})+c_p\cdot  I\Big(\omega\,;\,\big(s_1,\ldots,s_{n_0}\big)\Big)\cdot I(\omega\,;\,s_{n_0+1}),
\end{align*}
where we applied the two signal version of~\eqref{eq_superadditive_strict} for the pair of signals $(s_1,\ldots, s_{n_0})$ and $s_{n_0+1}$. Estimating $I(\omega\,;\,s_1,\ldots,s_{n_0})$ from below via the $n_0$-signal version of~\eqref{eq_superadditive_strict}, we get 
\begin{align*}
\lefteqn{I\Big(\omega\,;\,\big(s_1,\ldots,s_{n_0}, s_{n_0+1}\big)\Big)}\\
&\geq 
\sum_{i=1}^{n_0} I\big(\omega\,;\,s_i\big)+c_p\sum_{1\leq i<j\leq n_0}I\big(\omega\,;\,s_i\big)\cdot I\big(\omega\,;\,s_j\big)+ I(\omega\,;\,s_{n_0+1})+\\
&\quad+c_p \cdot I(\omega\,;\,s_{n_0+1})\cdot \left(\sum_{i=1}^{n_0} I\big(\omega\,;\,s_i\big)+c_p\sum_{1\leq i<j\leq n_0}I\big(\omega\,;\,s_i\big)\cdot I\big(\omega\,;\,s_j\big)\right)
\end{align*}
Eliminating all the cubic terms from the second line can only decrease the right-hand side and leads to inequality~\eqref{eq_superadditive_strict} for $n=n_0+1$:
$$I\Big(\omega\,;\,\big(s_1,\ldots,s_{n_0}, s_{n_0+1}\big)\Big)\geq 
\sum_{i=1}^{n_0+1} I\big(\omega\,;\,s_i\big)+c_p\sum_{1\leq i<j\leq n_0+1}I\big(\omega\,;\,s_i\big)\cdot I\big(\omega\,;\,s_j\big).$$
\smallskip

It thus remains to prove the result for $n=2$. We aim to show that 
\begin{equation}\label{eq_strict_supperadditive}
    I(\omega\,;\,s_1)+I(\omega\,;\,s_2)-I(\omega\,;\,s_1,s_2)\leq -c_p\cdot  I(\omega\,;\,s_1)\cdot I(\omega\,;\,s_2).
\end{equation}
Denote the left-hand side of~\eqref{eq_strict_supperadditive} by $\Delta$ and the posterior probabilities of the high state by $p_i=\Pr{\omega=1}{ s_i}$ and $p_{12}=\Pr{\omega=1}{s_1,s_2}$. By the martingale property, $\E{p_{12}}{s_i}=p_i$ and $\E{p_i}=p$. Thanks to the martingale property, we can represent $I(\omega\,;\,s_i)$ as follows:
$$I(\omega\,;\,s_i)=\E{p_{12}\log_2\left(\frac{p_i}{p}\right)+ (1-p_{12})\log_2\left(\frac{1-p_i}{1-p}\right)},$$
where $p_i$ outside of the logarithm was replaced by $p_{12}$. Hence,
$$\Delta=\E{p_{12}\log_2\left(\frac{p_1\cdot p_2}{p_{12}\cdot p}\right)+ (1-p_{12})\log_2\left(\frac{(1-p_1)(1-p_2)}{(1-p_{12})(1-p)}\right)}.$$
By the concavity of the logarithm, a convex combination of logarithms is at most the logarithm of the convex combination. Therefore,
$$\Delta\leq \E{\log_2\left(\frac{p_1p_2}{p}+\frac{(1-p_1)(1-p_2)}{(1-p)}\right)}.$$
Denote the centered posteriors by $\bar p_1=p_1-p$ and $\bar p_2=p_2-p$. The right-hand side simplifies to
$$\E{\log_2\left(\frac{p_1\cdot p_2}{p}+\frac{(1-p_1)(1-p_2)}{(1-p)}\right)}=\E{\log_2\left(1+\frac{\bar p_1\cdot \bar p_2 }{p(1-p)}\right)}.$$
Note that $\frac{\bar p_1\cdot \bar p_2}{p(1-p)}$ belongs to the interval $\left[-1,\max\left\{\frac{1-p}{p},\frac{p}{1-p}\right\}\right]$. 
Consider the function $f(x)=\log_2(1+x)$. By the Taylor formula, for any $x> -1$, 
$$f(x)=f(0)+f'(0)\cdot x+ \frac{f''(y)}{2}x^2$$
for some $y$ between $0$ and $x$. Computing the derivatives, we get
$$f(x)=\frac{1}{\ln 2}x+\frac{1}{2\ln 2}\frac{-1}{(1+y)^2}x^2\leq \frac{1}{\ln 2}x-\frac{\min\{p^2,\, (1-p)^2\}}{2\ln 2}x^2,$$
where in the last inequality we used the fact that $y\in \left[-1,\max\left\{\frac{1-p}{p},\frac{p}{1-p}\right\}\right]$. Taking into account that $\min\{p^2,\, (1-p)^2\}\geq 4p^2(1-p)^2$, we obtain 
$$f(x)\leq \frac{1}{\ln 2}x-\frac{2p^2(1-p)^2}{\ln 2}x^2$$
and conclude that
$$\E{\log_2\left(1+\frac{\bar p_1\cdot \bar p_2}{p(1-p)}\right)}\leq \frac{1}{\ln 2}\E{\frac{\bar p_1\cdot \bar p_2}{p(1-p)}}-\frac{2p^2(1-p)^2}{\ln2 }\E{\left(\frac{\bar p_1\cdot \bar p_2}{p(1-p)}\right)^2}.$$
Since the expectation of the product  is the product of expectations for the independent random variables $\bar p_1$ and $\bar p_2$, 
$$\Delta\leq -\frac{2}{\ln 2}\var[p_1]\cdot\var[p_2].$$
It remains to lower-bound the variance by the mutual information. The Kullback–Leibler divergence between Bernoulli random variables with success probabilities $p$ and $x$ is defined as follows:  
$D_\mathrm{KL}(x||p)=x\log_2\left(\frac{x}{p}\right)+(1-x)\log_2\left(\frac{1-x}{1-p}\right)$. Then $I(\omega\,;\,s_i)=\E{D_\mathrm{KL}(p_i||p)}$. Applying the inequality $\ln t\leq t-1$ to both logarithms
and taking into account that $\log_2 t =\frac{1}{\ln 2}\ln t$, we obtain 
$$D_\mathrm{KL}(x||p)\leq \frac{1}{\ln 2}\left(x\cdot\left( \frac{x}{p}-1\right)+(1-x)\cdot\left( \frac{1-x}{1-p}-1\right)\right)= \frac{1}{p(1-p)\ln 2}(x-p)^2$$
for $x\in[0,1]$. Therefore, $$\var[p_i]\geq \left(p(1-p)\ln 2\right)\cdot  I(\omega\,;\,s_i)$$ and we conclude that
$$\Delta \leq -2\ln 2\cdot p^2(1-p)^2\cdot I(\omega\,;\,s_1)\cdot I(\omega\,;\,s_2)=-c_p\cdot I(\omega\,;\,s_1)\cdot I(\omega\,;\,s_2),$$
which is equivalent to the desired inequality~\eqref{eq_strict_supperadditive}.
\end{proof}

{
\subsection{Proof of Claim~\ref{claim_feasible_reduced_forms}}\label{app:proof_reduced_forms}

\begin{proof}
Each social choice rule ${C}$ defines a joint distribution of the outcome $\omega$ and the types $(t_1,\ldots,t_n$). Interpreting $\omega$ as the state and types as independent signals about~$\omega$, we conclude that the posterior belief $p(t_i)$ equals the reduced-form rule ${C}_i(t_i)$. This equivalence completes the proof. 
\end{proof}
}

\section{{Signal Informativeness Measured by Variance Reduction}}\label{app_quadratic_informativeness}

{In \S\ref{sec_responsive}, we measured signal informativeness by mutual information. Mutual information is the expected utility associated with a particular decision problem: one in which the indirect utility is given by a constant minus the entropy. Here we consider quadratic indirect utility, which results in an alternative measure of signal informativeness.}

{Let $\omega\in \Omega$ be a random state taking a finite number of different values and 
$s$ be a signal about $\omega$ inducing a belief $p(s)\in \Delta(\Omega)$. Denote $\bar H(q) = \sum_{k \in \Omega}q(k)(1-q(k))$ for~$q \in \Delta(\Omega)$. Analogously to mutual information defined in \eqref{eq:I:measure}, consider
\begin{align*}
    \bar I(\omega;s) = \bar H\left(\E{p(s)}\right)-\E{\bar H(p(s))}.
\end{align*}
Loosely speaking, $\bar I(\omega;s)$  is the expected reduction in the variance in $\omega$ after observing~$s$.}

{We establish an analog of Proposition~\ref{prop:mutual-info} for $\bar I(\omega;s)$ thus obtaining new bounds on causal strength or the social choice rule's responsiveness as discussed in \S\ref{sec_responsive} and getting additional necessary conditions for feasibility of belief distributions under the privacy constraint.}

\begin{proposition}
\label{prop:quadartic-info}
{For any finite $\Omega$ 
and a private private structure $(\omega, s_1,\ldots, s_n)$, 
 $$\sum_i \bar I(\omega; s_i) \leq \bar H(p).$$
 where $p\in \Delta(\Omega)$ is the prior distribution of $\omega$.}
\end{proposition}

{While this statement of Proposition~\ref{prop:quadartic-info} is completely analogous to that of Proposition~\ref{prop:mutual-info}, } the proof uses a different technique, exploiting the  $L^2$ orthogonality of independent random variables. Indeed, we do not know of a unifying argument that implies both propositions, and we furthermore do not know of additional decision problems that yield analogous statements. We note that Proposition~\ref{prop:quadartic-info}  is a generalization---from the binary state case---of the ``concentration of dependence'' principle of \cite*{mossel2020social}. A very similar idea appeared earlier in the economics literature \citep*{al2000pivotal} and is standard in the analysis of Boolean functions \citep[see, e.g.,][]{kahn1988influence, o2014analysis}.

\begin{proof}
As in the proof of Proposition~\ref{prop:mutual-info}, we show a stronger statement:
\begin{align*}
    \sum_i \bar I(\omega ;s_i) \leq \bar I(\omega; (s_1,\ldots,s_n)).
\end{align*}
This implies the statement of Proposition~\ref{prop:quadartic-info} since $\bar H$ is concave, and so, as with mutual information, $\bar I(\omega; (s_1,\ldots,s_n)) \leq \bar H(p)$.

Applying the definition of $\bar I$, and using the martingale property $$\E{p(s_i)} = \E{p(s_1,\ldots,s_n)} = p,$$ what we want to prove is that
\begin{align*}
    \sum_i \sum_{k \in \Omega} \E{[p(s_i)(k)-p(k)]^2} \leq \sum_k \E{[p(s_1,\ldots,s_n)(k)-p(k)]^2}.
\end{align*}
In fact, we prove an even stronger statement, showing that the inequality holds already for each $k \in \Omega$ separately, rather than only when summed over $\Omega$. 

To this end, fix $k$, and denote the centered posteriors by $\bar p_i = p(s_i)(k)-p(k)$, so that $\bar p_i$ is a zero-mean bounded random variable. Likewise denote $\bar p = p(s_1,\ldots,s_n)(k)-p(k)$. We want to prove that $\E{\bar p^2} \geq \sum_i \E{\bar p_i^2}$.

Let $V$ be the vector space of zero-mean random variables spanned by $\{\bar p,\bar p_1,\ldots,\bar p_n\}$. As a subspace of $L^2$, it is endowed with the inner product given by the expectation of the product. 

Since the structure is private private, $\E{\bar p_i \cdot \bar p_j} = \E{\bar p_i}\cdot \E{\bar p_j}=0$ for $i \neq j$. That is, the vectors $\{\bar p_1,\ldots,\bar p_n\}$ are orthogonal. Hence, $V = \mathrm{span}\{q, \bar p_1,\ldots,\bar p_n\}$ for some $q \in V$ that is orthogonal to each $\bar p_i$ (note that $q=0$ is allowed and corresponds to the case where $\bar{p}$ can be represented as a linear combination of $\bar{p}_i$). Since $\bar p \in V$, we can write
\begin{align*}
    \bar p = \alpha q + \sum_i \alpha_i \bar p_i
\end{align*}
for some scalars $\alpha,\alpha_1,\ldots,\alpha_n$. By the martingale property, $\E{\bar p}{\bar p_i} = \bar p_i$, and so $\E{(\bar p-\bar p_i)\cdot p_i}=0$. That is, $\bar p - \bar p_i$ is orthogonal to $\bar p_i$. Hence $\alpha_i = 1$, and 
\begin{align*}
    \bar p = \alpha q + \sum_i \bar p_i.
\end{align*}
Since $\{q, \bar p_1,\ldots,\bar p_n\}$ are orthogonal,
\begin{align*}
    \E{\bar p^2} = \alpha^2 \E{q^2} + \sum_i \E{\bar p_i^2},
\end{align*}
and in particular $\E{\bar p^2} \geq \sum_i \E{\bar p_i^2}$.

\end{proof}
Note that we used the assumption that the structure is private private only inasmuch as it implies that posteriors {induced by different signals}
are uncorrelated.

\section{Non-Uniqueness of Optimal Privacy-Preserving Recommendation for Non-Binary States}\label{app_non_uniqueness}

{Theorem~\ref{thm:disclosure} shows that, when the state is binary, there is a dominant privacy-preserving recommendation $s_2^\star$ for each $s_1$. This $s_2^\star$ dominates any other $s_2$ independent of $s_1$, and so all optimal privacy-preserving recommendations are Blackwell-equivalent to  $s_2^\star$. In this section, we show that there may be non-equivalent optimal recommendations $s_2$ for non-binary states. In particular, a dominant recommendation may fail to exist.} 

Consider the case of $\Omega=\{0,1,2\}$ where $\omega\in \Omega$ is distributed according to the prior $p=\big(\nicefrac{1}{4},\nicefrac{1}{2},\nicefrac{1}{4}\big)$. The signal $s_1$ is binary: if 
$\omega=2$ then $s_1=1$, if $\omega=0$ then $s_1=0$, and   if $\omega=1$ then $s_1\in\{0,1\}$ equally likely. The induced beliefs  $p(s_1)$ are equal to either $\big(\nicefrac{1}{2},\nicefrac{1}{2},0\big)$ or $\big(0,\nicefrac{1}{2},\nicefrac{1}{2}\big)$, each with probability $\nicefrac{1}{2}$.

 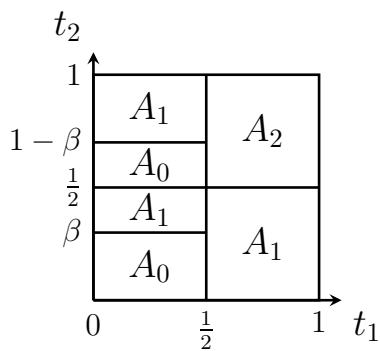
\begin{figure}[h]
\begin{center}
\begin{tikzpicture}[scale=0.3, line width=1pt]
\draw (0,0) -- (10,0)--(10,10)--(0,10)--(0,0);
\draw (5,0) -- (5,10);
\draw (5,5) -- (10,5);
\draw (0,3) -- (5,3);
\draw (0,7) -- (5,7);
\draw (0,5) -- (5,5);
\node at (7.5,2.5) {\large $A_1$};
\node at (7.5,7.5) {\large $A_2$};
\node at (2.5,1.5) {\large $A_0$};
\node at (2.5,6) {\large $A_0$};
\node at (2.5,8.5) {\large $A_1$};
\node at (2.5,4) {\large $A_1$};

\draw [->,>=stealth] (0,0) -- (11,0);
\draw [->,>=stealth] (0,0) -- (0,11);
\node[below right] at (11,0) {\large $t_1$};
\node[above left] at (0,11) {\large $t_2$}; 

\node[below] at (5,0) {$\frac{1}{2}$};
\node[below] at (10,0) {$1$};
\node[left] at (0,5) {$\frac{1}{2}$};
\node[left] at (0,3) {$\beta$};
\node[left] at (0,7) {$1-\beta$};
\node[left] at (0,10) {$1$};
\node[below] at (0,-0.1) {$0$};
\end{tikzpicture}

\end{center}
\caption{In the private private information structure associated with the partition $(A_0,A_1,A_2)$, the signal $t_1$ induces the same distribution of posteriors as  $s_1$. For any parameter $\beta\in [0,\nicefrac{1}{2}]$, this partition is a partition of uniqueness and, hence,  we get  a one-parametric family of non-equivalent optimal privacy-preserving recommendations given by  the signal $t_2$. 
\label{fig_nonunique}
}
\end{figure}

To construct an optimal privacy-preserving recommendation $s_2$ we first build an auxiliary private private information structure $(\omega,t_1,t_2)$, associated with the partition of $[0,1]^2$ into three sets $A_0$, $A_1$, and $A_2$ depending on a parameter $\beta\in[0,\nicefrac{1}{2}]$, as depicted in Figure~\ref{fig_nonunique}.
The pair of signals $(t_1,t_2)$ is uniformly distributed on $[0,1]^2$ and the state $\omega$ equals $k$ whenever the pair of signals belongs to $A_k$. Since the area of $A_1$ is twice the area of $A_0$ and $A_2$, and since the latter two areas are equal,  $\omega$ has the right distribution $p=\big(\nicefrac{1}{4},\nicefrac{1}{2},\nicefrac{1}{4}\big)$. 

Let us check that the signal $t_1$ is equivalent to $s_1$, i.e., it induces the same posterior distribution. Indeed, if the realization of $t_1$ belongs to $[0,\nicefrac{1}{2}]$, half of each vertical slice of the square is covered by $A_0$ and half by $A_1$, and so the induced posterior is $p(t_1)=\big(\nicefrac{1}{2},\nicefrac{1}{2},0\big)$ with probability $\nicefrac{1}{2}$. Similarly, for $t_1\in [\nicefrac{1}{2},1]$, we get $p(t_1)=\big(0,\nicefrac{1}{2},\nicefrac{1}{2}\big)$ also with probability $\nicefrac{1}{2}$.

Let us check that for different values of $\beta$ we obtain non-equivalent recommendations. For this purpose, we compute the distribution of posteriors induced by $t_2$. Note that $t_2$ is equivalent to a signal $s_2$ taking four different values corresponding to different pairs of sets $(A_i,A_j)$ intersected by the horizontal slice of the square. We get the following distribution of posteriors: 
\begin{align*}
p(s_2)=\left\{\begin{array}{cc}
\big(0,\nicefrac{1}{2},\nicefrac{1}{2}\big) & \mbox{with probability $\beta$}\\
\big(\nicefrac{1}{2},0,\nicefrac{1}{2}\big) & \mbox{with probability $\nicefrac{1}{2}-\beta$}\\
\big(0,1,0\big) & \mbox{with probability $\nicefrac{1}{2}-\beta$}\\
\big(\nicefrac{1}{2},\nicefrac{1}{2},0\big) & \mbox{with probability $\beta$}
\end{array}\right. 
\end{align*}
For different values of $\beta$ we get  different distributions, i.e., the constructed recommendations are  not equivalent.

It remains to show that for any value of $\beta$, the signal $s_2$ is an optimal privacy-preserving recommendation. To this end, we check that the partition $(A_0,A_1,A_2)$ is a partition of uniqueness (as defined in Appendix~\ref{app:Pareto_via_uniqueness}). Therefore, by Theorem~\ref{thm:pareto_appendix}, the information structure $(\omega,t_1,t_2)$ is Blackwell-Pareto optimal. Thus  $t_2$  is an optimal privacy-preserving recommendation, and so is $s_2$ as it is equivalent to $t_2$. 

To show that $(A_0,A_1,A_2)$ is a partition of uniqueness,  we rely on the following elementary but useful general observation: if in a partition $(A_0,\ldots, A_{m-1})$ of $[0,1]^n$ all sets except for possibly one are sets of uniqueness, then the partition itself is a partition of uniqueness. In our example, the set $A_2$ is upward-closed and hence is a set of uniqueness by Theorem~\ref{th_Lorentz}. The set $A_0$ is a rearrangement of an upward-closed set (since it can be made upward-closed via a measure-preserving reparametrization of the axes) and so is a set of uniqueness by the same theorem. Thus the partition $(A_0,A_1,A_2)$ is a partition of uniqueness and $s_2$ is an optimal privacy-preserving recommendation for any value of~$\beta$.

The partition $(A_0,A_1,A_2)$ provides an interesting example of the fact that  a partition of uniqueness is not necessarily composed of sets of uniqueness. Indeed, for $\beta\ne 0$, the set $A_1$ is not a set of uniqueness as it has the same marginals as the set obtained by the reflection of $A_1$ with respect to the vertical line $t_1=\nicefrac{1}{2}$.

\subsection{Representing Private Private Signals for Binary $\omega$ as Sets}\label{app_explicit_representation}

To simplify notation, in this section, we consider the case of $n=2$ agents and a binary state $\omega\in \Omega=\{0,1\}$. Nevertheless, the same ideas apply more generally to finitely many agents and possible values of the state. By Proposition~\ref{prop:associated}, any private private information structure $\mathcal{I}$ is equivalent to a structure  associated with some set $A\subseteq [0,1]^2$, which we denote by $\mathcal{I}_A=(\omega,s_1,s_2)$.
In this section, we show how to construct $\mathcal{I}_A$ given~$\mathcal{I}$. We begin with the case where $\mathcal{I}$ is uninformative and describe $\mathcal{I}_A$ for any prior $p=\Pr{\omega=1}$. Relying on this construction, we then describe how to construct $\mathcal{I}_A$ for any $\mathcal{I}$ with a finite number of possible signal values.

Recall that in $\mathcal{I}_A$, the signals $(s_1,s_2)$ are uniformly distributed on $[0,1]^2$, the state is $\omega=\one_A(s_1,s_2)$, and $A$ is some measurable subset of $[0,1]^2$ with Lebesgue measure $\lambda(A)=p$ so that $p=\Pr{\omega=1}$.
Recall that the distribution of posteriors induced by $\mathcal{I}_A$ can be computed as follows: the conditional probability of the high state given that agent $i$ receives a signal $s_i=t$ is exactly $\alpha_i^A(t)$, the one-dimensional Lebesgue measure of the cross-section $\{(y_1,y_2)\in A\,:\, y_i=t\}$. In other words, $\alpha_i^A(s_i)$ is $i$'s posterior corresponding to $s_i$ and the induced distribution of posteriors $\mu_i$ is the image of the uniform distribution under the map $\alpha_i^A$, i.e., $\mu_i([0,t])$ equals the Lebesgue measure of $\{x_i\in [0,1]\,:\, \alpha_i^A(x_i)\leq t\}$.

\begin{example}[Non-informative signals]\label{ex_noninformative}
Consider a private private information structure $\mathcal{I}$, where both agents receive completely uninformative signals, i.e., the induced posteriors are equal to the prior $p$.

To find an equivalent structure $\mathcal{I}_A$, we need to construct a set $A=A_p\subseteq[0,1]^2$ such that the Lebesgue measure of all its projections equals~$p$. To this end, let $Y$ be any subset of $[0,1]$ with measure~$p$ (e.g., $[0,p]$), and let
\begin{align*}
A = \left\{(x_1,x_2) \in [0,1]^2\,:\, \left\lfloor x_1+x_2 \right\rfloor \in Y\right\},
\end{align*}
where $\lfloor x\rfloor$ is the fractional part of $x \in \R$. It is easy to see that $A$ indeed has the desired property.

\end{example}

It turns out that the construction of an information structure $\mathcal{I}_{A_p}$ representing completely uninformative signals can be used to find a representation for any information structure with a finite number of possible signal values.
\begin{example}[Arbitrary finite number of signal values] Let $\mathcal{I}=(\omega, s_1,s_2)$ be a private private information structure with $n=2$ agents and finite signal spaces $S_1$ and $S_2$. Our goal is to construct a set $A\subseteq [0,1]^2$ such that the structure $\mathcal{I}_A$ associated with $A$ is equivalent to $\mathcal{I}$.

 For each agent $i\in \{1,2\}$, consider a disjoint partition of $[0,1]$ into intervals $A_{s_i}$, $s_i\in S_i$, so that the length of each $A_{s_i}$ coincides with the probability that the signal $s_i\in S_i$ is sent under $\mathcal{I}$. 
 Let $q(s_1,s_2)\in [0,1]$ be the conditional probability of $\{\omega=1\}$ given signals $(s_1,s_2)$. 
 
 Recall that, in Example~\ref{ex_noninformative}, we constructed a set $A_p\subseteq [0,1]^2$ such that its projection to each of the coordinates has a constant density $p$.
  Now we construct $A$ by pasting the appropriately rescaled copy of $A_{q(s_1,s_2)}$ into each rectangle $A_{s_1}\times A_{s_2}$. Denote by $T_{[a,b]\times [c,d]}$ an affine map $\R^2\to \R^2$ that identifies $[0,1]^2$ with $[a,b]\times [c,d]$:
  $$T_{[a,b]\times [c,d]}(x_1,x_2)=\big(a+(b-a)x_1, \  c+ (d-c)x_2\big).$$
  
  We define $A$ as the following disjoint union:
  $$A=\bigsqcup_{s_1\in S_1,\, s_2\in s_2} T_{A_{s_1}\times A_{s_2}}\Big(A_{q(s_1,s_2)}\Big).$$
  
  Let $s_i'\in [0,1]$ be a signal received by an agent $i$ in $\mathcal{I}_A$. The signal $s_i'$ falls into $A_{s_i}$ with the same probability that $i$ receives the signal $s_i$ in $\mathcal{I}$. By  construction, the conditional probability of $\{\omega=1\}$ given $s_i'$ is constant over each interval $A_{s_i}$ and coincides with the posterior $p_i(s_i)$ that $i$ gets under $\mathcal{I}$. We conclude that $\mathcal{I}$ and $\mathcal{I}_A$ induce the same distribution of posteriors and so are equivalent.
  
 \end{example}

\end{document}